\documentclass[a4paper,12pt] {scrartcl}
\usepackage[paperwidth=21cm,left=2.5cm,right=2.5cm]{geometry}
\usepackage[intlimits,sumlimits,namelimits,fleqn]{amsmath}
\usepackage{amsfonts,latexsym}

\usepackage{amscd}
\usepackage[arrow, matrix, curve]{xy}
\usepackage{amsthm}
\usepackage{amssymb} 
\usepackage{nomencl} %Nomenklatur 

\usepackage{bbm}
\usepackage{cancel}
\usepackage{slashed}
\usepackage{yfonts}
\usepackage{mathrsfs}
\usepackage{dsfont}

\usepackage{graphicx} 
\usepackage[dvips]{color}
\usepackage{pst-all}
\usepackage{pst-slpe}
\usepackage{pst-blur}

\usepackage{paralist}
\usepackage{makeidx}
\usepackage{multicol}

\newcommand{\quot}[1]{``#1''}

\newcommand{\E}[1]{\ensuremath{\mathrm{E}_{#1}}} % e.g. \E{8}

\newcommand{\SO}[1]{\ensuremath{\mathrm{SO}(#1)}}
\newcommand{\SU}[1]{\ensuremath{\mathrm{SU}(#1)}}
\newcommand{\SL}[2]{{\ensuremath{\mathrm{SL}({#1},\mathds{#2})}}}
\newcommand{\PSL}[2]{{\ensuremath{\mathrm{PSL}({#1},\mathds{#2})}}}
\newcommand{\GL}[2]{\ensuremath{\mathrm{GL}({#1},\mathds{#2})}}

\newcommand{\ex}[1]{\ensuremath{\mathrm{e}^{#1}}}

\newcommand{\hypmes}[1]{\ensuremath{\frac{d^2 #1}{{#1_2}^2}}\,}

\newcommand{\f}[2]{\ensuremath{\ln\left(\frac{T_2}{#1 #2}\left|\eta\left(\frac{T}{#1 #2}\right)\right|^4\,\frac{#1 U_2}{#2}\left|\eta\left(\frac{#1U}{#2}\right)\right|^4\right)}}

\newcommand{\sumfrac}[1]{\ensuremath{\frac{u_{#1}}{v_{#1}}}}
\newcommand{\sumfraca}[1]{\ensuremath{\frac{u_{#1}}{v_{#1}}}}
\newcommand{\sumfracb}[1]{\ensuremath{\frac{v_{#1}}{u_{#1}}}}
\newcommand{\sumset}[5]{\ensuremath{\hspace*{0.2em}^{#1}_{#2}\hspace*{-0.1em}{#5}^{#4}_{#3}}}

\newcommand{\I}{\mathrm{i}}

\DeclareMathOperator{\lcm}{lcm}

\DeclareMathOperator{\diag}{diag}

\numberwithin{equation}{section}

  \newtheoremstyle{style1}% name
  {5pt}%      Space above
  {5pt}%      Space below
  {}%         Body font
  {}%         Indent amount (empty = no indent, \parindent = para indent)
  {\normalfont\bfseries}% Thm head font
  {}%        Punctuation after thm head
  {.5em}%     Space after thm head: " " = normal interword space;
  {}%         Thm head spec (can be left empty, meaning `normal')

  \newtheorem{thm}{Theorem}[section]
  \newtheorem{lemdef}[thm]{Lemma and Definition}
  \newtheorem{lem}[thm]{Lemma}
  \newtheorem{cor}[thm]{Corollary}
  \newtheorem{defini}[thm]{Definition}
  \newtheorem*{claim}{Claim}
  \newtheorem{prop}[thm]{Proposition}
 
 \newtheorem*{remarks}{\textbf{Remarks}}
 \newtheorem*{remark}{\textbf{Remark}}

\makeindex

\begin{document}
%\maketitle
\thispagestyle{empty}

\begin{center}
{\Large\bf 
The Computation of One-Loop Heterotic String Threshold Corrections for General Orbifold Models with Discrete Wilson Lines
}

\vspace{1cm}

\textbf{
Michael A. Klaput\footnote[1]{Email: \texttt{Michael.Klaput@physics.ox.ac.uk}}$^{,\,a,b}$,
Christian Paleani\footnote[2]{Email: \texttt{christian.paleani@math.lmu.de}}$^{,\,c}$
}
\\[5mm]
\textit{\small
$^{a}$%
Rudolf Peierls Center for Theoretical Physics\\ 1 Keble Road, Oxford OX1 3NP, UK\\[2mm]
$^{b}$%
St. John's College\\ St. Giles, Oxford OX1 3JP, UK\\[2mm]
%$^{c}$%
%Physik-Department T30, Technische Universit\"at M\"unchen, \\
%James-Franck-Stra\ss e, 85748 Garching, Germany\\[2mm]
$^{c}$%
Mathematisches Institut, Ludwig--Maximilians-Universit\"at M\"unchen, \\
Theresienstr. 39, D-80333 M\"unchen, Germany
}
\end{center}
\vspace{0.3cm}

%------------------------------------------------------------------------------------------------------
%                                              Abstract
%------------------------------------------------------------------------------------------------------
\begin{abstract}
We calculate the moduli dependent part of string one-loop threshold corrections to gauge couplings for the heterotic string theory compactified on abelian toroidal orbifolds, allowing for arbitrary discrete Wilson lines. We show that the knowledge of threshold corrections for any such compactification is equivalent to solving a class of integrals. We solve a sub-class of these integrals and show how any model can be mapped onto this class by fractional linear transformations of its fixed plane moduli. Modular symmetries of the final expression are discussed.
\end{abstract}

%------------------------------------------------------------------------------------------------------
%                                              Chapters
%------------------------------------------------------------------------------------------------------
\tableofcontents

\pagenumbering{arabic}

\section{Introduction}

In the past decades, string theory emerged as one of the most promising candidates for a unified description of nature---providing a quantum theory of all known forces including gravity. However, it still remains an open problem to find a vacuum of the theory which consistently describes the well-known standard model of elementary particle physics and enables us to predict physics beyond the standard model---hopefully leading to a connection of string theory to experiment. 

The setting for our work will be the heterotic string, for which numerous models are known with the minimal supersymmetric standard model (MSSM) spectrum and gauge group. As compactification spaces we will be looking at abelian toroidal orbifolds. These have been investigated thoroughly since the early works \cite{stringsonorbifolds1} \cite{stringsonorbifolds2} and have attracted several interest recently, see for example \cite{ratz1} \cite{ratz2} \cite{ratz3}.

Besides the difficulty of obtaining the correct massless spectrum and standard model gauge group, there exist several other problems. The most severe of these possibly being the discrepancy between the expected energy scale for the unification of all forces from experimental data (the GUT scale) and the scale suggested by string theory (the string scale) at tree level. These two scales differ by an order of magnitude.

There are various possibilities of how to resolve this problem: new particles at intermediate mass scales, non-standard affine levels and string one-loop effects. See \cite{Dienes} for a review.

String one-loop effects can alter the unification scale through threshold corrections to the gauge couplings. It is essential that, for the reconciliation of GUT scale and string scale, it suffices to know the part of these threshold corrections which depends on the moduli of the compactifying space. This part can be obtained from the general formula determined in \cite{Kaplunovsky}, valid for any vacuum of the heterotic string. In ref. \cite{Dixon_et_al2} this formula has been applied to abelian toroidal orbifold compactifications. It was possible to evaluate formulas for a special sub-class of the orbifold compactifications, in which the string modes contributing to the threshold corrections are localized on a complete two dimensional sub-torus in the absence of discrete Wilson lines. This result was further generalized in \cite{Stieberger4} where the moduli dependent part of the threshold corrections has been evaluated for special orbifold geometries, in which the string states contributing to the threshold corrections are not localized in a complete two-dimensional sub-torus. However, the inclusion of discrete Wilson lines has remained a problem ever since. Possible effects of non-vanishing discrete Wilson lines on modular symmetries have been first considered in \cite{ErlerSpalinski} \cite{ErlerTdual} \cite{BailinCoxeter}. There it was realized that discrete Wilson lines can break the usual $\PSL{2}{Z}$ symmetry. Symmetry groups of various orbifold models with discrete Wilson lines were obtained in \cite{LoveDuality} \cite{LoveWilson}.

However, the calculation of threshold corrections in the presence of discrete Wilson lines remained as an open problem. The recent activity in orbifold model building has renewed interest in solving this problem, since all promising models posses non-vanishing discrete Wilson lines. Therefore, investigation of the unification scale in these models requires the knowledge of the moduli dependent part of the gauge coupling threshold corrections. Furthermore, many aspects of low-energy phenomenology rely on $\PSL{2}{Z}$ as the group of modular symmetries \cite{lustshapere}. The knowledge of the group of modular symmetries in the presence of Wilson lines is, therefore, useful to examine certain low-energy phenomena in these models.

This has been the main motivation behind the work presented here. In the following, we shall explain why we think that this work has solved the task of determining the moduli dependent part of gauge threshold corrections in arbitrary abelian toroidal orbifold compactification of the heterotic string (allowing for arbitrary discrete Wilson lines) in full generality: We will develop a constructive prescription of how to calculate an analytic expression for an arbitrary model.

The organisation of the paper is as follows. The first part of Chapter \ref{boundary} introduces the physical language we will be using throughout the work. The second part investigates the action of modular transformations on conformal field boundary conditions on the world sheet torus. The third part develops the notion of closed, minimally closed and generating sets of boundary conditions. Furthermore, their structure is analysed. Chapter \ref{genset} applies the structure uncovered in chapter \ref{boundary} to show that the calculation of threshold corrections simplifies significantly since the partition functions to different boundary conditions can be related to each other. This is used to reformulate the problem as an integral over the partition function associated to one specific boundary condition, with the domain of integration being the fundamental domain of the symmetry group of this partition function. Next, the fixed plane condition for the quantum numbers is considered. It is shown that this is a system of linear Diophantine equations which always possesses solutions and can be characterised by four rational numbers $\alpha,\beta,\gamma,\delta\in\mathds{Q}$. These four numbers determine a class of integrals which is equivalent to the calculation of one-loop threshold corrections. Chapter \ref{chapfour} starts with the calculation of the integrals of the classes $\alpha=\beta=\gamma=1$ and $\delta\in\mathds{Z}$ and the symmetry of the result is investigated. Afterwards, it is proven that all other classes can be reduced to the classes $\alpha=\beta=\gamma=1$ and $\delta\in\mathds{Z}$ via a redefinition of the complex structure and K\"ahler structure moduli of the compactification space.
\section{On Boundary Conditions, Closed Sets and Generating Elements}\label{boundary}

This section will follow several aims. Firstly, we will give a brief introduction into the physical background of our analysis. Secondly, we will show that the set of elements of boundary conditions which contribute to the one-loop gauge threshold corrections admit a certain structure. Thirdly, we will proof the existence of a generating system for these elements.

The main result of this section will be theorems \ref{closed}, \ref{speqex} and \ref{part}. The first states that the set of boundary conditions which contributes to $\Delta_a$ is closed under $\SL{2}{Z}$.\footnote{Therefore, the sum over all associated partition functions is modular invariant.} The second states that there is a natural choice of a generating system of $\mathcal{O}$. The third states that the partition functions of all boundary conditions of a generating set coincide if at least one of them is invariant under $\Gamma'\subset \Gamma$ for some finite index subgroup $\Gamma'$ of $\Gamma$. If a transformation is regarded as an action on boundary conditions, it is regarded as an element of $\SL{2}{Z}$. The definition of the action of such a transformation accounts for the fact that there exists an associated modular transformation, hence, an element of $\Gamma$, on the partition function which is associated to this boundary condition. We should mention that it is possible to generalize these results. One can look at closed sets under other groups than $\SL{2}{Z}$ and assign functions to those sets with crucial relation \eqref{219}. These sets do not necessarily have to be boundary conditions and the function not to be partition functions. But since we want to compute one-loop gauge threshold corrections, we will not develop this further in the present work.

\subsection{Physical terminology}

The aim of this subsection is to outline the physical motivation for our computation. Furthermore, we would like to establish a precise language for the class of string theory models we will be discussing. The reader without a background in string theory will find definitions for the physical terms used throughout this work. However, the definitions given are designed to suit this work and will most probably not be of great practicality in most other contexts.

Our starting point is the heterotic string theory. We compactify down to four dimensions by using the geometry $\mathds{M}^4 \times O$, where $O$ is a toroidal orbifold defined by \cite{stringsonorbifolds1} \cite{stringsonorbifolds2}
\begin{equation}
O = \mathrm{T}^6 / \mathcal{P}\;.
\end{equation}
This means that all points $x\in\mathds{T}^6$ are identified which are related by a group element $\theta\in\mathcal{P}$ as\footnote{Here $x$ and $y$ are elements of the \emph{torus}. If we view $x,y\in\mathds{R}^6$, they have to coincide only modulo lattice translations.}
\begin{equation}
x \sim  y \quad\Longleftrightarrow\quad y= Q(\theta)\, x
\end{equation}
where $Q(\theta)$ is a representation of $\theta$ on the torus lattice. The group $\mathcal{P}$ is called the point group of the orbifold. We will choose an abelian group $\mathcal{P}=\mathds{Z}_N$ which has been the choice in the vast majority of constructed orbifold models until today.

Note, that in contrast to earlier work, we do not impose further assumptions. We do not require the torus lattice to be decomposable into $\mathrm{T}^4 \times \mathrm{T}^2$ or $\mathrm{T}^2\times\mathrm{T}^2\times\mathrm{T}^2$ as in \cite{Dixon_et_al2}. Furthermore, we allow for arbitrary discrete Wilson lines, which can be understood as the $\mathcal{P}$ action in the $\E{8} \times \E{8}'$ or $\SO{32}$ gauge bundle over $O$ \cite{stringsonorbifolds2}.

This is of particular interest in string phenomenology since \quot{switching on} discrete Wilson lines can be used to lower the rank of the gauge group smoothly and reduce the number of generations in a given model without Wilson lines \cite{Ibanez:1986tp}, yielding numerous interesting models with three generations and the standard model gauge group, see \cite{ratz1} \cite{ratz2} \cite{ratz3} for recent constructions.

For our purposes, the following definition will be most suitable

\begin{defini} An (abelian toroidal) orbifold model (of the heterotic string) $(\Lambda,Q,\{\vec{A}_i\})$ is given by specifying a six-dimensional torus lattice $\Lambda$, the generator $Q$ of a $\mathcal{P}=\mathds{Z}_N$ representation on the torus lattice and a set of rational numbers $\{A_i^I\}=:\vec{A}_i$ $i=1,\ldots,6$, $I=1,\ldots,16$.

We call the generating element of $\mathds{Z}_N$ the (orbifold) twist, $N$ the order of the twist and the six sixteen dimensional vectors $\vec{A}_i$ the (discrete) Wilson lines of the model specified in that way.

An orbifold model $(\Lambda,Q,\{\vec{A}_i\})$ where all $\vec{A}_i = 0$ is called an orbifold model without Wilson lines.
\end{defini}
\begin{remarks}$ $
\begin{itemize}
\item The words in parentheses will be often ommitted in favour of brevity.
\item It has to be stressed again, that this definition is designed to suit our purposes. To construct a sensible physical model one has to impose additional restricitons on the twist and the Wilson lines \cite{stringsonorbifolds2}. However, since our results will remain true even for unphysical choices of $(\Lambda,Q,\{\vec{A}_i\})$, we appeal to this minimal definition in order to avoid a loss of generality. Obviously, every physical model is included in this definition.
\item Throughout the article, we will be explicitly using $\E{8} \times \E{8}'$ heterotic theory for concreteness. However, all of our results are most easily reformulated to apply to the $\SO{32}$ heterotic theory as well by simply choosing the Wilson lines to take values in the $\SO{32}$ root lattice and replacing the term $\E{8} \times \E{8}'$ by $\SO{32}$ wherever used in the text.
\item If $\theta$ is the generating element of $\mathds{Z}_N$, then $Q(\theta)$ will usually be called twist as well.
\end{itemize}
\end{remarks}

A special role is played by singular loci in the orbifold. These occur if the point group is not acting freely on $\mathrm{T}^6$. 

There are two concepts we will need. Firstly, the concept of fixed points in the orbifold. Secondly the notion of a fixed plane. The fixed points are given by
\begin{equation}
	\{x\in \mathds{R}^6 \;|\; Q(\theta)\,x\equiv x\,(\mathrm{mod}\,\Lambda)\,\}_{k\in\mathds{N}}\,.
\end{equation}
This means the set of all points which are equivalent under the action of the twist modulo some lattice translation. The fixed planes are relevant, since their intersection with the torus lattice is given by all states which are invariant under some power of the twist. They correspond to all states which are present after orbifolding and given by
\begin{equation}
	\{x\in \mathds{R}^6 \;|\; Q(\theta)^k\,x= x\,\}_{k\in\mathds{N}}\,.
\end{equation}
It can be shown \cite{Dixon_et_al2}, that this locus has either dimension $0$, $2$ or $6$ and accordingly defines an invariant point, an invariant plane (fixed plane) or acts trivially on the torus lattice. 
Fixed planes are given by the solutions of the equation
\begin{equation}
	Q(\theta)^k\,x=x
\end{equation}
for $k\in\mathds{N},\, k\not = 1$ and $Q^k \not = \mathds{1}$ \footnote{We will usually say that $\theta^k$ fixes a plane and call the plane the $\theta^k$ fixed plane.}. 

For our task at hand, the calculation of the moduli dependent part of a certain one-loop amplitude, it will be necessary to consider fields on a world-sheet torus of the string. The Hilbert space of fields decomposes into sectors of different boundary conditions along the fundamental cycles of the world-sheet torus.

\begin{defini} A field $\phi$ on the world-sheet torus $\mathrm{T}^2$ (parametrized as $\sigma_1 + \tau\,\sigma_2$, $\sigma_{1,2}\in\mathds{R}$, $\tau\in\mathds{C}$, $\sigma_{1,2} \cong  \sigma_{1,2} + 1$) is said to carry boundary conditions $(g,h)\in\mathcal{P}$ if
\begin{equation}\label{boundaryconds}
\begin{split}
\phi(\sigma_1 + 1,\sigma_2) &= g\cdot\phi(\sigma_1,\sigma_2)\\
\phi(\sigma_1 ,\sigma_2 + 1) &= h\cdot\phi(\sigma_1,\sigma_2)
\end{split}
\end{equation}
\end{defini}
\begin{remark} \quot{$\cdot$} denotes some action of $\mathcal{P}$ on the fields. In the case of $\mathcal{P}=\mathds{Z_N}$ this is given as follows: We can always choose a complex basis $\{z_i\}_{i=1,2,3}$ for our six real fields embedding the string into the orbifold. In these coordinates, a general irreducible representation of the twist is given by
$\diag(\theta_1,\theta_2,\theta_3)$
with $\theta_i\in\mathds{C}$, $\theta_i^N = 1$ and $|\theta_i|=1$.
\end{remark}

The set of boundary conditions is classified as follows:
\begin{defini} We call
\begin{equation}
\{(g,h)\in\mathcal{P}\times\mathcal{P}\,|\,g \text{ and } h \; \text{fix the same points}\}
\end{equation}
the $\mathcal{N}=1$ sector of the orbifold model,
\begin{equation}
\{(g,h)\in\mathcal{P}\times\mathcal{P}\,|\,g \text{ and } h \; \text{fix the same plane}\}
\end{equation}
the $\mathcal{N}=2$ sector of the orbifold model and
\begin{equation}
\{(g,h)\in\mathcal{P}\times\mathcal{P}\,|\,g \text{ and } h \; \text{leave the whole orbifold invariant}\}
\end{equation}
the $\mathcal{N}=4$ sector of the orbifold model.
\end{defini}
Let us define the main object of interest in the present work.
\begin{defini} We call $\Delta_a$ the (moduli dependent part of the) threshold correction, with
\begin{equation}
\label{objectofinterest}
	\Delta_a=\int_{R_\Gamma} \hypmes{\tau}\,\sum_{(g,h)\in\mathcal{O}}\,b_a^{(g,h)} \tau_2\,Z^{\text{1-\text{loop}}}_{(g,h)}(\tau)-R\,.
\end{equation}
Here $\mathcal{O}$ denotes the $\mathcal{N}=2$ sector of the orbifold model, $Z^{\text{1-\text{loop}}}_{(g,h)}(\tau)$ the partition function associated to the boundary conditions $(g,h)$ and $R_\Gamma$ the fundamental domain of the group $\Gamma$ in the upper half complex plane $\mathds{H}^+$.
\end{defini}

We will examine this expression in more detail in section \ref{genset}. For the discussion in the present section, it will be of importance that only the $\mathcal{N}=2$ sector contributes to the threshold correction. As we will see shortly, boundary conditions of this sector exhibit a useful structure with respect to modular transformations on the world-sheet torus.

\subsection{Transformation of Boundary Conditions}
Boundary conditions specify the transformation properties of the fields under translation of the world-sheet torus-lattice by fundamental cycles. It will be of special interest, how these properties behave under the action of a modular transformation on the torus. There exist different conventions and formulas in the literature of how boundary conditions on the world-sheet transform under modular transformations. To keep our work self-contained, we start with an examination how boundary conditions map onto each other under the action of $\SL{2}{Z}$. This is to avoid any confusion due to the different conventions used throughout the literature.

At the beginning we would like to state the result of this examination: The generators of the modular group, namely $S:\tau\mapsto-1/\tau$ and $T:\tau\mapsto\tau+1$ change the boundary conditions to 
\begin{equation}
\label{49}
\begin{split}
 &(\theta^k,\theta^l)\stackrel{S}{\to}(\theta^l,\theta^{N-k})\quad \text{and}\\
 &(\theta^k,\theta^l)\stackrel{T}{\to}(\theta^k,\theta^{l-k})\,,\quad\text{respectively.}
\end{split}
\end{equation}
Where $N$ denotes the order of the twist ($\theta^N=1$). Next, we will show why this is true.

Let us look at a modular transformation on the world-sheet torus. The torus is defined by a lattice $\Lambda$ in complex space via $T^2\cong \mathds{C}/\Lambda$. We will describe it by a single complex modular parameter $\tau$. The resulting object becomes a one dimensional complex manifold by choosing an atlas and a complex structure on it. Let us denote the complex variable on the torus (in local coordinates) by $\nu$. This $\nu$ is then parameterized by two real variables $\sigma_1$ and $\sigma_2$,
\begin{equation}
\label{52}
 \nu=\sigma_1+\sigma_2\,\tau\,.
\end{equation}
To see how the boundary conditions change under modular transformations one has to recall that two lattices $\Lambda$ and $\Lambda'$ define the same complex structure on a torus if and only if there exists a complex number $\xi\in\mathds{C}$ such that $\Lambda=\xi\,\Lambda'$. This condition is equivalent to the existence of a matrix $V\in\SL{2}{Z}$ which links the complex parameters $\tau \in \Lambda$, $\tau' \in \Lambda'$, by a modular transformation
\begin{equation}
\label{50}
 \tau'=\frac{a\,\tau+b}{c\,\tau+d}\quad\text{for}\quad V=\begin{pmatrix} a&b \\ c&d \end{pmatrix}\,.
\end{equation}
By inserting equation \eqref{50} into \eqref{52}, factoring out $c\,\tau+d$, re-expressing $\sigma_1$ and $\sigma_2$, and exploiting the definition of equivalence of complex structures, one can infer that this transformation corresponds to

\footnote{In the convention of \cite{GSW} equation \eqref{51} takes the form 
$(\sigma_1,\,\sigma_2)\longmapsto(\sigma'_1,\,\sigma'_2)=(a\,\sigma_1+b\,\sigma_2,c\,\sigma_1+d\,\,\sigma_2)\,.$ In our convention this would correspond to the modular transformation by the inverse matrix and, hence, this has no influence on the calculation of the thresholds.}
\begin{equation}
\label{51}
 (\sigma_1,\,\sigma_2)\longmapsto(\sigma'_1,\,\sigma'_2)=(d\,\sigma_1+b\,\sigma_2,c\,\sigma_1+a\,\,\sigma_2)\,.
\end{equation}
For fields $\phi(\sigma_1,\,\sigma_2)$ with boundary conditions \eqref{boundaryconds} it follows that a $V$ transformation on $\tau$ of the form \eqref{50} results in
\[
  \phi'(\sigma'_1+1,\,\sigma'_2)=\phi'(d\,\sigma_1+b\,\sigma_2+1\,,\,c\,\sigma_1+a\,\sigma_2)
\]
\[
 \,\qquad\qquad\qquad             =\phi'\left((d(\sigma_1+a)+b(\sigma_2-c)\,,\,c(\sigma_1+a)+a(\sigma_2-c)\right)
\]
\[
 \,\qquad\qquad\qquad             =V\,\phi(\sigma_1+a\,,\,\sigma_2-c)
\]
\[
 \,\qquad\qquad\qquad             =g^a\,h^{-c}\,V\,\phi(\sigma_1,\,\sigma_2)
\]
\[
 \,\qquad\qquad\qquad             =g^a\,h^{-c}\,\phi'(\sigma'_1,\,\sigma'_2)
\]
\begin{equation}
\label{47}
 \,\qquad\qquad\qquad             =g'\,\phi'(\sigma'_1,\,\sigma'_2)\,.
\end{equation}
And consequently $g'=g^a\,h^{-c}$. Here we used $[g,h]=0$. In analogy it follows that
\[
 h'=g^{-b}\,h^d\,.
\]
Therefore, the change of the boundary conditions under a modular transformation reads
\begin{equation}
\label{48}
 V:\quad (g,\,h)\longmapsto (g',h')=(g^a\,h^{-c},\,g^{-b}\,h^d)\,.
\end{equation}
If one applies \eqref{48} to twists, equation \eqref{49} is proven.

\subsection{Modular Subgroups, (Minimally) Closed Sets, Generating Sets and Partition Functions}

This section will provide necessary tools for the computation of one-loop gauge threshold corrections for general orbifold models with arbitrary discrete Wilson lines. In section \ref{genset} we will apply theorem \ref{part} to simplify the general expression for $\Delta_a$. Having established the notion of a generating set of boundary conditions in definition \ref{gensetdef}, this theorem tells us that all partition functions associated to elements of one generating set actually coincide if the generating set is defined with respect to the (modular) symmetry group $\Gamma' \subset \Gamma$ of one of these elements\footnote{Since the partition functions coincide, the symmetry group is the same for all these elements}. The examination whether generating sets of boundary conditions are unique is not important to our scope, since an immediate consequence of their definition is that the threshold correction does not depend on possible ambiguities. All this will be utilised together with lemma \ref{specialelement} to justify our ansatz to compute $\Delta_a$ later in section \ref{genset}.

We will start this issue by defining the notion of a group right action on a finite set. Afterwards we will define a binary operation which accounts for the transformation behaviour of the world-sheet torus partition function associated to the boundary conditions $(g,h)$ under modular transformations. This will give a group right action of modular transformations, viewed as elements of $\SL{2}{Z}$, on the (closed) set of boundary conditions. Then we will define closed and minimally closed sets and show under some assumptions that the $\mathcal{N}=2$ sector of the orbifold model $\mathcal{O}$ which contributes to $\Delta_a$ is closed. Having established their definition, the structure of closed sets will be examined. Afterwards, we will define generating sets and prove their existence within closed sets. At the end we will observe that it is possible to choose a generating set $\mathcal{O}_0$ of $\mathcal{O}$ such that the torus partition functions of $\mathcal{O}_0$ coincide.

In the proofs concerning boundary conditions of conformal fields on a world sheet torus we have taken into account that in our orbifold models we always have a commutative point group $\mathcal{P}=\mathds{Z}_N$. This has been done in order to simplify and shorten the proofs. However, many of the presented properties of sets of boundary conditions do not rely on the commutativity of the point group and the whole discussion can be generalised to the non-abelian case (using the fact, that in an element of boundary conditions $(g,h)$ physical consistency requires $[g,h]=0$). However, for the sake of brevity we will only consider abelian point groups since this is our concern in the following sections. Though, we tried to keep the notation as general as possible.

For the remainder of this section, if not stated differently, let $\mathcal{O}$ be a finite set, $(G,\cdot)$ be a group, $(G',\cdot)$ be a subgroup of $(G,\cdot)$ and $\ast$ be a right action of $(G,\cdot)$ on $\mathcal{O}$, i.e.
\begin{align}
\label{groupaction}
	&\forall\, x\in\mathcal{O}\,:\,\forall\, V_1,V_2 \in G\,:\,x\ast \left(V_1\cdot V_2\right)=\left(x\ast V_1\right) \ast V_2\\
	&\forall\,x\in\mathcal{O}\,:\, x\ast \mathds{1} =x
\end{align}
Now let us introduce a binary operation which accounts for the transformation of partition functions under modular transformations.
\begin{defini}\label{ast}
 Let $V\in \SL{2}{Z}$ and let $(g,h)\in\mathcal{P}^2$, with $\mathcal{P}$ denoting the point-group of an orbifold model. Furthermore, let $V$ be parameterised as
\begin{equation}
\label{ast1}
 V=\begin{pmatrix} a&b\\c&d \end{pmatrix}\,.
\end{equation}
 Then we define a binary operation $\mathcal{P}^2\times\SL{2}{Z}\to\mathcal{P}^2,\,\left((g,h),V\right)\longmapsto (g,h)\ast V$ via
\begin{equation}
\label{ast2}
 (g,h)\ast V:=(g^a\,h^{-c},g^{-b}\,h^d)\,.
\end{equation}
\end{defini}
This is motivated by \eqref{48} and further justified by the fact that
\begin{equation}
\label{219}
 V\cdot\tau_2 Z^{\text{1-loop}}_{(g,h)}(\tau)\equiv (\mathrm{Im}\, V \tau) Z^{\text{1-loop}}_{(g,h)}(V\tau)=\tau_2Z^{\text{1-loop}}_{(g,h)\ast V}(\tau)\,.
\end{equation}
This operation defines a group right action of $\SL{2}{Z}$ on the set of boundary conditions (and, hence, on the $\mathcal{N}=2$ sector $\mathcal{O}$), as stated in
\begin{prop}\label{inverse}
 Let $\mathcal{P}^2$ be the set of boundary conditions of an orbifold model and $\ast$ be defined in definition \ref{ast}. Then $\ast$ defines a group right action of $\SL{2}{Z}$ on $\mathcal{P}^2$.
\end{prop}
\begin{proof}
	Let $x\in\mathcal{P}^2$. Then it has the form $x=(g,h)$. Obviously, it holds $(g,h)\ast \mathds{1}=(g,h)$. Now, we will show
	\begin{equation}
		\left( (g,h)\ast V_1 \right)\ast V_2=(g,h)\ast(V_1\cdot V_2)\,,
	\end{equation}
	where the dot denotes the ordinary matrix-multiplication.
	
	Let $V_1,V_2\in\SL{2}{Z}$, $(g,h)\in\mathcal{P}^2$ and let $V_1$, $V_2$ be parameterised as
	\begin{equation}
	\label{inverse2}
		V_1=\begin{pmatrix} a_1&b_1\\c_1&d_1 \end{pmatrix} \quad \text{and} \quad V_2=\begin{pmatrix} a_2&b_2\\c_2&d_2 \end{pmatrix}\,,\quad \text{respectively.}
	\end{equation}
	Then matrix multiplication yields
	\begin{equation}
	\label{inverse3}
		V_1\cdot V_2=\begin{pmatrix} a_1&b_1\\c_1&d_1 \end{pmatrix}\cdot \begin{pmatrix} a_2&b_2\\c_2&d_2 \end{pmatrix}=\begin{pmatrix} a_1\,a_2+b_1\,c_2&a_1\,b_2+b_1\,d_2 \\c_1\,a_2+d_1\,c_2&c_1\,b_2+d_1\,d_2 \end{pmatrix}\,.
	\end{equation}
	Therefore, the action of $V_1\cdot V_2$ on $(g,h)$ reads
	\begin{equation}
	\label{inverse4}
		(g,h)\ast (V_1\cdot V_2)=\left(g^{a_1 a_2+b_1 c_2}\,h^{-c_1 a_2-d_1 c_2},g^{-a_1 b_2-b_1 d_2 }\,h^{c_1 b_2+d_1 d_2}\right)
	\end{equation}
	Now let us look at $((g,h)\ast V_1)\ast V_2$. Since $g$ and $h$ commute, it is given by
	\begin{equation}
	\label{inverse5}
		\begin{split}
			\left((g,h)\ast V_1\right)\ast V_2&= \left(g^{a_1}\,h^{-c_1},g^{-b_1}\,h^{d_1}\right)\ast V_2\\&=\left(g^{a_1 a_2}\,h^{-c_1 a_2}\,g^{b_1 c_2}\,h^{-d_1 c_2},g^{-a_1 b_2}\,h^{c_1 b_2}\,g^{-b_1 d_2}\,h^{d_1 d_2}\right)\\&=\left(g^{a_1 a_2 +b_1 c_2}\,h^{-c_1 a_2-d_1 c_2},g^{-a_1 b_2-b_1 d_2}\,h^{c_1 b_2+d_1 d_2}\right)\,.
		\end{split}
	\end{equation}
	If we compare \eqref{inverse5} with \eqref{inverse4}, we can deduce that
	\begin{equation}
	\label{inverse6}
		\left((g,h)\ast V_1\right)\ast V_2=(g,h)\ast (V_1\cdot V_2)\,.
	\end{equation}
	Therefore, $\ast$ defines a group right action of $\SL{2}{Z}$ on $\mathcal{P}^2$, indeed.\footnote{It is clear that this group action defines also a group action on the $\mathcal{N}=2$ sector $\mathcal{O}$ of an orbifold model.} 
\end{proof}
This result can also be seen in another way: Let $Z^{\text{1-loop}}_{(g,h)}(\tau)$ be the one-loop partition function associated to the boundary conditions $(g,h)$ and $V_1,V_2\in\SL{2}{Z}$. Then it follows that
\begin{equation}
\label{inverse7}
(V_1\cdot V_2)\cdot \tau_2Z_{(g,h)}^{\text{1-loop}}(\tau)\equiv \mathrm{Im}(V_1\cdot V_2\tau)\,Z_{(g,h)}^{\text{1-loop}}((V_1\cdot V_2)\,\tau)=\tau_2Z_{(g,h)\ast (V_1\cdot V_2)}^{\text{1-loop}}(\tau)\;.
\end{equation}
On the other hand, it is
\begin{equation}
\mathrm{Im}(V_1\cdot V_2\tau)\,Z_{(g,h)}^{\text{1-loop}}(V_1\cdot V_2\,\tau) = \mathrm{Im}(V_2\tau )\,Z_{(g,h)\ast V_1}^{\text{1-loop}}(V_2\,\tau)=\tau_2Z_{((g,h)\ast V_1)\ast V_2}^{\text{1-loop}}(\tau)
\end{equation}
and, therefore, $\left((g,h)\ast V_1\right)\ast V_2=(g,h)\ast \left(V_1\cdot V_2\right)$. This is exactly the statement of proposition \ref{inverse}, hence, ensuring definition \ref{ast} to be consistent with \eqref{219}.

As an example let $x_1,x_2,x_3$ be boundary conditions $x_i=(g_i,h_i)\in\mathcal{P}^2$ and let $Z^{\text{1-loop}}_{x_i}(\tau)$ be the torus partition function which is associated to $x_i$. Then $V_{12}$, $V_{13}$ and $V_{23}$ are defined via the following (commutative) diagram:
\begin{equation}
\begin{xy}
  \xymatrix{
      x_1 \ar[rr]^{\ast V_{12}} \ar[rd]_{\ast V_{13}}  &     &  x_2 \ar[dl]^{\ast V_{23}} \\
                             &  x_3   &
  }
\end{xy}
\end{equation}
It is obvious that
\begin{equation}
 \tau_2Z^{\text{1-loop}}_{(g_3,h_3)}(\tau)=V_{12}\cdot V_{23}\, \left(\tau_2Z^{\text{1-loop}}_{(g_1,h_1)}(\tau)\right)\,.
\end{equation}
The computation of the threshold correction $\Delta_a$ relies on the concept of closed sets of boundary conditions. Their  definition is given by
\begin{defini}\label{minset}
Let $\mathcal{O}$ be a finite set and $\ast$ be a group right action of a group $(G,\cdot)$ on $\mathcal{O}$. Then we call this set closed under $G$ if
\begin{equation}
		\forall\, x\in\mathcal{O}\,:\,\forall \, V\in G\,:\, x\ast V\in\mathcal{O}\,.
\end{equation}
If there is no $\mathcal{O}'\subsetneq\mathcal{O}$ such that $\mathcal{O}'$ is closed, we call $\mathcal{O}$ a minimally closed set.
\end{defini}
As an example for a minimally closed set consider $\mathcal{O}=\{(1,\theta^2),(\theta^2,1),(\theta^2,\theta^2)\}$ for a $\mathds{Z}_4$ orbifold. An example for a non minimally closed set is the $\mathcal{N}=2$ sector of the $\SU{2}\times\SO{10}/\mathds{Z}_8$ orbifold in ref. \cite{Stieberger4}. An immediate corollary of definition \ref{minset} is
\begin{cor}\label{sym}
 Let $\mathcal{O}$ be minimally closed under $(G,\cdot)$. Then for all $x,x'\in\mathcal{O}$ there exists a $V\in G$ such that $x'=x \ast V$.
\end{cor}
\begin{proof}
 Let $\mathcal{O}$ be minimally closed and $x\in\mathcal{O}$. Let us assume that there exists a $x'\in\mathcal{O}$ such that there is no $V\in G$ with $x'=x \ast V$. Then we can deduce that $\mathcal{O}':=x' \ast G$ is per definition closed under $G$ and $\mathcal{O}'\subsetneq\mathcal{O}$. Since $\mathcal{O}$ is minimally closed, this is a contradiction and the theorem is proven.
\end{proof}
Next, we will show 
\begin{thm}\label{closed}
 Let $\mathcal{O}$ be the $\mathcal{N}=2$ sector of an orbifold model. Then $\mathcal{O}$ is closed under $\SL{2}{Z}$ and in particular under all subgroups $G'$ of $\SL{2}{Z}$. All elements of a minimally closed set in $\mathcal{O}$ fix the same plane. 
\end{thm}
\begin{proof} 
  Let an orbifold model be given and $\theta$ denote its twist of order $N$. Furthermore, let $\theta^{k_0}$, $k_0\in\mathds{Z}$, leave exactly one plane invariant. At the beginning, we will construct a closed set of boundary conditions out of the element $\theta^{k_0}$. Then we will use this construction to prove the theorem. To construct a closed set of boundary conditions out of $\theta^{k_0}$ we have to observe that there exist $m,n\in\mathds{Z}$ such that
\begin{equation}\label{closed2}
 k_0\,m-N\,n=\gcd(k_0,N)\,,
\end{equation}
which means
\begin{equation}\label{closed3}
 k_0\,m \equiv \gcd(k_0,N)\mod N\,.
\end{equation}
Thus, $\theta^{\gcd(k_0,N)}=\theta^{k_0 m}$ and, obviously, $\theta^{\gcd(k_0,N)}$ leaves the same plane invariant as $\theta^{m\,k_0}=\left(\theta^{k_0}\right)^m$. A power $\theta^k$ of the twist $\theta$ can only fix one point, one plane or three planes. Since $\theta^{k_0}$ leaves one plane invariant, it follows that $\theta^{k_0 m}$ leaves at least the same plane invariant as $\theta^{k_0}$. Hence, $\theta^{k_0 m}$ can either fix exactly one plane, or three planes. If $\theta^{k_0 m}$ leaves three planes invariant, it is possible to choose a basis of $\mathds{R}^6$ which is left invariant under the action of $\theta^{k_0 m}$. Therefore, $\mathds{R}^6$ is left invariant and we can infer that $\theta^{k_0 m}=1$. Thus, $\theta^{k_0 m}$ either leaves exactly the same plane invariant as $\theta^{k_0}$ or is the identity. If $\theta^{k_0 m}$ is the identity, it follows that $\theta^{\gcd(k_0,N)}$ is also the identity. Since there exists an integer $\lambda\in\mathds{Z}$ such that $k_0=\lambda\,\gcd(k_0,N)$ we can infer that $\theta^{k_0}=1$. This is a contradiction to the fact that $\theta^{k_0}$ leaves exactly one plane invariant.   

Let us define
\begin{equation}\label{closed4}
 \mathcal{O}(k_0,N):=\{\left.(\theta^{k_1},\theta^{k_2})\right|k_{1,2}=\gamma_{1,2} \,\gcd(k_0,N)\,,\;\gamma_{1,2}\in\mathds{Z}\}\setminus \mathcal{M}\,,
\end{equation}
where $\mathcal{M}:=\{ (\theta^{\gamma_1 N},\theta^{\gamma_2 N}|\gamma_1,\gamma_2\in\mathds{Z} \}$ and $\setminus$ denotes the difference between two sets.

If we introduce an equivalence relation on $\{(\theta^k,\theta^l)|k,l\in\mathds{Z}\}$ as $\sim\,:\; (\theta^k,\theta^l)\sim(\theta^{k'},\theta^{l'}):\Leftrightarrow k\equiv k'\!\mod\! N \,\wedge\, l\equiv l'\!\mod\! N$, we can observe that $\mathcal{O}(k_0,N)/\!\! \sim$ is finite. Let us denote the equivalence-class of $(\theta^k,\theta^l)$ by $\left[(\theta^k,\theta^l)\right]$. Furthermore, let
\begin{equation}
\label{eqast1}
	\left[(\theta^k,\theta^l)\right]\ast V:= (\theta^k,\theta^l) \ast V\,.
\end{equation}
Since any element $(g,h)\in\left[(\theta^k,\theta^l)\right]$ can be written as $\left(\theta^{k+\alpha_k N},\theta^{l+\alpha_l N}\right)$ and 
\[
	\left(\theta^{k+\alpha_k N},\theta^{l+\alpha_l N}\right)\ast V = \left(\theta^{ak-cl+(a\alpha_k-c\alpha_l)N},\theta^{-bk+dl+(-b\alpha_k+d\alpha_l)N}\right)=
\]
\begin{equation}
\label{eqast2}
	\,\quad\qquad\qquad\qquad\qquad =(\theta^{ak-cl},\theta^{dl-bk})=(\theta^k,\theta^l)\ast V
\end{equation}
for
\begin{equation}
\label{eqast3}
	V=\begin{pmatrix} a&b\\c&d \end{pmatrix}\,,
\end{equation}
equation \eqref{eqast1} is well-defined. 

Let $\alpha\in\mathds{Z}$ be defined via $N=\alpha\,\gcd(k_0,N)$. From equation \eqref{ast2} it follows that
\begin{align*}
 &\left[(\theta^{\gamma_1\,\gcd(k_0,N)},\theta^{\gamma_2\,\gcd(k_0,N)})\right]\ast S=\left[(\theta^{\gamma_2\,\gcd(k_0,N)},\theta^{(\alpha-\gamma_1)\gcd(k_0,N)})\right]\in\mathcal{O}(k_0,N)/\!\! \sim\,,\\
 &\left[(\theta^{\gamma_1\,\gcd(k_0,N)},\theta^{\gamma_2\,\gcd(k_0,N)})\right]\ast T=\left[(\theta^{\gamma_1\,\gcd(k_0,N)},\theta^{(\gamma_2-\gamma_1)\,\gcd(k_0,N)})\right]\in\mathcal{O}(k_0,N)/\!\! \sim\,.
\end{align*}
Therefore, $\mathcal{O}(k_0,N)/\!\!\sim$ is closed under $\SL{2}{Z}$ and in particular under all subgroups $G'$ of $\SL{2}{Z}$.

Now, let another $l_0\in\mathds{Z}$ be given, with $\theta^{l_0}$ leaving exactly one plane fixed. Then $\mathcal{O}(l_0,N)$ is closed, too. Hence, $\mathcal{O}(k_0,N)\cup\mathcal{O}(l_0,N)$ is closed.

Let $\mathcal{I}$ be the set of all $k\in\mathds{Z}$ with $0<k<N$ and $\theta^k$ leaving exactly one plane invariant. Then we can deduce that
\begin{equation}
	\mathcal{O}:=\bigcup_{k\in\mathcal{I}} \limits \mathcal{O}(k,N)
\end{equation}
is closed under $\SL{2}{Z}$ and in particular under all subgroups $G'$ of $\SL{2}{Z}$.

Now we will show that all elements in a minimally closed subset $\mathcal{O}_0$ of the $\mathcal{N}=2$ sector $\mathcal{O}$ of an orbifold model leave the same plane invariant. Let $x,x'\in\mathcal{O}_0$, $x=(g,h)$ and $x'=(g',h')$. Because of corollary \ref{sym}, there exists a $V\in\SL{2}{Z}$ such that $x'=x\ast V$. Since $x\in\mathcal{O}_0\subset\mathcal{O}$, it follows that $x=(\theta^k,\theta^l)$ with $\theta^k$ and $\theta^l$ leaving the same plane invariant. Because of $Q^k x=x\Leftrightarrow Q^{-k} x = x$, $\theta^k$ and $\theta^{-k}$ leave the same plane invariant for all $k\in\mathds{Z}$. Together with the discussion below equation \eqref{closed3} we deduce that
\begin{equation}
	x\ast V=(\theta^k,\theta^l)\ast V = (\theta^{ak}\theta^{-cl},\theta^{dl}\theta^{-bk}).
\end{equation}
leaves the same plane invariant as $x$.
\end{proof}

Our next step towards proving the main result of this chapter, theorem \ref{part}, is to show that every minimally closed set in an orbifold model contains an element of the form $(1,\theta^k)$. This is summarised in
\begin{lem}\label{specialelement}
 Let $\mathcal{O}$ be a minimally closed set of boundary conditions of an orbifold model and $\theta$ its twist. Then for every $x\in\mathcal{O}$ there exist $m\in\mathds{Z}$ and $V\in\SL{2}{Z}$ such that 
\begin{equation}
x \ast V=(1,\theta^m)\in\mathcal{O}\;.
\end{equation}
\end{lem}
\begin{proof}
Let x be the boundary condition of an orbifold model. It can be written in terms of the twist $\theta$
 \begin{equation}
  x=(g,h)=(\theta^k,\theta^l)
  \;.
 \end{equation}
 Furthermore, let $V\in\SL{2}{Z}$ be parameterised according to \eqref{ast1}. Then it follows from equation \eqref{ast2} that
 \begin{equation}
  x \ast V=(\theta^k,\theta^l)\ast V=(\theta^{a\,k-c\,l},\theta^{d\,l-b\,k})\,.
 \end{equation}
 Therefore, we have to find a solution of
 \begin{equation}\label{163}
   a\,k-c\,l=\alpha\,N\quad\text{and}\quad a\,d-b\,c=1\quad\text{with}\quad \alpha\in\mathds{Z}\,.
 \end{equation} 
 The choice of $\alpha$ is completely arbitrary. The only important property is that there exists at least one $\alpha$ such that \eqref{163} has a solution. Therefore, let us choose $\alpha =0$. Then \eqref{163} implies 
 \begin{equation}\label{164}
  a\,k=c\,l.
 \end{equation}
 The general solution of equation \eqref{164} is given by
 \begin{equation}\label{165}
  a_\gamma=\frac{l}{\gcd (k,l)}\,\gamma\quad\text{and}\quad c_\gamma=\frac{k}{\gcd (k,l)}\,\gamma\,,
 \end{equation}
 with $\gamma\in\mathds{Z}$. We search for $a,c\in\mathds{Z}$ such that there exist $b,d\in\mathds{Z}$ so that $a\,d-b\,c=1$. Therefore, $a_\gamma$ and $c_\gamma$ have to fulfil $\gcd (a_\gamma,c_\gamma)=1$. This is equivalent to $\gamma=1$. Therefore, we have gained a solution of \eqref{163} and the lemma is proven.
\end{proof}
Next, we will prove two statements. Firstly, we will show that the intersection of two different minimally closed sets contained in a closed set is empty. Secondly, we will prove that a closed set is a unique disjoint union of minimally closed sets.
\begin{lem}\label{splitting}
 Let $\mathcal{O}$ be a closed set under $G$ and let $\mathcal{O}_1,\mathcal{O}_2\subset \mathcal{O}$ be distinct and minimally closed. Then it holds that $\mathcal{O}_1\cap\mathcal{O}_2=\emptyset$.
\end{lem}
\begin{proof}
 Let $V\in G$. Furthermore, let us assume that there exists a $x\in\mathcal{O}$ in such a way that $x\in\mathcal{O}_1\cap\mathcal{O}_2$. Since $\mathcal{O}_1$ \emph{and} $\mathcal{O}_2$ are closed, it follows that $x\ast V\in\mathcal{O}_1\cap\mathcal{O}_2$. This means that $\mathcal{O}_1\cap\mathcal{O}_2$ is closed itself. Since $\mathcal{O}_1$ and $\mathcal{O}_2$ are minimally closed, this is a contradiction.
\end{proof}
This enables us to state following
\begin{prop}\label{sp}
 Let $\mathcal{O}$ be a finite closed set under $G$. Then $\mathcal{O}$ is a unique disjoint union of minimally closed sets.
\end{prop}
\begin{proof}
 Let $\mathcal{O}$ be finite and closed. If $\mathcal{O}$ is minimally closed, proposition \ref{sp} is trivial. Now let $\mathcal{O}$ be not minimal. First we will show that every closed set contains at least one minimally closed set. Since $\mathcal{O}$ is closed and not minimal, it has to contain a closed set $\mathcal{O}_1\subsetneq\mathcal{O}$ by definition. Let us assume that $\mathcal{O}$ does not contain a minimally closed set. Then $\mathcal{O}_1\subsetneq \mathcal{O}$ cannot be minimally closed. Hence, $\mathcal{O}_1$ contains a closed set $\mathcal{O}_2\subsetneq \mathcal{O}_1$ which cannot be minimal. Thus, $\mathcal{O}_2$ contains a closed set $\mathcal{O}_3\subsetneq \mathcal{O}_2$ which cannot be minimal, and so on.

Altogether, there exists an infinite sequence of nested sets
\begin{equation}
 \mathcal{O}\supsetneq\mathcal{O}_1\supsetneq \mathcal{O}_2\supsetneq \mathcal{O}_3\supsetneq\ldots\,
\end{equation}
Since the set of boundary conditions is finite, this \emph{infinite} sequence is obviously a contradiction. Therefore $\mathcal{O}$ contains a minimally closed set, which we call $\mathcal{O}_1$.

Next we will show that $\mathcal{O}$ has to be the disjoint union of minimally closed sets. Since $\mathcal{O}$ is closed and $\mathcal{O}_1$ is minimally closed, we can deduce that $\mathcal{O}\setminus \mathcal{O}_1$ has to be closed, too. This closed set has to contain a minimally closed set $\mathcal{O}_2$. From lemma \ref{splitting} it follows that $\mathcal{O}_1\cap \mathcal{O}_2=\emptyset$. Now consider $\mathcal{O}\setminus \left(\mathcal{O}_1\cup \mathcal{O}_2\right)$. It is closed, too. Again, it contains a minimally closed set $\mathcal{O}_3$, with $\mathcal{O}_1\cap \mathcal{O}_2\cap \mathcal{O}_3=\emptyset$, and so on.

Since $\mathcal{O}$ is finite, this construction will terminate at some $n\in\mathds{N}$. Therefore,
\[
 \mathcal{O}=\bigcup_{k=1}^n \limits \mathcal{O}_k\,.
\]
Let us assume that there exists another decomposition $\{\mathcal{O}_k'\}$ of $\mathcal{O}$ which is truly distinct from $\{\mathcal{O}_k\}$. Then we can infer that there exists at least one $\mathcal{O}_k$ and one $\mathcal{O}_l'$ such that $\emptyset\not=\mathcal{O}_k\cap \mathcal{O}_l'\subsetneq \mathcal{O}_k$. Since lemma \ref{splitting} holds, this is a contradiction and we have proven that the decomposition of $\mathcal{O}$ is unique up to ordering ambiguities.
\end{proof}
\begin{defini}
	Let $(G,\cdot)$ be a group, $(G',\cdot)$ be a subgroup of $(G,\cdot)$, $\mathcal{O}$ be a finite set, $x\in\mathcal{O}$ and $\ast$ be a right action of $(G,\cdot)$ on $\mathcal{O}$. Furthermore, let $\mathcal{O}$ be closed under the action of $G$. Then we denote by $G_x$ ($G'_x$) the stabilizer of $(G,\cdot)$ ($(G',\cdot)$) at $x$ under $\ast$, i.e.
	\begin{align}
		&G_x:=\{V\in G\,|\,x\ast V=x\}\\
		&G'_x:=\{P\in G'\,|\,x\ast P=x\}
	\end{align}
\end{defini}
The stabilisers of $G$ and $G'$ at $x$ form a group. This is formulated in
\begin{cor}\label{stabgroup}
	Let $G_x$ and $G'_x$ be the stabilisers of $(G,\cdot)$ and $(G',\cdot)$, respectively, at $x$ under $\ast$. Then $(G_x,\cdot)$ and $(G'_x,\cdot)$ are both groups. Furthermore, $(G'_x,\cdot)$ is a subgroup of $(G_x,\cdot)$.
\end{cor}
\begin{proof}
	It holds $G_x\subset G$ and $G'_x\subset G'\subset G$. Let $V,V_1,V_2\in G_x$ and $P,P_1,P_2\in G'_x$. Then it follows that $x\ast V=x\Leftrightarrow x=x\ast V^{-1}$ and $x\ast P=x\Leftrightarrow x=x\ast P^{-1}$. Hence, $V\in G_x\Rightarrow V^{-1}\in G_x$ and $P\in G'_x\Rightarrow P^{-1}\in G'_x$. Furthermore, $x\ast(V_1\cdot V_2)=(x\ast V_1)\ast V_2=x\ast V_2=x$ and $x\ast(P_1\cdot P_2)=(x\ast P_1)\ast P_2=x\ast P_2=x$. Therefore, $V_1,V_2\in G_x\Rightarrow V_1\cdot V_2\in G_x$ and $P_1,P_2\in G'_x\Rightarrow P_1\cdot P_2\in G'_x$. Obviously, $G'_x\subset G_x$ and $(G'_x,\cdot)$ is a subgroup of $(G_x,\cdot)$.
\end{proof}
We will choose one element $x\in \mathcal{O}$ and will be interested in group elements $V_1,V_2\in G$ which lead to different elements in $\mathcal{O}$, i.e. $x\ast V_1 \not= x \ast V_2$. This leads us to
\begin{defini}
	Let $(G,\cdot)$, $(G',\cdot)$, $(G_{x},\cdot)$ and $(G'_{x},\cdot)$ be as above. Then we define four equivalence relations $\sim_1$, $\sim_2$, $\sim_3$ and $\sim_4$ as
	\begin{align}
		& \forall\, V_1,V_2\in G \, : \, V_1 \sim_1 V_2:\Leftrightarrow \exists \, g\in G_{x}\,:\, V_1=g\cdot V_2\,,\\
		& \forall\, P_1,P_2\in G' \, : \, P_1 \sim_2 P_2:\Leftrightarrow \exists \, g'\in G'_{x}\,:\, P_1=g'\cdot P_2\,,\\
		& \forall\, V_1,V_2\in G \, : \, V_1 \sim_3 V_2:\Leftrightarrow \exists \, g'\in G'_{x}\,:\, V_1=g'\cdot V_2\,,\\
		& \forall\, V_1,V_2\in G \, : \, V_1 \sim_4 V_2:\Leftrightarrow \exists \, P\in G'\,:\, V_1=P\cdot V_2\,.
	\end{align}
	We denote $G/\!\!\sim_1$ as $G/G_{x}$, $G'/\!\!\sim_2$ as $G'/G'_{x}$, $G/\!\!\sim_3$ as $G/G'_{x}$ and $G/\!\!\sim_4$ as $G/G'$.
\end{defini}
Next we define the action of an equivalence class $[V]$ on an element $x\in\mathcal{O}$.
\begin{defini}
	Let $G/G_{x}$, $G'/G'_{x}$ and $G/G'_{x}$ as above. Then we define for $[V]\in G/G_{x}$, $[P]\in G'/G'_{x}$ and $[C]\in G/G'_{x}$
	\begin{align}
		& x\ast [V]:=x\ast V\,,\\
		& x\ast [P]:=x\ast P\,,\\
		& x\ast [C]:=x\ast C\,.
	\end{align}
\end{defini}
It can be easily verified that these operations are well defined, i.e. they do not depend on the choice of representatives of the equivalence classes. Strictly speaking, we should differentiate between the equivalence classes $[V]_1$, $[V]_2$, $[V]_3$ and $[V]_4$ with respect to the equivalence relations $\sim_1$, $\sim_2$, $\sim_3$ and $\sim_4$. Since it should be clear from the context, we omit the index to simplify notation. Now we are ready to formulate one of the main results of this section, which will enable us to essentially simplify the expression of one-loop threshold corrections:
\begin{thm}\label{speqex}
	Let $g_i$, $V_j$, $P_l$ and $M_k$ be representative systems of $G_x/G'_x$, $G/G_x$, $G'/G'_x$ and $G/G'$, respectively. Then every equivalence class $[C]\in G/G'_x$ includes exactly one element $P_l\cdot M_k$ and exactly one element $g_i\cdot V_j$. 
\end{thm}
\begin{proof}
Let $C\in G$ and $[C]\in G/G_x$ be the equivalence class which contains $C$. Then there exists a unique $V_j$ such that $[C]=[V_j]$. Therefore, there exists $g\in G_x$ so that $C=g\cdot V_j$. Since $g\in G_x$, it follows that there exists a unique $g_i\in G'_x$ such that $[g]=[g_i]\in G_x/G'_x$. Then there exists a unique $g'\in G'_x$ such that $g=g'\cdot g_i$. Thus, $C$ can be uniquely decomposed as $C=g'\cdot g_i \cdot V_j$.

On the other hand, there exists a unique $M_k$ such that $[C]=[M_k]\in G/G'$. Thus, there exists a unique $P\in G$ such that $C=P\cdot M_k$. Because of $P\in G'$, it follows that there exists a unique $P_l\in G'$ so that $[P]=[P_l]\in G'/G'_x$. Hence, there exists a unique $g''\in G'_x$ such that $P=g''\cdot P_l$. Therefore, $C$ can be uniquely decomposed as $C=g''\cdot P_l \cdot M_k$.

Let us assume $[P_l\cdot M_k]=[P_{l'}\cdot M_{k'}]\in G/G'_x$. Then there exists $g'\in G'_x$ such that $P_l\cdot M_k=g'\cdot P_{l'}\cdot M_{k'}$. Therefore, $M_k\cdot M_{k'}^{-1}=P_l^{-1} \cdot g' \cdot P_{l'}\in G'$. Hence, $[M_k]=[M_{k'}]\in G/G'$. This implies together with the definition of $M_k$ that $M_k=M_{k'}$. Thus, it is true that $P_l=g'\cdot P_{l'}$ and, therefore, $[P_l]=[P_{l'}]\in G'/G'_x$. This is equivalent to $P_l=P_{l'}$. Altogether,
\begin{equation}
	M_k\not= M_{k'} \vee P_l\not= P_{l'} \Leftrightarrow[P_l\cdot M_k]\not=[P_{l'}\cdot M_{k'}]\,.
\end{equation}

Let us now assume $[g_i\cdot V_j]=[g_{i'}\cdot V_{j'}]\in G/G'_x$. Then there exists $g''\in G'_x$ such that $g_i\cdot V_j=g''\cdot g_{i'}\cdot V_{j'}$. Thus, $V_j\cdot V_{j'}^{-1}=g_i^{-1} \cdot g'' \cdot g_{i'}\in G_x$. Therefore, $[V_j]=[V_{j'}]\in G/G_x$. This implies $V_j=V_{j'}$. Hence, it holds $g_i=g''\cdot g_{i'}$ and, therefore, $[g_i]=[g_{i'}]\in G_x/G'_x$. This is equivalent to $g_i=g_{i'}$. Altogether,
\begin{equation}
	g_i\not= g_{i'} \vee V_j\not= V_{j'} \Leftrightarrow[g_i\cdot V_j]\not=[g_{i'}\cdot V_{j'}]\,.
\end{equation}

Since every $C\in G$ can be uniquely decomposed as $C=g'\cdot g_i \cdot V_j=g''\cdot P_l\cdot M_k$, it holds $[C]=[g_i\cdot V_j]=[P_l\cdot M_k]\in G/G'_x$. The equivalence class of $C$ does not change if $g'$ and $g''$ change. Thus, every equivalence class $[C]\in G/G'_x$ contains exactly one element $P_l\cdot M_k$ and exactly one element $g_i\cdot V_j$.
\end{proof}
Let us look at all combinations of $x\ast P_l\cdot M_k$. This is equivalent to all combinations of $x\ast g_i\cdot V_j$. Now, it holds that $x\ast g_i\cdot V_j=x\ast V_j$ for all $g_i$. But all combinations of $x\ast V_j$ are exactly $\mathcal{O}$. This means that we get all $x'\in\mathcal{O}$ equally often. The cardinality of how often every element is counted is given by the cardinality of all $g_i$. Hence, every element is counted $[G_x:G'_x]$ often. In particular it is easy to observe that $[G_x:G'_x]<\infty$. 

We will need
\begin{defini}\label{gensetdef}
	Let $\mathcal{O}$ be a finite set, $\ast$ be a group right action of $(G,\cdot)$ on $\mathcal{O}$ and $(G',\cdot)$ be a subgroup of $(G,\cdot)$. Then we call $\mathcal{O}^G_0(x,G'):=x\ast G'/G'_x:=\{x\ast [P]\,|\,[P]\in G'/G'_x\}$ a generating set of order $[G_x:G'_x]$.
\end{defini}
This definition can be used to state one of the main results of this section. 
\begin{thm}\label{part}
	Let $\mathcal{O}$ be the $\mathcal{N}=2$ sector of some orbifold model. Furthermore, for any minimally closed set $\mathcal{O}_k\subset\mathcal{O}$ let there exist an $x_k=(g_k,h_k)\in\mathcal{O}_k$ such that $\tau_2\,Z^{\text{1-loop}}_{(g_k,h_k)}$ is invariant under $\Gamma'_k\subset \Gamma$, $\Gamma'_k$ being a finite index subgroup of $\Gamma=\PSL{2}{Z}$. Then every minimally closed set $\mathcal{O}_k$ contains a generating set $\mathcal{O}_{0,k}=\mathcal{O}^\PSL{2}{Z}_0(x_k,\Gamma'_k)$ such that all partition functions of $x\in\mathcal{O}_{0,k}$ coincide.
\end{thm}
\begin{proof}
	Theorem \ref{closed} shows that the $\mathcal{N}=2$ sector of some orbifold model is closed under $\SL{2}{Z}$. Thus it is also closed under modular transformations $\Gamma=\PSL{2}{Z}$. Proposition \ref{sp} shows that this closed set is the unique disjoint union of minimally closed sets $\mathcal{O}_k$. Thus, by definition \ref{gensetdef} it holds that $\mathcal{O}^\PSL{2}{Z}_0(x,\Gamma'):=x\ast \Gamma'/\Gamma'_x$ is a generating set in $\mathcal{O}_k$ for all elements $x\in\mathcal{O}_k$ and all modular subgroups $\Gamma'\subset\Gamma$. By assumption it is true that all minimally closed sets contain an element $x_k=(g_k,h_k)\in\mathcal{O}_k$ with $\tau_2\,Z^{\text{1-loop}}_{(g_k,h_k)}$ being invariant under some finite index subgroup $\Gamma'_k$ of the modular group $\Gamma$. Therefore, $\mathcal{O}^\PSL{2}{Z}_0(x_k,\Gamma'_k)$ is a generating set so that all partition functions associated to elements of this generating set coincide.
\end{proof}
Now, the sum \eqref{54}
over all boundary conditions with a fixed plane effectively shrinks to a sum over the generating elements, provided that the integration domain of the integral of $\mathcal{B}_a$ over $R_{\Gamma}$ in \eqref{54} is extended to $R_{\Gamma'}=\bigcup_{l=1}^{[\Gamma:\Gamma']}M_l R_\Gamma$. We will need theorem \ref{part} later to perform the discussion of the threshold corrections to all abelian toroidal orbifold models with arbitrary discrete Wilson lines, without having to assume a certain model.
\section{General Setup and Characteristic Numbers $(\alpha,\beta,\gamma,\delta)\in\mathds{Q}^4$}\label{genset}
In this section, we are going to formulate the problem of computing threshold corrections in orbifold models more precisely. In particular, we will consider the inclusion of discrete Wilson lines, not discussed in earlier work, so far. We show how this influences the various momenta and winding lattices. It will turn out, that the Wilson lines do not change the orientation of the fixed planes but twist them in a sense we will specify precisely.

The starting point of the computation of the moduli dependent part of one-loop gauge threshold corrections $\Delta_a$ is given by \cite{Kaplunovsky}\cite{Dixon_et_al2}\cite{Stieberger4}
\begin{equation}
\label{54}
	\Delta_a=\int_{R_\Gamma} \hypmes{\tau}\,\sum_{(g,h)\in\mathcal{O}}\,b_a^{(g,h)} \tau_2\,Z^{\text{1-\text{loop}}}_{(g,h)}(\tau)-R\,.
\end{equation}
Here $\mathcal{O}$ denotes all boundary conditions on the world sheet which admit a fixed plane, i.e.\ the $\mathcal{N}=2$ sector of the orbifold model. The partition functions $Z^{\text{1-\text{loop}}}_{(g,h)}(\tau)$ associated to the boundary conditions $(g,h)$ are integrated over a fundamental domain of $\Gamma$ and the regulator $R$ is given by
\begin{equation}
\label{regul}
	R=\int_{R_\Gamma} \hypmes{\tau} \sum_{(g,h)\in\mathcal{O}}\,b_a^{(g,h)} \tau_2 \equiv b_a(\mathcal{O})\int_{R_\Gamma} \hypmes{\tau}\,\tau_2.
\end{equation}
$\Delta_a$ is a modular invariant function.

It should be stressed that $b_a^{(g,h)}$ is not modular invariant, although it seems to be a constant. Under a modular transformation $V\in\Gamma$ it transforms as
\begin{equation}
\label{ba1}
	V\; b_a^{(g,h)}
=b_a^{V\ast (g,h)}\,.
\end{equation}
Here, 
\begin{equation}
	\lim_{\tau_2\to\infty}B_a^{(g,h)}=b_a^{(g,h)}
\end{equation}
and the definition of $\mathcal{B}_a$ can be found in \cite{Kaplunovsky}, \cite{Dixon_et_al2}. It should be noted, that the $b_a^{(g,h)}$ correspond to the beta function coefficients of the theory, see the references for more details.

Furthermore, we denote the contribution of all states obeying boundary conditions $(g,h)$ to $\mathcal{B}_a$ by $\mathcal{B}_a^{(g,h)}$. In  \cite{Kaplunovsky} and \cite{Dixon_et_al2} it has also been shown that
\begin{equation}
\label{ba2}
	\mathcal{B}_a^{(g,h)}=b_a^{(g,h)}\; Z_{(g,h)}^{1-\text{loop}}\,.
\end{equation}
Equation \eqref{ba2} and the fact that
\begin{equation}
\label{ba3}
	\lim_{\tau_2\to\infty} Z_{(g,h)}^{1-\text{loop}}(\tau)=1
\end{equation}
 can be used to express $b_a^{V\ast (g,h)}$ through $b_a^{(g,h)}$. 

Theorem \ref{closed} states that $\mathcal{O}$ is closed under $\SL{2}{Z}$. Since $\mathcal{O}$ is closed under $\SL{2}{Z}$, it is clearly closed under modular transformations $\Gamma=\PSL{2}{Z}$. Here $\PSL{2}{Z}$ is constructed out of $\SL{2}{Z}$ by identifying matrices $V$ and $-V$. As a closed set it is a unique disjoint union of minimally closed sets $\mathcal{O}_k$, which is guaranteed by proposition \ref{sp}. Every minimally closed set $\mathcal{O}_k$ contains an element of the form $x_k:=\left(1,\theta^{l_k}\right)$, where $\theta^{l_k}$ leaves exactly one plane fixed here, cf. lemma \ref{specialelement}. By definition \ref{gensetdef} there exists a generating set $\mathcal{O}_0^{\PSL{2}{Z}}(x_k,\Gamma'_k)$ which contains $x_k=\left(1,\theta^{l_k}\right)$. This holds true for any finite index subgroup $\Gamma'_k$ of $\Gamma=\PSL{2}{Z}$. Later it will be convenient to choose the symmetry group of $Z_{\left(1,\theta^{l_k}\right)}^{1-\text{loop}}(\tau)$ as  $\Gamma'_k$. 

Using equation \eqref{219}
, theorem \ref{speqex}, the invariance of the hyperbolic measure under Moebius transformations as well as\footnote{Here, $M_k R_{\Gamma}$ means the action of the modular transformation $M_k$ on the fundamental domain $R_\Gamma\subset\mathds{H}^+$ as a point set.}
\begin{equation}
\label{funddom}
	R_{\Gamma_k'}=\bigcup_{l=1}^{[\Gamma:\Gamma_k']} M_l^k R_{\Gamma}
\end{equation}
for a coset decomposition
\begin{equation}\label{coset}
	\Gamma=\bigcup_{l=1}^{[\Gamma:\Gamma_k']} \Gamma' M_l^k\,,
\end{equation}
it follows that
\[
	\Delta_a+R
=\sum_k \sum_{(g,h)\in\mathcal{O}_k}\,\int_{R_\Gamma} \hypmes{\tau} \tau_2\,\mathcal{B}_a^{(g,h)}(\tau)
\]
\[
\quad\,\,\,=\sum_k\frac{1}{[\PSL{2}{Z}_{x_k}:(\Gamma'_k)_{x_k}]}\sum_{[C]\in \Gamma/(\Gamma'_k)_{x_k}}\, \int_{R_\Gamma} \hypmes{\tau} \tau_2\,\mathcal{B}_a^{x_k\ast [C]}(\tau)
\]
\[
\quad\,\,\,=\sum_k\sum_{l=1}^{[\Gamma:\Gamma'_k]}\frac{1}{[\PSL{2}{Z}_{x_k}:(\Gamma'_k)_{x_k}]}\sum_{(g,h)\in x_k\ast \Gamma'_k/(\Gamma'_k)_{x_k}}\, \int_{R_\Gamma} \hypmes{\tau} \text{Im} (M_l \tau)\,\mathcal{B}_a^{(g,h)}(M_l^k \tau)
\]
\begin{equation}
\quad\,\,\,=\sum_k\frac{1}{[\PSL{2}{Z}_{x_k}:(\Gamma'_k)_{x_k}]}\sum_{(g,h)\in x_k\ast \Gamma_k'/(\Gamma'_k)_{x_k}}b_a^{(g,h)}\int_{R_{\Gamma'_k}}\hypmes{\tau}\, \tau_2 \,Z^{1-\text{loop}}_{(g,h)}(\tau)\label{1}
\end{equation}
The change of boundary conditions acts essentially on $\tau_2 \mathcal{B}_a$. The proof of theorem \ref{closed} shows that each element of a minimally closed set which contains $x_k=(1,\theta^{l_k})$ leaves the same plane invariant as $\left(1,\theta^{l_k}\right)$. According to theorem \ref{part} all partition functions associated to $x_k\ast \Gamma_k'/(\Gamma'_k)_{x_k}=\mathcal{O}^{\PSL{2}{Z}}_0\left(\left(1,\theta^{l_k}\right),\Gamma'_k\right)$ coincide and \eqref{1} simplifies to
\begin{equation}
\label{2}
	\Delta_a+R=\sum_k \frac{1}{[\PSL{2}{Z}_{x_k}:(\Gamma'_k)_{x_k}]}\left[\sum_{(g,h)\in x_k\ast \Gamma_k'/(\Gamma'_k)_{x_k}}b_a^{(g,h)}\right]\int_{R_{\Gamma'_k}}\hypmes{\tau}\,\tau_2\,Z_{\left(1,\theta^{l_k}\right)}^{1-\text{loop}}(\tau)
	\;.
\end{equation}
The sum in front of the integral can be determined under the assumption that one-loop gauge threshold corrections have to be finite.
For convenience, we write this $\mathcal{O}_k$ dependent constant as 
\begin{equation}
\label{3}
	\frac{1}{[\PSL{2}{Z}_{x_k}:(\Gamma'_k)_{x_k}]}\sum_{(g,h)\in x_k\ast \Gamma_k'/(\Gamma'_k)_{x_k}}b_a^{(g,h)}
\equiv
A(\mathcal{O}_k) b_a(\mathcal{O}_k)\,,
\end{equation}
with
\begin{equation}
	b_a(\mathcal{O}_k):=\sum_{(g,h)\in\mathcal{O}_k} b_a^{(g,h)}\,.
\end{equation}
Now \eqref{2} yields
\begin{equation}
\label{4}
 \Delta_a=\sum_k b_a(\mathcal{O}_k) \left(A(\mathcal{O}_k)\,\int_{R_{\Gamma'_k}}\hypmes{\tau}\,\tau_2 \,Z_{\left(1,\theta^{l_k}\right)}^{1-\text{loop}}(\tau)-\int_{R_\Gamma}\hypmes{\tau} \tau_2\right)\,.
\end{equation}
So, we have expressed one-loop gauge threshold corrections by an integral of the partition function which is associated to a special element $(1,\theta^{l_k})$ over a fundamental domain of \emph{any} symmetry group of this partition function. We assumed this symmetry to be of finite group index in $\Gamma$, since we will only be interested in such groups. But what are the symmetries of $\tau_2 \,Z_{\left(1,\theta^{l_k}\right)}^{1-\text{loop}}(\tau)$? We will answer this question when we have obtained a concrete form of the partition function, later.

The partition function associated to certain boundary conditions is by definition given by a sum over all states which obey these boundary conditions. In \cite{Dixon_et_al2} and \cite{Kaplunovsky} it has been shown that only boundary conditions which leave exactly one plane invariant contribute to $\Delta_a$. Therefore, it suffices to look at an invariant sub-lattice---the fixed plane---of the lattice of all states.

Our next step is to construct these invariant sub-lattices. For that purpose we need a representation of the twist in terms of a lattice basis of the (complete) lattice of states. This representation of $\theta$ is given by $Q\in\SO{6}$ with $Q^N=\mathds{1}$ for some $N$, which we call the order of the twist. Let us denote the quantum numbers of all states by $w$, $p$ and $l$, where $w\in\mathds{Z}^6$ denotes a vector in the defining lattice of the theory, the compactification lattice, $p\in\mathds{Z}^6$ denotes a momentum vector, out of the dual lattice, and $l\in\mathds{Z}^{16}$ denotes the momentum in the additional $E_8\times E_8$ lattice.\footnote{Although we will focus on the $E_8\times E_8$ heterotic string, the discussion of the $\SO{32}$ heterotic string follows analogously.} All these vectors are written in a basis of the corresponding lattice. We combine them into a tuple of quantum numbers $u:=(w,p,l)\in\mathds{Z}^{26}$. They are connected to the corresponding left- and right moving physical momenta of the string via \cite{Witten1}
\begin{equation}
\label{8}
\begin{split}
	&p_L= \left(\frac{p}{2}+\left(\mathrm{g}-\mathrm{b}-\frac{1}{4} A^T C A\right)w -\frac{1}{2} A^T C l\,,\,l+A w\right)\quad \text{and}\\
	&p_R= \left(\frac{p}{2}-\left(\mathrm{g}+\mathrm{b}+\frac{1}{4} A^T C A\right)w -\frac{1}{2} A^T C l\,,\,0\right)\;.
\end{split}
\end{equation}
with $C$ denoting the metric on the root lattice under consideration and $A\in\mathds{Q}^{16\times 6}$ being the matrix of discrete Wilson lines. The matrices $\mathrm{g}$ and $\mathrm{b}$ denote the metric tensor and antisymmetric background, which are the most general symmetric and anti-symmetric matrices compatible with the orbifold twist \cite{ErlerKlemm}
\begin{equation}\label{metricdef}
\begin{aligned}Q^T \mathrm{g}\, Q=\mathrm{g} \\ \mathrm{g}^T = \mathrm{g}\end{aligned} \qquad\text{and}\qquad\begin{aligned}Q^T \mathrm{b}\, Q=b\\ b^T = -b\;.\end{aligned}
\end{equation}

To search for the invariant subspaces, we need to give the action of the orbifold twist on the quantum numbers $u$. It is given by \cite{NULLES2}
\begin{equation}
\label{9}
	\mathcal{Q}=\begin{pmatrix} Q							&	0	& 	0									\\
								\xi						&	Q^*	&	\left(\mathds{1}-Q^*\right)A^T C	\\
								A\left(\mathds{1}-Q\right)	&	0	&	\mathds{1}	\end{pmatrix}\,,
\end{equation}
with
\begin{equation}
	\xi = \frac{1}{2} A^T C A (\mathds{1}-Q)+\frac{1}{2}\left(\mathds{1}-Q^*\right)A^T C A \quad \text{and} \quad Q^*:=\left(Q^{-1}\right)^T=SQS^{-1}\,.
\end{equation}
Therefore, we search for powers of $l$ such that (integral) solutions of
\begin{equation}
\label{10}
	\mathcal{Q}^l u = u
\end{equation}
exist and want to determine them. To perform this task it is useful to have a closed formula for $\mathcal{Q}^l$. It can be constructed if we diagonalise $\mathcal{Q}$. For that purpose we define
\begin{align}
\label{11}
	&\hat w := w\,,\\
	&\hat p := p-\frac{1}{2} A^T C A w- A^T C l\quad \text{and}\\
	&\hat l := l + A w\,.
\end{align}
Then a new operator $\hat{\mathcal{Q}}$, which exactly corresponds to $\mathcal{Q}$, acts on $\hat u:=(\hat w,\hat p,\hat l)$  by
\begin{equation}
\label{12}
	\hat{ \mathcal{Q}}\,\hat u=	\begin{pmatrix} 	Q & 0   & 0 \\
													0 & Q^* & 0 \\
													0 & 0   & \mathds{1} \end{pmatrix} 
								\begin{pmatrix} \hat w \\ \hat p \\ \hat l \end{pmatrix}\,.
\end{equation}
A simple calculation shows that
\begin{equation}
\label{13}
		\hat u=\Omega u\quad\text{and}\quad\hat{\mathcal{Q}}=\Omega\,\mathcal{Q} \, \Omega^{-1}\,, \quad\text{with} \quad \Omega = \begin{pmatrix}	\mathds{1}			  &	  0			&0\\
						  -\frac 1 2 A^T C A	&	\mathds{1}&-A^T C\\
						  A						&	0		  &\mathds{1}\end{pmatrix}\,.
\end{equation}
Furthermore, it is easy to verify
\begin{equation}
\label{14}
	\Omega^{-1}=\begin{pmatrix} \mathds{1}			  	&	0			&	0\\
								-\frac 1 2 A^T C A		&	\mathds{1}	&	A^T C\\
								-A						&	0		  	&	\mathds{1}\end{pmatrix}\,.
\end{equation}
If we use $\hat{\mathcal{Q}}^k=\Omega \mathcal{Q}^k \Omega^{-1}$ we get
\begin{equation}
\label{15}
	\mathcal{Q}^k=	\begin{pmatrix} Q^k								&	0		& 	0									\\
									\xi'							&	(Q^*)^k	&	\left(\mathds{1}-(Q^*)^k\right)A^T C	\\
									A\left(\mathds{1}-Q^k\right)	&	0		&	\mathds{1}	\end{pmatrix}\,,
\end{equation}
where
\begin{equation}
	\xi' = \frac{1}{2} A^T C A \left(\mathds{1}-Q^k\right)+\frac{1}{2}\left(\mathds{1}-(Q^*)^k\right)A^T C A\,.
\end{equation}
The system of equations \eqref{10} now reads
\begin{align}
\label{16}
	&Q^k w = w\,,\\
\label{17}
	&\frac{1}{2} A^T C A \left(\mathds{1}-Q^k\right)w+\frac{1}{2}\left(\mathds{1}-(Q^*)^k\right)A^T C A w+\left(Q^*\right)^k p + \left(\mathds{1}-(Q^*)^k\right)A^T C l = p\,,\\
\label{18}
	&A\left(\mathds{1}-Q^k\right)w+l=l\,.
\end{align}
Examining these equations, we can immediately observe that the power of the twist which leaves one plane invariant does not change if we allow for discrete Wilson lines. The power is completely independent of them. But the concrete form of the sub-lattice will change. If a power of the twist admits a fixed plane, this plane can be parameterised with two \emph{real} variables. Since we are interested in the sub-lattice which lies inside this plane, we can parameterise this sub-lattice by two integral valued variables. So, from a more physical point of view, the position or orientation of the fixed planes within the winding and momentum lattices remains unchanged by switching on Wilson lines. However, non-vanishing Wilson lines do deform the lattices. 

Because \eqref{54} is a sum over boundary conditions which admit a fixed plane, it follows that \eqref{16} is solvable and can be parameterised by two integral variables $(n_1,n_2)\in\mathds{Z}^2$ via
\begin{equation}
\label{19}
	w=W \cdot \begin{pmatrix} n_1 \\ n_2\end{pmatrix}\quad \text{where}\quad W\in\mathds{Z}^{6\times 2}\,.
\end{equation}
Then \eqref{17} and \eqref{18} simplify accordingly to
\begin{align}
\label{20}
	\left(\mathds{1}-(Q^*)^k\right)\left(\frac{1}{2}A^T C A w + A^T C l \right)&= \left(\mathds{1}-(Q^*)^k\right)p\\
\label{21}
	l&=l\,.
\end{align}
Equation \eqref{21} is always true and \eqref{20} can be solved by applying Gau\ss ' algorithm, in which we only add and subtract multiples of one row to another. Since the right-hand side of \eqref{20} contains rational numbers only, the result of the algorithm will be a (linear) sum of rational multiples of the variables of $w$, $l$ and two additional variables $m_1$ and $m_2$, which are present since we look at a fixed plane in the dual lattice. Now this parameterised object has to be a subset of the original lattice, which corresponds to the fact that all components of $p$ which solve \eqref{20} have to be integral numbers. Since $(w,p,l)=(0,0,0)$ is a special solution of equations \eqref{16}, \eqref{17} and \eqref{18}, the theory of linear Diophantine equations tells us that the most general solution $(w,p,l)$ is given by,
\begin{equation}
\label{varold}
	\begin{pmatrix} w \\ p \\ l \end{pmatrix} = L \cdot \begin{pmatrix} m_1 \\ m_2 \\ n_1 \\ n_2 \\ l_1 \\ \vdots \\ l_{16} \end{pmatrix}\,,\quad \text{where} \quad L \in \mathds{Z}^{20\times 20}\,.
\end{equation}
We have gained a solution of the system of Diophantine equations \eqref{16} to \eqref{18}. They determine the hatted variables uniquely. Let us observe that \eqref{20} can be written as
\begin{equation}
\label{23}
	\left(\mathds{1}-(Q^*)^k\right)\left(p-\frac{1}{2}A^T C A w - A^T C l \right)= 0 \Leftrightarrow\left(\mathds{1}-(Q^*)^k\right) \hat p = 0 \,.
\end{equation}
Obviously it is more natural to view $\hat w$, $\hat p$ and $\hat l$ as fundamental variables, since they transform as ordinary winding, momentum and lattice vectors under the action of the twist. In these variables the left and right moving momenta read
\begin{equation}
\label{28}
\begin{split}
	&p_L= \left(\frac{\hat p}{2}+\left(\mathrm{g}-\mathrm{b}\right)\hat w \,,\,\hat l\right)\quad \text{and}\\
	&p_R= \left(\frac{\hat p}{2}-\left(\mathrm{g}+\mathrm{b}\right)\hat w \,,\,0\right)\,.
\end{split}
\end{equation}
This expression is formally the same as it is without Wilson lines. Recall that in this case the additional sum over the $E_8\times E'_8$ lattice can be absorbed into the beta function coefficients as was shown in \cite{Dixon_et_al2}.

This can be understood from a physical point of view. The $E_8\times E'_8$ momenta do not influence the partition function of the internal manifold, but they do influence the gauge degrees of freedom, hence, the beta function coefficients. Therefore, all information needed from the $E_8\times E'_8$ lattice is already encoded in the computation of the beta function coefficients.

For the following argumentation we need 
\begin{equation}
\label{26}
	\{ a_1 x_1 + a_2 x_2 + \ldots + a_n x_n | x_1 ,\, \ldots ,\, x_n \in \mathds{Z}\}=\{\text{gcd}(a_1,\,\ldots,\,a_n)k|k\in\mathds{Z}\}\,,
\end{equation}
for given $a_1,\,\ldots,\,a_n\in\mathds{Z}$. This can be seen if we observe that $c=a_1 x_1 + a_2 x_2 + \ldots + a_n x_n$ is solvable if and only if $\text{gcd}(a_1,\,\ldots,\,a_2)$ divides $c$.\footnote{This can be seen as getting all equations when going through all possible solutions.}

Now if we solve \eqref{23} for $\hat p$ we get as a solution a plane lying in the six dimensional space representing momentum vectors. Its parametrisation is given by the one which we would have gained if there were no Wilson lines present. Therefore, if we put the solutions $w$, $p$, and $l$ into \eqref{23} we get, using above argument, some sub-lattice of this plane. It differs from the intersection of the plane which solves \eqref{23} with the lattice $\mathds{Z}^6$, which corresponds to all quantized momenta. Sometimes it can happen that the sub-lattice not only contains all points of this intersection, but even more, i.e.\ the momentum $\hat p$ becomes fractional. In that sense Wilson lines deform the lattices and if we switch off the Wilson lines, the lattices will be of the shape as it has to be without them.

Next, we will show that the hatted momentum and the hatted winding lattice can each be parameterised via two integral variables.\footnote{Above they were parameterised by twenty integral variables.} 

Solving \eqref{16}, \eqref{20} and \eqref{21} for vanishing Wilson lines results in a fixed plane in momentum space and a fixed plane in winding space. Let us parameterise the former through two integral variables $m_1$, $m_2\in\mathds{Z}$ and the latter through $n_1$, $n_2\in\mathds{Z}$. Switching on the Wilson lines results, with the condition of $p$ being a vector in $\mathds{Z}^6$, in a deformation of the variables according to\footnote{Note that, although $\hat w = w$, the condition of $p\in \mathds{Z}^6$ in equation \eqref{23} causes an interdependence of $w=\hat w$ and $l$. Therefore, the solution for $\hat w$ depends, in general, not only on $n_1,n_2$ but also on $l_1,\ldots ,l_{16}$.}
\begin{equation}
\label{latticesub1}
	\begin{pmatrix} \hat n_1 \\ \hat n_2 \end{pmatrix} =\begin{pmatrix} \alpha_w & 0 \\ 0 & \beta_w \end{pmatrix} \cdot M_w \cdot \begin{pmatrix} n_1 \\ n_2 \\ l_1 \\ \vdots \\ l_{16} \end{pmatrix}\,, \quad \text{with} \quad M_w \in \mathds{Z}^{2\times 18}\quad \text{and}
\end{equation}
\begin{equation}
\label{latticesub2}
	\begin{pmatrix} \hat m_1 \\ \hat m_2 \end{pmatrix} =\begin{pmatrix} \alpha_p & 0 \\ 0 & \beta_p \end{pmatrix} \cdot M_p \cdot \begin{pmatrix} m_1 \\ m_2 \\ n_1 \\ n_2 \\ l_1 \\ \vdots \\ l_{16} \end{pmatrix}\,, \quad \text{with} \quad M_p \in \mathds{Z}^{2\times 20}\,.
\end{equation}
Here, $\hat n_1$, $\hat n_2$ parametrise the hatted winding lattice, whereas $\hat m_1$, $\hat m_2$ parametrise the hatted momentum lattice (both in the presence of discrete Wilson lines). The diagonal matrix is chosen such that $M_w$ and $M_p$ are integral, indeed. The entries in their rows are relatively prime to the entries in the same row. Thus, we can state that $\hat n_1$ takes all values in $\alpha_w \mathds{Z}$. If $\hat n_2$ and $\hat n_1$ have no variables in common, it immediately follows that $\hat n_2$ takes all values in $\beta_w \mathds{Z}$, independently of $\hat n_1$. If they have common variables, let us express one of those variables in terms of $\hat n_1/\alpha_w=:\hat n'_1$. Thus, a possible parametrisation of the lattice given by \eqref{latticesub1} reads
\begin{equation}
\label{latticesub3}
	\begin{pmatrix} \hat n_1 \\ \hat n_2 \end{pmatrix} = \begin{pmatrix} \alpha_w & 0 \\ 0 & \beta_w \end{pmatrix} \cdot \begin{pmatrix} 1 & 0_{1\times 17} \\ c'_1 & M'_w \end{pmatrix} \cdot \begin{pmatrix} \hat n'_1 \\ n_2 \\ l_1 \\ \vdots \\ l_{16} \end{pmatrix}\,,
\end{equation}  
with $c_1\in\mathds{Q}$, $0_{1\times 17}$ being a $1\times 17$ matrix of zeros and $M'\in \mathds{Z}^{1\times 17}$. Furthermore, let
\begin{equation}
\label{latticesub4}
	M'_w=\begin{pmatrix} c'_2 & c'_3 & \cdots & c'_{18} \end{pmatrix}\,,
\end{equation}
 where $c'_2$, $c'_3$, ..., $c'_{18}\in\mathds{Q}$. 

Let us denote the greatest common divisor of rational numbers $a,b\in\mathds{Q}$ by the rational number $c$, so that $c^{-1}$ is the lowest rational number such that $a\,c^{-1}\in\mathds{Z}$ and $b\,c^{-1}\in\mathds{Z}$. Its concrete form is given by
\begin{prop}
For $a=u/v\in\mathds{Q}$, $b=x/y\in\mathds{Q}$ and $\gcd(u,v)=\gcd(x,y)=1$ it holds
\begin{equation}
\label{latticesub5}
	\gcd\left(\frac{u}{v},\frac{x}{y}\right)=\frac{\gcd(u,x)}{\lcm(v,y)}\,.
\end{equation}
\end{prop}
\begin{proof}One can prove this statement as follows. Since the right hand side of \eqref{latticesub5} clearly divides both $a$ and $b$, it follows that the right hand side of \eqref{latticesub5} also divides $\gcd(a,b)$. Hence, there exists a $k\in\mathds{Z}$ such that
\begin{equation}
\label{latticesub6}
	\gcd\left(\frac{u}{v},\frac{x}{y}\right)=k\,\frac{\gcd(u,x)}{\lcm(v,y)}
\end{equation}
which yields
\begin{equation}
\label{latticesub7}
	\frac{u}{v}\cdot\frac{\lcm(v,y)}{k\,\gcd(u,x)}\in\mathds{Z}\quad \text{and} \quad \frac{x}{y}\cdot\frac{\lcm(v,y)}{k\,\gcd(u,x)}\in\mathds{Z}\,.
\end{equation}
This is equivalent to
\begin{equation}
\label{latticesub8}
	k\left| \frac{u}{\gcd(u,x)}\cdot\frac{\lcm(v,y)}{v} \right.\quad \text{and} \quad k\left| \frac{x}{\gcd(u,x)}\cdot\frac{\lcm(v,y)}{y} \right.\,.
\end{equation}
Using $\gcd(v,y)\cdot \lcm(v,y)=v\cdot y$ and $\gcd(u,v)=\gcd(x,y)=1$ we infer that $k|1$ and the statement is true.
\end{proof}
Now let us define $\beta'_w:=\gcd(c'_1,c'_2,\ldots,c'_{18})$ and $c_i:=c'_i/\beta'_w\in\mathds{Z}$ for $i=1,2,\ldots,18$. Then the lattice parameterised by \eqref{latticesub3} can be written as
\begin{equation}
\label{latticesub9}
	\begin{pmatrix} \hat n_1 \\ \hat n_2 \end{pmatrix} = \begin{pmatrix} \alpha_w & 0 \\ 0 & \beta_w \,\beta'_w \end{pmatrix} \cdot \begin{pmatrix} 1 & 0 & \cdots & 0 \\ c_1 & c_2 & \cdots & c_{18} \end{pmatrix} \cdot \begin{pmatrix} \hat n'_1 \\ n_2 \\ l_1 \\ \vdots \\ l_{16} \end{pmatrix}\,.
\end{equation}
Our next step to show that the winding lattice and the momentum lattice can each be parameterised via two independent integral variables is 
\begin{prop} The lattice $\{(y_1,y_2)\}$ given by all pairs of $(y_1,y_2)$ with
\begin{equation}
\label{latticesub10}
	\begin{pmatrix} y_1 \\ y_2 \end{pmatrix} = \begin{pmatrix} x_1 \\ c_1\,x_1 + c_2\,x_2 + \ldots + c_{n}\,x_{n} \end{pmatrix}\quad \text{with} \quad c_i,\,x_i\in\mathds{Z}\,,\,\,n\ge 2
\end{equation}
can be parameterised via two integral variables, i.e.\\ there exist $\bar c_1,\bar c_2\in\mathds{Z}$ such that
\begin{equation}
\label{latticesub11}
	\left\{\begin{pmatrix} x_1 \\ c_1\,x_1 + c_2\,x_2 + \ldots + c_{n}\,x_{n} \end{pmatrix}: x_i\in\mathds{Z}\right\}=\left \{ \begin{pmatrix} x_1 \\ \bar c_1 \, x_1 + \bar c_2 \, \chi \end{pmatrix} : x_1,\chi\in\mathds{Z}\right \}
\end{equation}
\end{prop}
\begin{proof}
We have to show that there exist $\bar c_1,\bar c_2\in\mathds{Z}$ such that for all $x_2,x_3\ldots,x_n\in\mathds{Z}$ there exists one $\chi\in\mathds{Z}$ in such a way that
\begin{equation}
\label{latticesub12}
	c_1\,x_1+c_2\,x_2+\ldots + c_n\,x_n=\bar c_1\,x_1 + \bar c_2\,\chi
\end{equation}
independently of $x_1$. If we look at the difference of the left hand side and the right hand side of equation \eqref{latticesub12} and use the lemma of B\'ezout, it becomes clear that we have to choose $\bar c_1=c_1$ and $\bar c_2=\gcd(c_2,c_3,\ldots,c_n)$. This proves the proposition.
\end{proof}
Let us apply this proposition to \eqref{latticesub9}. The momentum lattice can, therefore, be written as
\begin{equation}
\label{latticesub13}
\begin{split}
\begin{pmatrix} \hat n_1 \\ \hat n_2 \end{pmatrix} &= \begin{pmatrix} \alpha_w & 0 \\ 0 & \beta_w \,\beta'_w \end{pmatrix} \cdot \begin{pmatrix} \hat n'_1 \\c_1 \,\hat n'_1 + \gcd(c_2,c_3,\ldots,c_n)\hat n'_2\end{pmatrix}=\\
&=\begin{pmatrix} \alpha_w & 0 \\ 0 & \beta_w \,\beta'_w \end{pmatrix}\cdot \begin{pmatrix} 1 & 0 \\ c_1 & \gcd(c_2,c_3,\ldots,c_n) \end{pmatrix} \cdot \begin{pmatrix} \hat n'_1 \\ \hat n'_2 \end{pmatrix}\,.
\end{split}
\end{equation}
Observe that $\gcd(c_2,c_3,\ldots,c_n)=1$ for $n\ge 3$ and $c_2$ for $n=2$. Hence, we have shown that we can always parameterise the hatted winding lattice with two variables. To achieve this we had to redefine the variables of the hatted momentum lattice and of the $E_8\times E_8$ lattice. Next we have to express $n_1,n_2,l_1,\ldots,l_{16}$ in terms of these new variables. This is clearly possible. Then we can apply the same argumentation as above to show that the momentum lattice can also be parameterised via two integral variables. Altogether, we can state that there exist matrices $A_w,A_p\in\mathds{Q}^{2\times 2}$ and $A_l\in\mathds{Q}^{16\times 16}$ such that
\begin{align}
\label{latticesub14}
& \hat w =W\cdot\begin{pmatrix} \hat n'_1 \\ \hat n'_2 \end{pmatrix} = W\cdot A_w \cdot \begin{pmatrix}  n'_1 \\ n'_2 \end{pmatrix}\,,\quad \text{for} \quad n'_1,n'_2\in\mathds{Z}\,,\\
\label{latticesub15}
& \hat p = P \cdot \begin{pmatrix} \hat m_1 \\ \hat m_2 \end{pmatrix} =P\cdot A_p \cdot \begin{pmatrix}  m'_1 \\ m'_2 \end{pmatrix}\,,\quad \text{for} \quad m'_1,m'_2\in\mathds{Z}\quad \text{and}\\
\label{latticesub16}
&\hat l= \begin{pmatrix} \hat l_1 \\ \hat l_2 \\ \vdots \\ \hat l_{16} \end{pmatrix} = A_l \cdot \begin{pmatrix}  l'_1 \\ l'_2 \\ \vdots \\ l'_{16} \end{pmatrix}\,,\quad \text{for} \quad l'_1,l'_2,\ldots,l'_{16}\in\mathds{Z}
\end{align} 
parameterise the fixed planes for non-vanishing Wilson lines. Here, $W$ and $P$ are defined in such a way that they parameterise the fixed planes in momentum and winding lattice for vanishing Wilson lines (cf. eq \eqref{19}). It is crucial that we can choose variables such that all lattices decouple from the other two lattices. In particular, this allows us to absorb the sum over the $E_8\times E'_8$ lattice into the beta function coefficients. Later we will see how to deal with these matrices.

To write down the moduli of the lattice, we have to firstly construct a metric and an antisymmetric tensor through \eqref{metricdef}. Secondly, we have to construct the projection of the metric and antisymmetric tensor field onto the constructed sub-lattice $\hat w$. It is uniquely defined in terms of $\hat w$, $\mathrm{g}$ and $\mathrm{b}$ via
\begin{align}
\label{29}
	&\hat w^T g\, \hat w=\begin{pmatrix} n_1 & n_2 \end{pmatrix}\cdot g^\perp \cdot \begin{pmatrix} n_1 \\ n_2 \end{pmatrix} \quad \text{and}\\
\label{30}
	&\hat w^T b\, \hat w=\begin{pmatrix} n_1 & n_2 \end{pmatrix}\cdot b^\perp \cdot \begin{pmatrix} n_1 \\ n_2 \end{pmatrix}\,,
\end{align}
where $g^\perp$ is symmetric and $b^\perp$ is antisymmetric. Since $b^\perp$ is an antisymmetric $2\times 2$ matrix, it follows that it contains only one independent entry, which we also denote as $b^\perp$. Given the metric and the antisymmetric tensor in the fixed plane parameterised by $\hat w$, we can write down the moduli of this plane\footnote{We work with the convention $\alpha'=2$.}:
\begin{align}
\label{31}
	&T=T_1+\I\, T_2=2\left(\mathrm{b}^\perp+\I\sqrt{\det \mathrm{g}^\perp}\right)\quad \text{and}\\
\label{32}
	&U=U_1+\I\, U_2=\frac{1}{\mathrm{g}^\perp_{11}}\left(\mathrm{g}^\perp_{12}+\I\sqrt{\det \mathrm{g}^\perp}\right)
\end{align}
It is important to note that the moduli of the fixed plane do not change if we switch on discrete Wilson lines. This is true since the power of the twist and the real two-dimensional vector space which contains the fixed plane (the integer variables considered as real ones) remain unchanged.

Now we are able to evaluate the one-loop partition function associated to the boundary conditions $\left(1,\theta^{l_k}\right)$. In the following we will need certain tools, whose derivation we postpone until chapter \ref{chapfour} where their careful development does not disturb the line of reasoning.

Let $\Lambda_k^*$ denote the lattice of all states in momentum space and $\left(\Lambda_k^*\right)^\perp$ the invariant sub-lattice. Then the general expression for the partition function which is associated to $\left( 1,\theta^{l_k}\right)$ is given by \cite{ErlerKlemm}
\begin{equation}
\label{34}
	Z_{\left(1,\theta^{l_k}\right)}^{1-\text{loop}}(\tau)=\sum_{p\in\left(\Lambda_k^*\right)^\perp} q^{\,p^2_L/2}\,{\bar q}^{\,p^2_R/2}\,.
\end{equation}
Without discrete Wilson lines this can be simplified by putting \eqref{28}, \eqref{31} and \eqref{32} into \eqref{34} and using the Smith normal form (SNF) (cf.\ theorem \ref{SNF}) of the matrix $P^T\cdot W$. This would result in a partition function
\begin{equation}
\label{33}
\begin{split}
	Z^{\text{no-WL}}(\tau)=&\sum_{\genfrac{}{}{0pt}{2}{n_1,n_2\in\mathds{Z} }{m_1,m_2\in\mathds{Z}}}\ex{2\pi \I \,\tau(\gamma_1\, m_1 n_1+\gamma_2\, m_2 n_2 )}\\&\times\exp\left[-\frac{\pi\,\tau_2}{T'_2\,U'_2}\left| T' U' \,n_2+T' \,n_1 -\gamma_1\,m_1\,U'+\gamma_2\,m_2\right|^2\right]=\\
	&=\sum_{A}\ex{-2\pi\I\,\tau\,\det A}\exp\left[-\frac{\pi\,\tau_2}{T'_2\,U'_2}\left| \begin{pmatrix} 1 & U' \end{pmatrix} \cdot A \cdot \begin{pmatrix} T' \\ 1 \end{pmatrix}\right|^2\right]\,,
\end{split}
\end{equation}
with $\gamma_1|\gamma_2\in\mathds{Z}$, $T'$ and $U'$ being determined by the SNF and
\begin{equation}
\label{wilsonpart1}
	A=\begin{pmatrix} n_1 & \gamma_2 m_2 \\ n_2 & - \gamma_1 m_1 \end{pmatrix}=\begin{pmatrix} n_1 & 0 \\ n_2 & 0 \end{pmatrix} + \begin{pmatrix} 0 & 1 \\ -1 & 0 \end{pmatrix} \cdot \begin{pmatrix} \gamma_1 & 0 \\ 0 & \gamma_2 \end{pmatrix} \cdot \begin{pmatrix} 0 & m_1 \\ 0 & m_2 \end{pmatrix} \,.
\end{equation}

If we allow for discrete Wilson lines, we have to sum over all hatted variables which give integral momentum $p$, instead of all unhatted variables in \eqref{33}. We showed that this results in a redefinition of momenta and winding according to equations \eqref{latticesub14} - \eqref{latticesub16}. Inserting this redefinition into \eqref{wilsonpart1} yields
\begin{equation}
\label{wilsonpart2}
	A=A_w \cdot \begin{pmatrix} n_1 & 0 \\ n_2 & 0 \end{pmatrix} + \begin{pmatrix} 0 & 1 \\ -1 & 0 \end{pmatrix} \cdot \begin{pmatrix} \gamma_1 & 0 \\ 0 & \gamma_2 \end{pmatrix} \cdot A_p \cdot \begin{pmatrix} 0 & m_1 \\ 0 & m_2 \end{pmatrix}\;.
\end{equation}
Let us denote the greatest common divisor of the entries of $A_w$ and $A_p$ by $\omega_w$ and $\omega_p$, respectively. Then $A'_w:={\omega_w}^{-1} A_w$ and $A'_p:={\omega_p}^{-1} A_p$ are both integral matrices. Equation \eqref{wilsonpart2} can now be written as
\begin{equation}
\label{wilsonpart3}
	A= \omega_w\,A'_w \cdot \begin{pmatrix} n_1 & 0 \\ n_2 & 0 \end{pmatrix} + \omega_p\, \begin{pmatrix} 0 & 1 \\ -1 & 0 \end{pmatrix} \cdot \begin{pmatrix} \gamma_1 & 0 \\ 0 & \gamma_2 \end{pmatrix} \cdot A'_p \cdot \begin{pmatrix} 0 & m_1 \\ 0 & m_2 \end{pmatrix}\,.
\end{equation}
Using theorem \ref{SNF} and the discussion below this theorem, we can infer that there exist two matrices $P_1,Q_1\in\SL{2}{Z}$ and a diagonal matrix $D_1$ with integral components, such that $A'_w={P_1}^{-1} \cdot  D_1 \cdot {Q_1}^{-1}$. Then \eqref{wilsonpart3} is given by
\begin{equation}
\label{wilsonpart4}
	A=P_1^{-1}\cdot \left( \omega_w\,D_1 \cdot {Q_1}^{-1} \begin{pmatrix} n_1 & 0 \\ n_2 & 0 \end{pmatrix} + \omega_p\, S^{-1}\,SP_1 S^{-1} \cdot \begin{pmatrix} \gamma_1 & 0 \\ 0 & \gamma_2 \end{pmatrix} \cdot A'_p \cdot \begin{pmatrix} 0 & m_1 \\ 0 & m_2 \end{pmatrix} \right)\;.
\end{equation}
Using the discussion below equation \eqref{78}, the prefactor ${P_1}^{-1}$ results in a modular transformation of $U$ and, again, we should redefine summation variables $n_1,n_2\in\mathds{Z}$,
\begin{align}
\label{wilsonpart5}
	&\begin{pmatrix} n'_1 \\ n'_2 \end{pmatrix} = {Q_1}^{-1} \cdot \begin{pmatrix} n_1 \\ n_2 \end{pmatrix}\\
\label{wilsonpart6}
	&U''=({P_1}^{-1})^\sharp \,U'\,,\quad \text{with} \quad \begin{pmatrix} a & b \\ c & d \end{pmatrix}^\sharp:= \begin{pmatrix} d & b \\ c & a \end{pmatrix}\,. 
\end{align}
For a general $m\times n$ matrix $A$ it holds in analogy that $(A^\sharp)_{i,j}=A_{n+1-j,m+1-i}$. Observe that $Q_1^{-1} \cdot \mathds{Z}^2=\mathds{Z}^2$ for $Q_1^{-1}\in\SL{2}{Z}$. If we use $S\, M\, S^{-1}=M^*$ for any $\mathds{Z}^{2\times 2}$, we arrive at
\begin{equation}
\label{wilsonpart7}
	A'=\omega_w\,D_1 \begin{pmatrix} n'_1 & 0 \\ n'_2 & 0 \end{pmatrix} + \omega_p\, S^{-1}\,{P_1}^* \cdot \begin{pmatrix} \gamma_1 & 0 \\ 0 & \gamma_2 \end{pmatrix} \cdot A'_p \cdot \begin{pmatrix} 0 & m_1 \\ 0 & m_2 \end{pmatrix}\,.
\end{equation}

Replacing $A$ by $A'$ in \eqref{33} and using the discussion below equations \eqref{62} and \eqref{63}, we can deal with the diagonal matrix $\omega_w\, D_1$, which has rational entries. Hence, we set $\omega_w\, D_1 =: \diag (\alpha_1,\beta_1)$ and rescale the moduli according to
\begin{align}
\label{wilsonpart8}
	&T\longmapsto T'=\alpha_1 \beta_1\,T\\
\label{wilsonpart9}
	&U''\longmapsto U'''=\frac{\beta_1}{\alpha_1}\, U''
\;,
\end{align}
then the matrix $A$ in \eqref{33} is replaced by 
\begin{equation}
\label{wilsonpart10}
	A''=\begin{pmatrix} n'_1 & 0 \\ n'_2 & 0 \end{pmatrix} + \omega_p\,\omega_w\, S^{-1} \cdot D_1 \cdot \,{P_1}^* \cdot \begin{pmatrix} \gamma_1 & 0 \\ 0 & \gamma_2 \end{pmatrix} \cdot A'_p \cdot \begin{pmatrix} 0 & m_1 \\ 0 & m_2 \end{pmatrix}\,,
\end{equation}
where $n'_1,n'_2,m_1,m_2\in\mathds{Z}$. By construction it is true that
\begin{equation}
\label{wilsonpart11}
	A''_p:=D_1 \cdot \,{P_1}^* \cdot \begin{pmatrix} \gamma_1 & 0 \\ 0 & \gamma_2 \end{pmatrix} \cdot A'_p\in\mathds{Z}^{2\times2}\,.
\end{equation}
Thus, again by theorem \ref{SNF} and the discussion below this theorem we can infer that there exist two matrices $P_2,Q_2\in\SL{2}{Z}$ and a diagonal matrix $D_2$ with integral components, such that $A''_p={P_2}^{-1} \cdot  D_2 \cdot {Q_2}^{-1}$ and equation \eqref{wilsonpart10} can be expressed as\footnote{Remember $S A S^{-1}=A^*=\left(A^{-1}\right)^T$ for a matrix $A$.}
\begin{equation}
\label{wilsonpart12}
	A''={P_2}^{T} \left( {P_2}^*\begin{pmatrix} n'_1 & 0 \\ n'_2 & 0 \end{pmatrix} + \omega_p\,\omega_w\, S^{-1}\cdot  D_2 \cdot {Q_2}^{-1} \cdot \begin{pmatrix} 0 & m_1 \\ 0 & m_2 \end{pmatrix}\right)\,.
\end{equation}
Again, we redefine the variables of the momentum- and the winding-lattice, as well as the $U'''$-modulus,
\begin{align}
\label{wilsonpart13}
	&U'''':=\left({P_2}^T\right)^\sharp U'''\,,\\
\label{wilsonpart14}
	&\begin{pmatrix} n''_1 \\ n''_2 \end{pmatrix}:={P_2}^* \cdot \begin{pmatrix} n'_1 \\ n'_2 \end{pmatrix}\quad \text{and}\\
\label{wilsonpart15}
	&\begin{pmatrix} m'_1 \\ m'_2 \end{pmatrix}:={Q_2}^{-1} \cdot \begin{pmatrix} m_1 \\ m_2 \end{pmatrix}\,.
\end{align}
If we now define $\omega_p\,\omega_w\, D_2=:\diag(\gamma,\delta)$ we arrive, finally, at
\begin{equation}
\label{wilsonpart16}
	A'''=\begin{pmatrix} n''_1 & 0 \\ n''_2 & 0 \end{pmatrix} + S^{-1}\cdot  \begin{pmatrix} \gamma & 0 \\ 0 & \delta \end{pmatrix} \cdot \begin{pmatrix} 0 & m'_1 \\ 0 & m'_2 \end{pmatrix}\,.
\end{equation}
Therefore, there exist $\gamma,\delta\in\mathds{Q}$ such that the partition function for non-vanishing Wilson lines can be written as
\begin{equation}
\label{YEHAW}
	\begin{split}
	Z^{\text{WL}}(\tau)=&\sum_{\genfrac{}{}{0pt}{2}{n''_1,n''_2\in\mathds{Z} }{m'_1,m'_2\in\mathds{Z}}}\ex{2\pi \I \,\tau(\gamma\, m'_1 n''_1+\delta\, m'_2 n''_2 )}\\&\times\exp\left[-\frac{\pi\,\tau_2}{T'_2\,U''''_2}\left| T' U'''' \,n_2+T' \,n_1 -\gamma\,m_1\,U''''+\delta\,m_2\right|^2\right]\,.
	\end{split}
\end{equation}

We will look at a more general case. We assume that the most general form of the partition function reads after Poisson resummation (omitting the primes)
\begin{equation}
\label{7}
	\tau_2\,Z_{\left(1,\theta^{l_k}\right)}^{1-\text{loop}}(\tau)=\sum_{A\in\sumset{\alpha}{\beta}{\delta}{\gamma}{\mathds{M}}}\ex{-2\pi \I \,T\det{A}}\,\frac{T_2}{\gamma \delta}\,\exp\left[-\frac{\pi\,T_2}{\tau_2\,U_2}\left| \begin{pmatrix} 1 & U \end{pmatrix} A \begin{pmatrix} \tau\\1 \end{pmatrix}\right|^2\right]\,,
\end{equation}
where
\begin{equation}
\label{genpartf}
\sumset{\alpha}{\beta}{\delta}{\gamma}{\mathds{M}}=\left\{\left.\begin{pmatrix} \alpha n_1 & \frac{1}{\gamma} l_1 \\ \beta n_2 & \frac{1}{\delta} l_2 \end{pmatrix} \right|n_1,\,n_2,l_1,\,l_2\in\mathds{Z}\right\}\quad \text{with} \quad \alpha,\,\beta,\,\gamma,\,\delta\in\mathds{Q}\,.
\end{equation}
This expression allows us to examine the set of possible symmetries $\Gamma'$ of the partition function which is associated with $(1,\theta^{l_k})$. In principle, the concrete form of the symmetry won't be necessary for the computation of one-loop gauge threshold corrections. The only fact which will be of importance is that there \emph{exists} a symmetry $\Gamma'$. We will prove its existence by constructing it explicitly.

Let us examine a modular transformation $P\in\Gamma$ on the variable $\tau$ in \eqref{7} and let $P$ as a transformation on $\tau$ be represented as 
\begin{equation}
\label{36}
	P_\tau = \begin{pmatrix} a & b \\ c & d \end{pmatrix}\quad \text{with} \quad ad-bc=1\,.
\end{equation}
Then the associated modular transformation is given by
\begin{equation}
\label{35}
	P_{\tau}\,\left(\tau_2\,Z_{\left(1,\theta^{l_k}\right)}^{1-\text{loop}}(\tau)\right):=\frac{\tau_2}{\left|c \tau + d\right|^2}\,Z_{\left(1,\theta^{l_k}\right)}^{1-\text{loop}}\left(\frac{a\tau+b}{c\tau+d}\right)=
\end{equation}
\[
	\quad=\sum_{A\in\sumset{\alpha}{\beta}{\delta}{\gamma}{\mathds{M}}}\ex{-2\pi \I \,T\det{A}}\,\frac{T_2}{\gamma \delta}\,\exp\left[-\frac{\pi\,T_2\,\left|c \tau + d\right|^2}{\tau_2\,U_2}\left| \begin{pmatrix} 1 & U \end{pmatrix} A \begin{pmatrix} \frac{a\tau+b}{c\tau+d}\\1 \end{pmatrix}\right|^2\right]=
\]
\[
	\quad=\sum_{A\in\sumset{\alpha}{\beta}{\delta}{\gamma}{\mathds{M}}}\ex{-2\pi \I \,T\det{A\,P}}\,\frac{T_2}{\gamma\delta}\,\exp\left[-\frac{\pi\,T_2}{\tau_2\,U_2}\left| \begin{pmatrix} 1 & U \end{pmatrix} A \begin{pmatrix} a & b \\ c & d \end{pmatrix} \begin{pmatrix} \tau\\1 \end{pmatrix}\right|^2\right]=
\]
\[
	\quad=\sum_{A\in\sumset{\alpha}{\beta}{\delta}{\gamma}{\mathds{M}}\cdot P}\ex{-2\pi \I \,T\det{A}}\,\frac{T_2}{\gamma \delta}\,\exp\left[-\frac{\pi\,T_2}{\tau_2\,U_2}\left| \begin{pmatrix} 1 & U \end{pmatrix} A \begin{pmatrix} \tau\\1 \end{pmatrix}\right|^2\right]\,.
\]
Therefore, a modular transformation on $\tau$ is equivalent to a multiplication of any matrix of $\sumset{\alpha}{\beta}{\delta}{\gamma}{\mathds{M}}$ with $P$ from the right, which means that we replace the domain of summation $\sumset{\alpha}{\beta}{\delta}{\gamma}{\mathds{M}}$ by $\sumset{\alpha}{\beta}{\delta}{\gamma}{\mathds{M}}\cdot P$. Hence, the invariance of the partition function $\tau_2\,Z_{\left(1,\theta^{l_k}\right)}^{1-\text{loop}}(\tau)$ is fulfilled if
\begin{equation}
\label{37}
	\sumset{\alpha}{\beta}{\delta}{\gamma}{\mathds{M}}\cdot P=\sumset{\alpha}{\beta}{\delta}{\gamma}{\mathds{M}}\,.
\end{equation}
Let us write
\begin{equation}
	\begin{pmatrix} \alpha & \gamma \\ \beta & \delta \end{pmatrix} = \begin{pmatrix} \sumfrac{\alpha} & \sumfrac{\gamma} \\ \sumfrac{\beta} & \sumfrac{\delta} \end{pmatrix}\,,
\end{equation}
with $\gcd(u_\alpha,v_\alpha)=\gcd(u_\beta,v_\beta)=\gcd(u_\gamma,v_\gamma)=\gcd(u_\delta,v_\delta)=1$. Then, the above mentioned multiplication reads 
\begin{equation}
\label{38}
\begin{split}
	\begin{pmatrix} \sumfraca{\alpha}\, n_1 & \sumfracb{\gamma}\, l_1 \\ \sumfraca{\beta}\, n_2 &  \sumfracb{\delta}\, l_2 \end{pmatrix} \cdot \begin{pmatrix} a & b \\ c & d \end{pmatrix}=
	\begin{pmatrix} 
		\sumfraca{\alpha}\left(n_1 a+\frac{v_\alpha\, v_\gamma}{u_\alpha\,u_\gamma}l_1 c\right) & \sumfracb{\gamma} \left( \frac{u_\gamma\,u_\alpha}{v_\gamma\,v_\alpha}n_1 b+ l_1 d\right)\\
		\sumfrac{\beta}\left(n_2 a+\frac{v_\beta\, v_\delta}{u_\beta\,u_\delta}l_2 c\right) & \sumfracb{\delta} \left( \frac{u_\delta\,u_\beta}{v_\delta\,v_\beta}n_2 b+ l_2 d\right)
	\end{pmatrix}
\end{split}
\end{equation}
and, therefore,
\begin{align}
\label{39}
	&\nu:=\left.\lcm\left(\frac{v_\gamma\,v_\alpha}{\gcd (u_\gamma\,u_\alpha,v_\gamma\,v_\alpha)},\frac{v_\delta\,v_\beta}{\gcd (u_\delta\,u_\beta,v_\delta\,v_\beta)}\right)\right|b \quad \text{and}\\
	&\mu:=\left.\lcm\left(\frac{u_\alpha\,u_\gamma}{\gcd (u_\alpha\,u_\gamma,v_\alpha\, v_\gamma)},\frac{u_\beta\,u_\delta}{\gcd (u_\beta\,u_\delta,v_\beta\, v_\delta)}\right) \right|c\,.
\end{align}
Thus, we can read off
\begin{equation}
\label{40}
	\Gamma'=\Gamma(\mu,\nu):=\Gamma_0(\mu)\cap\Gamma^0(\nu)\,.
\end{equation}
We will compute the integrals \eqref{4} (with the partition function given in \eqref{7}) by successively reducing the cases of $(\alpha,\beta,\gamma,\delta)$ via rescaling of the moduli $T$ and $U$. Firstly, we will reduce the case of $(\alpha,\beta,\gamma,\delta)\in\mathds{Q}^4$ to $(1,1,\tilde\gamma,\tilde\delta)$ with $\tilde\gamma,\,\tilde\delta\in\mathds{Q}$. This will turn out to be equivalent to $(1,1,\lambda\bar \gamma,\lambda\bar \delta)$ with $\lambda\in\mathds{Q}$ and $\bar\gamma,\bar \delta\in\mathds{Z}$. Secondly, we will reduce the case $(1,1,\lambda\bar\gamma,\lambda\bar\delta)$, if $\gamma<\delta$, to $(1,1,\lambda,\lambda\bar\gamma\bar\delta)$ and, if $\gamma>\delta$, to $(1,1,\lambda\bar\gamma\bar\delta,\lambda)$ . Thirdly, we will reduce $(1,1,\lambda,\lambda\bar\gamma\bar\delta)$ to $(1,1,1,\bar\gamma\bar\delta)$ and $(1,1,\lambda\bar\gamma\bar\delta,\lambda)$ to $(1,1,\bar\gamma\bar\delta,1)$. All these reductions will require different lines of reasoning.

Altogether, we can state that the problem of computing one-loop gauge threshold corrections in heterotic string theory compactified on an arbitrary (abelian, toroidal) orbifold (allowing for arbitrary discrete Wilson lines), is solved by evaluating integrals of the form
\begin{equation}
\label{master}
\begin{split}
	\sumset{\alpha}{\beta}{\delta}{\gamma}{I}(T,U)=\sumset{\alpha}{\beta}{\delta}{\gamma}{A}\,&\int_{R_{\Gamma'}}\hypmes{\tau}\,\sum_{A\in\sumset{\alpha}{\beta}{\delta}{\gamma}{\mathds{M}}}\ex{-2\pi \I \,T\det{A}}\,\frac{T_2}{\gamma \delta}\,\\&\times\exp\left[-\frac{\pi\,T_2}{\tau_2\,U_2}\left| \begin{pmatrix} 1 & U \end{pmatrix} A \begin{pmatrix} \tau\\1 \end{pmatrix}\right|^2\right]-\int_{R_\Gamma}\hypmes{\tau} \tau_2\,,
\end{split}
\end{equation}
with $\sumset{\alpha}{\beta}{\delta}{\gamma}{A}$ being a constant that forces \eqref{master} to be finite. This constant will be determined later. 

In the remainder of the work we will, therefore, be concerned in developing a method of how to deal with a generic integral of the type \eqref{master}.
\section{Computation of Threshold Corrections}\label{chapfour}
In this section we will compute the Integrals of the form \eqref{master}. We will begin this issue by computing the case $\alpha=\beta=\gamma=1$ and $\delta\in\mathds{Z}$. Afterwards, we will show that this result is already enough to compute threshold corrections, because it is possible to reduce the general case $\alpha$, $\beta$, $\gamma$, $\delta\in\mathds{Q}$ to the case $\alpha=\beta=\gamma=1$ and $\delta\in\mathds{Z}$. This will be done in three steps.

\subsection{Computation in case of $\alpha=\beta=\gamma=1$ and $\delta\in\mathds{Z}$}
Let us begin with the computation of $\sumset{1}{1}{\delta}{1}{I}$. In the case under consideration we can deduce from equations \eqref{39} to \eqref{40} that the symmetry  of $\tau_2 Z_{(1,\theta^{l_k})}^{\text{one-loop}}$ is given by $\Gamma'=\Gamma_0(\delta)$ and that the integral is given by
\begin{equation}
\label{102}
\begin{split}
	\sumset{1}{1}{\delta}{1}{I}(T,U)&=\sumset{1}{1}{\delta}{1}{A}\,\int_{R_{\Gamma_0(\delta)}}\hypmes{\tau}\,\sum_{A\in\sumset{1}{1}{\delta}{1}{\mathds{M}}}\ex{-2\pi i \,T\det{A}}\,\frac{T_2}{\delta}\,\\&\times\exp\left[-\frac{\pi\,T_2}{\tau_2\,U_2}\left| \begin{pmatrix} 1 & U \end{pmatrix} A \begin{pmatrix} \tau\\1 \end{pmatrix}\right|^2\right]-\int_{R_\Gamma}\hypmes{\tau} \tau_2\,.
\end{split}
\end{equation}
Furthermore, let us parameterize the matrices $A\in\sumset{1}{1}{\delta}{1}{\mathds{M}}$ as
\begin{equation}
\label{103}
	A=\begin{pmatrix} n_1 & l_1 \\ n_2 & \frac{1}{\delta}\,l_2 \end{pmatrix}\,.
\end{equation}
The main result of this subsection is given by
\begin{thm}
\label{main}
Let $\sumset{1}{1}{\delta}{1}{I}$ be given by \eqref{102}. Then it holds that
\begin{equation}
\label{main1}
	\sumset{1}{1}{\delta}{1}{I}(T,U)=-\sumset{1}{1}{\delta}{1}{A}\,\sum_{d|\delta}\sumset{1}{1}{\delta}{1}{C}(d) \left[\f{}{d}+\ln\left(\frac{8\,\pi\,\ex{1-\gamma_E}}{3\sqrt{3}}\right)\right]\,,
\end{equation}
where
\begin{equation}
\label{main2}
	\sumset{1}{1}{\delta}{1}{C}(d)=\prod_{p|d\,\wedge\,p\left|\frac{\delta}{d}\right.\wedge\,p\in\mathds{P}}\left(1-\frac{1}{p}\right)\quad \text{and} \quad \sumset{1}{1}{\delta}{1}{A}=\left(\sum_{d|\delta}\sumset{1}{1}{\delta}{1}{C}(d)\right)^{-1}\,.
\end{equation}
$\eta$ denotes Dedekind's eta function.
\end{thm}

\begin{remarks}$ $
\begin{itemize}
\item Note that in the sums above, $d$ runs over \emph{all} divisors of $\delta$, not only its prime factors. Or in other words, $d$ runs over all possible integer numbers which can be constructed from $\delta$ by omitting one of its prime factors. In particular, the set of all divisors $d$ always includes $\delta$ and $1$.
\item The product in the definition of $\sumset{1}{1}{\delta}{1}{C}(d)$ is understood to give $1$ if there exists no $p$ fulfilling the condition $p|d\,\wedge\,p\left|\frac{\delta}{d}\right.\wedge\,p\in\mathds{P}$.
\item Since $\ln\left(\frac{8\,\pi\,\ex{1-\gamma_E}}{3\sqrt{3}}\right)$ is independent of $d$, we could have also evaluated the sum over it:
\[
\sumset{1}{1}{\delta}{1}{A}\,\sum_{d|\delta}\sumset{1}{1}{\delta}{1}{C}(d) \ln\left(\frac{8\,\pi\,\ex{1-\gamma_E}}{3\sqrt{3}}\right) = \ln\left(\frac{8\,\pi\,\ex{1-\gamma_E}}{3\sqrt{3}}\right)
\]
by definition of $\sumset{1}{1}{\delta}{1}{A}$. However, the given form has a more apparent relation to the case $\delta=1$ and we chose it for later convenience.
\end{itemize}
\end{remarks}

The proof of the theorem requires several lemmas, which we have to develop first. They will provide the technical tools to compute \eqref{102}. The concepts we will be using can be summarised as
\begin{enumerate}
\item The domain of integration is the fundamental domain of a modular subgroup.
\item A matrix multiplication $A\cdot P$ corresponds to a modular transformation $P_\tau$ on $\tau$.
\item Imposing divisibility of $l_2$ by a prime factor of $\delta$ results in a reduction $\delta\mapsto\delta'$ with $\delta'\,|\,\delta$.
\item Imposing non-divisibility of the reduced $l_2$ by a prime factor of $\delta'$ results in an integral which can be computed directly using the reference integral.
\item There exists a reference integral to which everything can be traced back.
\end{enumerate}
These points should be clarified. 

The integrals we wish to compute look similar to the ones solved in \cite{Dixon_et_al2}, which correspond to the case $\alpha=\beta=\gamma=\delta=1$. However, several kinds of new problems arise.

Firstly, the domain of integration is no longer the fundamental domain of $\Gamma$ but the fundamental domain of a subgroup $\Gamma_0(\delta)$. This problem already occurs in models with vanishing discrete Wilson lines, where the fixed planes do not lie in a two-dimensional sub-torus of the torus lattice. Thresholds for such models were considered in \cite{Stieberger4}, where the cases $(\alpha,\beta,\gamma,\delta)=(1,1,2,2)$, $(1,1,1,3)$ and $(1,1,2,1)$ occur. There, the problem with the domain of integration was solved by explicitly constructing the necessary fundamental domains $R_{\Gamma_0(2)}$, $R_{\Gamma_0(3)}$. In more general cases (in particular if $\delta\,{\not\in}\,\mathds{P}$), the construction of $R_{\Gamma_0(\delta)}$ is a difficult and complicated task. Furthermore, knowledge of the value of $\delta$ is necessary and the computation of \eqref{102} would have to be performed individually for every model at hand. Therefore, we choose
to use a method which only applies the defining properties of a fundamental domain (point 1 above):
\begin{enumerate}
\item $\Gamma_0(\delta)\,R_{\Gamma_0(\delta)}=\mathds{H}^+\;$ and
\item $\forall\tau_1,\tau_2\in R_{\Gamma_0(\delta)}:\forall P\in\Gamma_0(\delta):P\,\tau_1\not=\tau_2$
\end{enumerate}
Secondly, we will use that multiplying $A$ in \eqref{102} by a matrix $P$ from the right can be reinterpreted as a modular transformation $P_\tau$ acting on $\tau$ (point 2 above). Using this fact, the set of matrices we sum over can be decomposed into orbits under a modular subgroup, giving contributions of the zero matrix, the set of matrices with non-zero determinant and the set of non-zero matrices with vanishing determinant.

Thirdly, in \cite{Dixon_et_al2} the summation has to run over \emph{all} integer matrices, while in our case the summation runs over \emph{special} matrices fulfilling certain divisibility conditions. This causes a naive application of the method in \cite{Dixon_et_al2} to fail for our case. However, in the present work we will develop several methods to express the sums over matrices with divisibility conditions as sums without divisibility conditions (points 3 and 4 above). Hence, we can trace back the most general case to the case solved in reference \cite{Dixon_et_al2} ($\alpha=\beta=\gamma=\delta=1$) which provides the mentioned reference integral (point 4 above) and will act as a building block in the final result. 

The first lemma that we will need to prove theorem \ref{main} provides the reference integral, in particular the case $\alpha=\beta=\gamma=\delta=1$. It's proof can be found in the literature and we won't repeat it here.
\begin{lem}
\label{refint}
Let $\sumset{1}{1}{1}{1}{I}$ be given by \eqref{102} with $\delta=1$. Then it holds
\begin{equation}
\begin{split}	
	\sumset{1}{1}{1}{1}{I}(T,U)&=f_1(T,U)+f_2(T,U)+f_3(T,U)=\\&=-\ln\left(\frac{8\pi\,\ex{1-\gamma_E}}{3\sqrt{3}}\right)-\ln\left(T_2\,\left|\eta\left(T\right)\right|^4\,U_2\,\left|\eta\left(U\right)\right|^4\right)\,,
\end{split}
\end{equation}
with
\begin{align}
	&f_1= \int_{R_{\Gamma}}\frac{d^2\tau}{{\tau_2}^2} T_2\,,\\
	&f_2= \sum_{\genfrac{}{}{0pt}{2}{0\leq j<k }{p\not = 0}} T_2\,\ex{-2\pi i \,T\cdot k p}\,\int_{-\infty}^{\infty}d\tau_1\int_0^{\infty}\frac{d\tau_2}{{\tau_2}^2}\exp\left(-\frac{\pi\, T_2}{\tau_2\,U_2}|k\tau+j+p\,U|^2\right)\quad \text{and}\\
	&f_3= \int_{-1/2}^{+1/2}d\tau_1\int_0^\infty\frac{d\tau_2}{{\tau_2}^2}\left[T_2\sideset{}{'}\sum_{j,p}\exp\left(-\frac{\pi\,T_2}{\tau_2\,U_2}|j+Up|^2\right)-\tau_2\theta_{R_\Gamma}(\tau)\right]\,,
\end{align}
for all $T\,,U\in\mathds{H}^+$.
\end{lem}
\begin{proof} The proof of this statement can be found in reference \cite{Dixon_et_al2}.
\end{proof}

\begin{remarks}$ $
\begin{itemize}
\item The $f_i$ decode the different contributions of orbits to $\sumset{1}{1}{1}{1}{I}(T,U)$. 

\item The first integral is given by the contribution of the zero matrix, the second by all matrices $A$ with non-vanishing determinant, which can be written as
\begin{equation}
	A\in \begin{pmatrix} k & j \\ 0 & p \end{pmatrix} \cdot \Gamma\,,
\end{equation} 
and the third integral by all non-zero matrices $A'$ with vanishing determinant, which can be written as
\begin{equation}
	A'\in \begin{pmatrix} 0 & j \\ 0 & p \end{pmatrix} \cdot \Gamma \,.
\end{equation}
\item The results for $f_1$, $f_2$ and $f_3$ read \cite{Dixon_et_al2}
\begin{align}
&f_1=\frac{\pi}{3}T_2=-4\text{Re}\,\ln\ex{2\pi i\,T\,\frac{1}{24}}\,, \\
&f_2= -4\text{Re} \ln\prod_{n=1}^{\infty} \left(1-{q_T}^n\right)\,,\quad \text{where} \quad q_T=\ex{2\pi i\,T}\,, \quad \text{and}\\
&f_3=-4\text{Re}\ln \eta(U)-\ln(T_2\,U_2)-\ln\left(\frac{8\pi\,\ex{1-\gamma_E}}{3\sqrt{3}}\right)\,.
\end{align}
\item Note that
\begin{equation}
	\sumset{1}{1}{\delta}{1}{I}(T,U)=-\sumset{1}{1}{\delta}{1}{A}\,\sum_{d|\delta}\sumset{1}{1}{\delta}{1}{C}(d)\,\sumset{1}{1}{1}{1}{I}(T/d,U/d) \;.
\end{equation}
\end{itemize}
\end{remarks}

Our method will rely on the following observation: If we look at a matrix
\[
\begin{pmatrix} n_1 & l_1 \\ n_2 & \frac{1}{\delta}\,l_2 \end{pmatrix}
\]
and run through all $l_2\in\mathds{Z}$ then, inevitably, we will hit all integer multiples of $\delta$: $l_2 = n\,\delta$ with a $n\in\mathds{Z}$. The contribution of all these matrices will look like the case $\delta=1$, which is our reference integral. 

However, what about the missing contributions? To incorporate them, it will be useful to look at all integer multiples of divisors of $\delta$. To see this, let $d|\delta$ and have a look at all matrices where $l_2 = n\,d$ with a $n\in\mathds{Z}$. Then, these matrices look like
\[
\begin{pmatrix} n_1 & l_1 \\ n_2 & \frac{d}{\delta}\,l_2 \end{pmatrix}=\begin{pmatrix} n_1 & l_1 \\ n_2 & \frac{1}{\delta/d}\,l_2 \end{pmatrix}
\]
so they look like the contribution of the smaller integer number $\delta/d$ instead of $\delta$. Hence, all the divisors of $\delta$ fractionise the matrix sum into matrix sums corresponding to smaller integer numbers, which we will call the \emph{reductions} of $\delta$. Of course, the method can be reapplied to those reductions again yielding smaller reductions. After finitely many steps we will get several contributions which look like the $\delta=1$ case.

Still, we are missing those integers $l_2$ which are not divisible by any prime factor of $\delta$. However, we can use an elementary trick to incorporate these in a similar manner as above. If the partition into prime factors of $\delta$ is given by
\begin{equation}
\label{primedecomp}
	\delta=\delta_1^{x_1}\cdot \delta_2^{x_2}\cdot\ldots\cdot\delta_n^{x_n}
\end{equation}
then the contribution of matrices not fulfilling any divisibility conditions is
\begin{equation}
\delta_1{\not|}\,l_2\wedge\ldots\wedge \delta_n{\not|}\,l_2
\;.
\end{equation}
However, we can write
\begin{equation}\label{idea123}
\delta_1\not|\,l_2\wedge\ldots\wedge \delta_n\not|\,l_2
=
(\text{all}\;l_2) \;\setminus\; \neg (\delta_1\not|\,l_2\wedge\ldots\wedge \delta_n\not|\,l_2)
=
(\text{all}\;l_2)
\;\setminus\;
\delta_1|\,l_2\vee\ldots\vee \delta_n|\,l_2
\;,
\end{equation}
where we used $\neg(A \wedge B) = (\neg A)\vee (\neg B)$. These contributions can then again be written in terms of contributions $l_2=n\,d$ with $d|\delta$.

Now, let us convert these ideas into practicable lemmas. Lemma \ref{red} will treat in detail the mentioned reduction of contributions $l_2=n\,d$ with $d|\delta$. Afterwards, we will formulate lemma \ref{choice} which establishes the precise connection of idea \eqref{idea123} to the sums we wish to compute. Finally, lemmas \ref{novanish} and \ref{vanish} show how to reformulate the contributions of the different orbits as sums of reductions of $\delta$.

We start with
\begin{lemdef}
\label{red}
Let
\begin{equation}
\sumset{1}{1}{\delta}{1}{\mathcal{I}}
:=
\int_{R_{\Gamma_0(\delta)}}\hypmes{\tau}\,\sum_{A\in\sumset{1}{1}{\delta}{1}{\mathds{M}}}\ex{-2\pi i \,T\det{A}}\,\frac{T_2}{\delta}\,\exp\left[-\frac{\pi\,T_2}{\tau_2\,U_2}\left| \begin{pmatrix} 1 & U \end{pmatrix} A \begin{pmatrix} \tau\\1 \end{pmatrix}\right|^2\right]
\end{equation}
and
\[
\sumset{1}{1}{d|\delta}{1}{\mathds{M}}:=\left\{\left.\begin{pmatrix} n_1 & l_1 \\ n_2 & \frac{1}{\delta}\,l_2 \end{pmatrix}\right|l_2=d\,l'_2\;\, n_1,\, n_2,\, l_1,\, l'_2\in\mathds{Z}\,\right\},
\]
where $d\,|\,\delta$. If we denote the contribution of all matrices $A\in\sumset{1}{1}{d|\delta}{1}{\mathds{M}}$ by $\sumset{1}{1}{d|\delta}{1}{\mathcal{I}}(T,U)$ it holds that
\begin{equation}
\label{red0}
	\sumset{1}{1}{d|\delta}{1}{\mathcal{I}}(T,U)=\frac{\left[\Gamma_0\left(\frac{\delta}{d}\right):\Gamma_0(\delta)\right]}{d}\,\sumset{1}{1}{\frac{\delta}{d}}{1}{\mathcal{I}}(T,U)
	\;,
\end{equation}	
with $\left[\Gamma_0\left(\frac{\delta}{d}\right):\Gamma_0(\delta)\right]$ being the index of $\Gamma_0(\delta)$ in $\Gamma_0\left(\frac{\delta}{d}\right)$.
\end{lemdef}
\begin{proof}
To prove this lemma, look at the contribution of all matrices with $l_2=d\,l'_2$ and $d\,|\,\delta$ to \eqref{102}. It is
\begin{equation}
\label{red1}
	A'=\begin{pmatrix} n_1 & l_1 \\ n_2 & \frac{d}{\delta} l'_2 \end{pmatrix}=\begin{pmatrix} n_1 & l_1 \\ n_2 & \frac{1}{\frac{\delta}{d}} l'_2 \end{pmatrix}\,.
\end{equation}
Since $\Gamma_0(\delta)\subsetneq \Gamma_0\left(\frac{\delta}{d}\right)$ it follows that $R_{\Gamma_0(\delta)}\supsetneq R_{\Gamma_0\left(\frac{\delta}{d}\right)}$ with
\begin{equation}
\label{red2}
	R_{\Gamma_0\left(\delta\right)}=\bigcup_{k=1}^{\left[\Gamma_0\left(\frac{\delta}{d}\right):\Gamma_0(\delta)\right]} M_k\,R_{\Gamma_0(\frac{\delta}{d})}
\end{equation}
and $M_k\in\Gamma_0\left(\frac{\delta}{d}\right)$. Inserting \eqref{red1} into \eqref{102} yields an expression, which is invariant under $\Gamma_0\left(\frac{\delta}{d}\right)$. Using this fact and factoring out $1/d$ of this expression yields equation \eqref{red0} and the lemma is proven.
\end{proof}
Now we make the idea \eqref{idea123} more precise and applicable:
\begin{lemdef}
\label{choice}
	Let $\delta={\delta_1}^{x_1}\cdot {\delta_2}^{x_2}\cdot\ldots\cdot{\delta_n}^{x_n}$ and $\delta^{(l)}:=\delta_{i_1}\cdot\delta_{i_2}\cdot\ldots\cdot\delta_{i_l}\,|\,\delta$ with $\delta_i\in\mathds{P}$, be the product of a choice of $l$ prime factors of $\delta$. Moreover, let $\mathcal{C}_l(\delta)$ be the set of all possible products of choices of $l$ prime factors of $\delta$ and $f$ some function on $\mathds{Z}$. We define $\mathcal{C}_0(\delta):=\{1\}$ and $\mathcal{C}_l(1):=\{1\}$ for all $l$ and $\delta$. Then it holds, at least as a formal sum, that
	\begin{equation}
	\begin{split}
	\label{105}
		\sum_{\delta_1|k\,\vee\ldots\vee\,\delta_n|k}f(k)=\sum_{l=1}^{n}(-1)^{l+1}\,\sum_{\delta^{(l)}\in\,\mathcal{C}_l(\delta)}\sum_{k\in\mathds{Z}} f(\delta^{(l)}\cdot k)\;.
\end{split}	
\end{equation}
\end{lemdef} 
\begin{proof}
	We want to show that every number $k$ which fulfills $\delta_1|k\vee\ldots\vee\delta_n|k$ has been counted once and only once in \eqref{105}. To that end let us consider an arbitrary choice of prime factors $\delta^{(l)}=\delta_{i_1}\cdot\delta_{i_2}\cdot\ldots\cdot\delta_{i_l}$. 
How many times has a number $k$ which is divisible by $\delta^{(l)}$, i.e.\ with prime factorization $k={\delta_{i_1}}^{j_1+1}\cdot{\delta_{i_2}}^{j_2+1}\cdot\ldots\cdot{\delta_{i_l}}^{j_l+1}\cdot k'$,  where all $\delta_i\not| \, k'$, been counted by the right-hand sight of equation $\eqref{105}$? 
It has been counted $\binom{l}{1}$ times by summing over all multiples of $\delta_{i_1}$,all multiples of $\delta_{i_2}$,... and all multiples of $\delta_{i_l}$. By subtracting all multiples of $\delta_{i_1}\delta_{i_2}$, all multiples of $\delta_{i_1}\delta_{i_3}$,... and all multiples of $\delta_{i_{n-1}}\delta_{i_n}$ it has been counted $-\binom{l}{2}$ times. In this manner $k$ has been counted $(-1)^{l'+1}\,\binom{l}{l'}$ times by a choice of $l'$ prime factors of $\delta^{(l)}$. Altogether $k$ has been counted
\begin{equation}
\label{106}
	\binom{l}{1}-\binom{l}{2}+\ldots+(-1)^{l+1} \binom{l}{l}=1
\end{equation}
times. Since \eqref{106} holds for arbitrary choices $\delta^{(l)}$, lemma \ref{choice} is proven. 
\end{proof}

To clarify this lemma we consider an example. Let us assume that $\delta=2\cdot 3 \cdot 5^2=150$ and, therefore, we want to compute a sum of some function $f(k)$ over $2|k\vee 3|k\vee 5|k$. Lemma \ref{choice} states that this sum is given by
\begin{equation}
\begin{split}
	\sum_{2|k\,\vee\, 3|k\,\vee\, 5|k}f(k)=&\sum_{k}f(2k) + \sum_k f(3k)+\sum_k f(5k)-\\-&\sum_k f(6k)-\sum_k f(10k)-\sum_k f(15k)+\\+&\sum_k f(30k)\,.
\end{split}
\end{equation}
It should be mentioned that $\mathcal{C}_l(\delta)$ only depends on the prime numbers which divide $\delta$. As an example consider $30=2\cdot 3\cdot 5$ and $150 = 2 \cdot3 \cdot 5^2$, which results in
\[
\begin{split}
	&\mathcal{C}_0(30)=\mathcal{C}_0(150)=\{1\} \\
	&\mathcal{C}_1(30)=\mathcal{C}_1(150)=\{2,3,5\} \\
	&\mathcal{C}_2(30)=\mathcal{C}_2(150)=\{6,10,15\} \\
	&\mathcal{C}_3(30)=\mathcal{C}_3(150)=\{30\}
\end{split}
\] 
To give a procedure to sum over all matrices which satisfy $\delta_1{\not{|}}\,l_2\wedge\ldots\wedge \delta_n{\not|}\,l_2$, we formulate two lemmas. The first one enables us to sum over all matrices with non-vanishing determinant and the second one deals with all non-zero matrices with vanishing determinant. 
\begin{lem}
\label{novanish}
Let $\delta={\delta_1}^{x_1}\cdot {\delta_2}^{x_2}\cdot\ldots\cdot{\delta_n}^{x_n}$ and $\delta^{(l)}=\delta_{i_1}\cdot\delta_{i_2}\cdot\ldots\cdot\delta_{i_l}$, with $\delta_i\in\mathds{P}$, be the product of a choice of $l$ prime factors of $\delta$. Moreover, let $\mathcal{C}_l(\delta)$ be the set of all possible products of choices of $l$ prime factors of $\delta$ and $f$ some function on $\mathds{Z}$. Then it holds, at least as a formal sum, that
\begin{equation}
\label{107}
\begin{split}
	\sum_{\genfrac{}{}{0pt}{2}{A\in\sumset{1}{1}{\delta}{1}{\mathds{M}}\wedge \det A \not =0}{\delta_1\not{\,|\,}l_2\wedge\ldots\wedge \delta_n\not{\,|\,}\,l_2}}f(A)=&\sum_{l=0}^{n}(-1)^{l}\,\sum_{\delta^{(l)}\in\,\mathcal{C}_l(\delta)}\sum_{\genfrac{}{}{0pt}{2}{k,j,p\in\mathds{Z}}{k>j\ge0\wedge p\not = 0}}\sum_{P\in\Gamma_0(\delta)}f\left( \begin{pmatrix} k & j \\ 0 & \frac{\delta^{(l)}}{\delta}\,p \end{pmatrix}\cdot P \right)\,.
\end{split}
\end{equation}
\end{lem}
\begin{proof}
To prove this lemma, we use an immediate corollary of lemma \ref{choice}. With the same assumptions it is\footnote{Using $\sum\limits_{\delta_1\not{\,|\,}p\wedge\ldots\wedge \delta_n\not{\,|\,}\,p}\,\ldots=\sum\limits_{p\in\mathds{Z}}\,\ldots - \sum\limits_{\delta_1{\,|\,}p\vee\ldots\vee\delta_n{\,|\,}\,p}$\,\ldots}
\begin{equation}
	\label{108}
		\sum_{\genfrac{}{}{0pt}{2}{\delta_1\not{\,|\,}p\wedge\ldots\wedge \delta_n\not{\,|\,}\,p}{p\not =0}}F(p)=\sum_{l=0}^{n}(-1)^{l}\,\sum_{\delta^{(l)}\in\,\mathcal{C}_l(\delta)}\sum_{\genfrac{}{}{0pt}{2}{p\in\mathds{Z}}{p\not= 0}} F(\delta^{(l)}\cdot p)\,.
\end{equation}
Let us define
\begin{equation}
\label{109}
	F(p):=\sum_{\genfrac{}{}{0pt}{2}{k,j\in\mathds{Z}}{k>j\ge 0}}\sum_{P\in\Gamma_0(\delta)}f\left( \begin{pmatrix} k & j \\ 0 & \frac{p}{\delta} \end{pmatrix}\cdot P\right)\,.
\end{equation}
and parameterise $P$ according to
\begin{equation}
\label{110}
	P=\begin{pmatrix} a & b \\ \delta c & d \end{pmatrix}\,.
\end{equation}
Then \eqref{108} implies 
\begin{equation}
\label{111}
	\sum_{\genfrac{}{}{0pt}{2}{\delta_1\not{\,|\,}p\wedge\ldots\wedge \delta_n\not{\,|\,}\,p}{p\not =0}}F(p)=\sum_{l=0}^{n}(-1)^{l}\,\sum_{\delta^{(l)}\in\,\mathcal{C}_l(\delta)}\sum_{\genfrac{}{}{0pt}{2}{k,j,p\in\mathds{Z}}{k>j\ge0\wedge p\not = 0}}\sum_{P\in\Gamma_0(\delta)}f\left( \begin{pmatrix} k & j \\ 0 & \frac{\delta^{(l)}}{\delta}\,p \end{pmatrix}\cdot P \right)
\end{equation}
Thus, we have to show
\begin{equation}
\label{112}
	\sum_{\genfrac{}{}{0pt}{2}{\delta_1\not{\,|\,}p\wedge\ldots\wedge \delta_n\not{\,|\,}\,p}{p\not =0}}F(p)=\sum_{\genfrac{}{}{0pt}{2}{A\in\sumset{1}{1}{\delta}{1}{\mathds{M}}\wedge \det A \not=0}{\delta_1\not{\,|\,}l_2\wedge\ldots\wedge \delta_n\not{\,|\,}\,l_2}}f(A)\,.
\end{equation}
This is equivalent to prove that every matrix $A\in\sumset{1}{1}{\delta}{1}{\mathds{M}}$, with $\delta_1{\not{|}}\,l_2\wedge\ldots\wedge \delta_n{\not{|}}\,l_2$ and $\det A\not= 0$, can be uniquely decomposed as
\begin{equation}
\label{113}
	A=A_0\cdot P=\begin{pmatrix} n_1 & l_1 \\ n_2 & \frac{1}{\delta} \, l_2 \end{pmatrix}=\begin{pmatrix} k & j \\ 0 & \frac{1}{\delta}\,p \end{pmatrix}\cdot \begin{pmatrix} a & b \\ \delta \, c & d\end{pmatrix}=\begin{pmatrix} k\,a + \delta\,j\,c & k\,b + j\,d \\ p\,c & \frac{1}{\delta}\,p\,d\end{pmatrix}
\end{equation}
Since $a\,d-\delta\,b\,c=1$, we can infer that (using the lemma of B\'{e}zout) $\delta_i {\not|}\, d\,$ for $\delta_i|\delta$ and $\delta_i\in\mathds{P}$. Hence, it holds that $\delta_i{\not|}\, l_2 $ for $\delta_i|\delta$ and $\delta_i\in\mathds{P}$. To solve \eqref{113}, we multiply it with $P^{-1}$ from the right. This results in the system of Diophantine equations
\begin{equation}
\label{114}
	\begin{pmatrix} k & j \\ 0 & \frac{1}{\delta} \, p \end{pmatrix} =\begin{pmatrix} n_1 & l_1 \\ n_2 & \frac{1}{\delta} l_2 \end{pmatrix}\cdot \begin{pmatrix} d & -b \\ -\delta\,c & a  \end{pmatrix}=\begin{pmatrix} n_1\,d-l_1\,\delta\,c & l_1\,a- n_1\,b \\ n_2\,d-l_2\,c & \frac{1}{\delta}\,\left(l_2\,a-\delta\,n_2\,b\right) \end{pmatrix}\,,
\end{equation}
with $a\,d-\delta\,b\,c=1$. Let us look at the two-one component of \eqref{114}. It states
\begin{equation}
\label{115a}
	n_2\,d-l_2\,c=0\,.
\end{equation}
Because of $\gcd(\delta,l_2)=1$ it follows that $\gcd(n_2,l_2)=\gcd(\delta\,n_2,l_2)$ and the most general integral solution of \eqref{115a} is given by
\begin{equation}
\label{115}
	d=\epsilon\,\frac{l_2}{\gcd(\delta \, n_2,l_2)} \quad \text{and} \quad c=\epsilon\,\frac{n_2}{\gcd(\delta \, n_2,l_2)}\,,\quad \text{for} \quad \epsilon\in\mathds{Z}\,.
\end{equation}
Inserting equation \eqref{115} into $a\,d-\delta\,b\,c=1$ yields
\begin{equation}
\label{116}
	l_2\,a-\delta\, n_2\,b=\frac{\gcd(\delta\,n_2,l_2)}{\epsilon}
\end{equation}
Equation \eqref{116} is solvable if and only if
\begin{equation}
\label{117}
	\gcd(\delta\,n_2,l_2)\left|\frac{\gcd(\delta\,n_2,l_2)}{\epsilon}\right.\Leftrightarrow \exists k\in\mathds{Z}\,:\,k\,\gcd(\delta\,n_2,l_2)=\frac{\gcd(\delta\,n_2,l_2)}{\epsilon}
\end{equation}
\[
	\Leftrightarrow k\,\epsilon = 1 \Leftrightarrow k=\epsilon=\pm 1
\]
Putting this into \eqref{116} leads to
\begin{equation}
\label{118}
	l_2\,a-\delta\,n_2\,b=\epsilon\,\gcd(\delta\,n_2,l_2)
\end{equation}
The lemma of B\' ezout tells us that the most general solution of \eqref{118} is given by
\begin{align}
\label{119}
	&a_\xi=a-\epsilon\,\xi\,\frac{\delta\,n_2}{\gcd(l_2,\delta\,n_2)}\quad \text{and}\\
\label{120}
	&b_\xi=b-\epsilon\,\xi\,\frac{l_2}{\gcd(l_2,\delta\,n_2)}\quad \text{for} \quad \xi\in\mathds{Z}\;.
\end{align}
Where $(a,b)$ is a special solution of \eqref{118} which can, for example, be gained by the extended Euclidean algorithm. The other solutions $k$, $j$ and $p$ are determined by \eqref{114} and read
\begin{align}
\label{121}
	\begin{split}
	k&= n_1\,d-l_1\,\delta\,c=n_1\,\epsilon\,\frac{l_2}{\gcd(\delta\,n_2,l_2)}-\delta\,l_1\,\epsilon\,\frac{n_2}{\gcd(\delta\,n_2,l_2)}=\\
	&=\delta\,\epsilon\,\frac{\det A}{\gcd(\delta\,n_2,l_2)}\,,
	\end{split}\\
	j&= l_1\,a-n_1b+\xi\,k \quad \text{and}\\
	\label{122}
	p&=\epsilon\,\gcd(\delta\,n_2,l_2)\not=0\,.
\end{align}
These numbers are integral by construction. Therefore, we have proven the existence of $A_0$ and $P$ in \eqref{113}. Next, we want to show their uniqueness if we require $k>j\ge 0$.

Let now $k>j\ge 0$. From $k>0$ we can infer $\epsilon=\text{sgn} \det{A}$. Then $j\ge 0$ implies together with equations \eqref{121} to \eqref{122}
\begin{equation}
\label{mc11}
 l_1\,a-n_1\,b+\xi k\ge 0\,\Rightarrow\, \xi \ge \frac{n_1\,b-l_1\,a}{\delta|\det A|}\gcd{(\delta\, n_2,l_2)}=: \xi_{\text{min}}\ .
\end{equation}
With $j<k$ it follows that
\begin{equation}
\label{mc12}
 l_1\,a-n_1\,b+\xi\,\ k<k\,\Rightarrow\,\xi<\frac{k+(n_1\,b-l_1\,a)}{k}=1+\xi_{\text{min}}\,.
\end{equation}
Hence,
\begin{equation}
\label{mc13}
 \xi_{\text{min}}\le\xi<\xi_{\text{min}}+1\,.
\end{equation}
Together with $\xi\in\mathds{Z}$ we deduce that
\begin{equation}
\label{mc14}
 \xi=\xi_0(a,b)=\left \lceil \frac{n_1\,b-l_1\,a}{\delta|\det A|}\gcd{(a_{21},a_{22})} \right \rceil\;,
\end{equation}
where $\lceil\cdot\rceil$ denotes the ceiling function\footnote{$\lceil\alpha\rceil:=\min\limits_{n\geq \alpha}\{n\in\mathds{Z}\}$}.

Now we define $a_0$ and $b_0$ through $\xi_0$ via eq. \eqref{119} and \eqref{120}. It remains to show the independence of eq. \eqref{119} and \eqref{120}, with $\xi=\xi_0$, from the special choice of solutions $(a,b)$. To achieve this one has to keep in mind that $\left \lceil x+\xi \right \rceil=\left \lceil x \right \rceil + \xi$ for all $\xi\in\mathds{Z}$. Let $a$ and $b$ be arbitrary solutions of \eqref{118}. Then all solutions $(a_\xi,b_\xi)$ to \eqref{118} can be written in the form \eqref{119} and \eqref{120}. Inserting all these solutions into \eqref{118} with $\xi=\xi_0(a,b)$ we obtain
\[
\begin{split}
 a_0(a_\xi,b_\xi) &:=a_\xi-\xi_0(a_\xi,b_\xi)\,\epsilon\,\frac{\delta\,n_2}{\gcd{(\delta\,n_2,l_2)}}=\\&=a-\epsilon\,\xi\, \frac{\delta\,n_2}{\gcd{(\delta\,n_2,l_2)}}-\left \lceil \frac{n_1\,b-l_1\,a}{\delta|\det A|}\gcd(\delta\,n_2,l_2)\right.-
	\\&\qquad\left. -\epsilon\,\xi\,\frac{n_1\, l_2-\delta\,n_2\, l_1}{\delta\,|\det A|\gcd(\delta\,n_2,l_2)}\gcd(\delta\,n_2,l_2) \right \rceil \epsilon\,\frac{\delta\, n_2}{\gcd(\delta\,n_2,l_2)}=
	\\&=a-\xi_0(a,b) \frac{\delta\,n_2}{\gcd(\delta\,n_2,l_2)}=a_0(a,b)\,.
 \end{split}
\]
Independence of $b_0$ from the solutions $(a,b)$ can be shown analogous and so the lemma is proven.
\end{proof}

Let us illustrate this lemma by assuming $\delta=2\cdot 3 \cdot 5^2=150$ and, therefore, computing a sum of some function $f(A)$ over $\delta_1{\not|}\,l_2\wedge\ldots\wedge \delta_n{\not|}\,l_2$, $\det A {\not =} 0$. Lemma \ref{novanish} states that this sum is given by
\begin{equation}
	\sum_{\genfrac{}{}{0pt}{2}{A\in\sumset{1}{1}{150}{1}{\mathds{M}}\wedge \det A \not = 0}{2\not{\,|\,}l_2\wedge 3 \not{\,|\,} l_2\wedge 5\not{\,|\,}\,l_2}}f(A)=\sum_{\genfrac{}{}{0pt}{2}{n_1,n_2,l_1,l_2\in\mathds{Z}}{2\not{\,|\,}l_2\wedge 3 \not{\,|\,} l_2\wedge 5\not{\,|\,}\,l_2}}f\left( \begin{pmatrix} n_1 & l_1 \\ n_2 & \frac{1}{150}\,l_2 \end{pmatrix} \right)=
\end{equation}
\[
	=\sum_{l=1}^{3}(-1)^{l}\,\sum_{\delta^{(l)}\in\,\mathcal{C}_l(150)}\sum_{\genfrac{}{}{0pt}{2}{k,j,p\in\mathds{Z}}{k>j\ge 0\wedge p\not = 0}}\sum_{P\in\Gamma_0(150)}f\left( \begin{pmatrix} k & j \\ 0 & \frac{\delta^{(l)}}{150}\,p \end{pmatrix}\cdot P \right)=
\]
\[
	=\sum_{\genfrac{}{}{0pt}{2}{k,j,p\in\mathds{Z}}{k>j\ge 0\wedge p\not = 0}}\sum_{P\in\Gamma_0(150)}
	\left[
		f\left( \begin{pmatrix} k & j \\ 0 & \frac{1}{150}\,p \end{pmatrix}\cdot P \right)-\right.
\]
\[
		-f\left( \begin{pmatrix} k & j \\ 0 & \frac{2}{150}\,p \end{pmatrix}\cdot P \right) - f\left( \begin{pmatrix} k & j \\ 0 & \frac{3}{150}\,p \end{pmatrix}\cdot P \right)- f\left( \begin{pmatrix} k & j \\ 0 & \frac{5}{150}\,p \end{pmatrix}\cdot P \right)+
\]
\[
		+ f\left( \begin{pmatrix} k & j \\ 0 & \frac{6}{150}\,p \end{pmatrix}\cdot P \right) + f\left( \begin{pmatrix} k & j \\ 0 & \frac{10}{150}\,p \end{pmatrix}\cdot P \right) + f\left( \begin{pmatrix} k & j \\ 0 & \frac{15}{150}\,p \end{pmatrix}\cdot P \right)-
\]
\[
		-\left. f\left( \begin{pmatrix} k & j \\ 0 & \frac{30}{150}\,p \end{pmatrix}\cdot P \right)\right]
\]
Next we want to state the lemma which enables us to sum over all non-zero matrices $A\in\sumset{1}{1}{\delta}{1}{\mathds{M}}$, with $\delta_1{\not|}\,l_2\wedge\ldots\wedge \delta_n{\not|}\,l_2$ and $\det A = 0$, namely
\begin{lem}
\label{vanish}
	Let $\delta={\delta_1}^{x_1}\cdot {\delta_2}^{x_2}\cdot\ldots\cdot{\delta_n}^{x_n}$ and $\delta^{(l)}=\delta_{i_1}\cdot\delta_{i_2}\cdot\ldots\cdot\delta_{i_l}$, with $\delta_i\in\mathds{P}$, be the product of a choice of $l$ prime factors of $\delta$. Moreover, let $\mathcal{C}_l(\delta)$ be the set of all possible products of choices of $l$ prime factors of $\delta$ and $f$ some function on $\mathds{Z}$. Then it holds, at least as a formal sum,
\begin{equation}
\begin{split}
	\sum_{\genfrac{}{}{0pt}{2}{A\in\sumset{1}{1}{\delta}{1}{\mathds{M}}\wedge A\not=0\wedge \det A =0}{\delta_1\not{\,|\,}l_2\wedge\ldots\wedge \delta_n\not{\,|\,}\,l_2}}f(A)=&\frac{1}{2}\,\sum_{l=0}^{n}(-1)^{l}\,\sum_{\delta^{(l)}\in\,\mathcal{C}_l(\delta)}\sum_{\genfrac{}{}{0pt}{2}{j,p\in\mathds{Z}}{(j,p)\not = (0,0)}}\sum_{P\in\Gamma_0(\delta)/\langle T \rangle}f\left( \begin{pmatrix} 0 & j \\ 0 & \frac{\delta^{(l)}}{\delta}\,p \end{pmatrix}\cdot P \right)\,.
\end{split}
\end{equation}
The homogeneous space $\Gamma_0(\delta)/\langle T \rangle$ is defined by considering $P_1,P_2\in\Gamma_0(\delta)$ equivalent if there exists a $m\in\mathds{Z}$ such that\footnote{Recall that $T:=\begin{pmatrix}1 & 1\\ 0 & 1\end{pmatrix}$} $P_1=T^m\cdot P_2$.
\end{lem}
\begin{proof}
This lemma can be proven in a similar way as lemma \ref{novanish}. We define a function
\begin{equation}
F(p):=\left\{
		\begin{array}{cl} 	\frac{1}{2}\, \sum_{j\in\mathds{Z}} \limits \sum_{P\in\Gamma_0(\delta)/\langle T \rangle}\limits f\left( \begin{pmatrix} 0 & j \\ 0 & \frac{\delta^{(l)}}{\delta}\,p \end{pmatrix}\cdot P \right)\,, & \mbox{if }p\not= 0\\ 
		\frac{1}{2}\, \sum_{\genfrac{}{}{0pt}{2}{j\in\mathds{Z}}{j\not = 0}}\limits \sum_{P\in\Gamma_0(\delta)/\langle T \rangle}\limits f\left( \begin{pmatrix} 0 & j \\ 0 & \frac{\delta^{(l)}}{\delta}\,p \end{pmatrix}\cdot P \right)\,, &  \mbox{if } p=0
			\end{array}
		\right.
\end{equation}
analogously to \eqref{109}. Then it holds (cf. \eqref{108}) that
\begin{equation}
\sum_{\delta_1\not{\,|\,}p\,\wedge\ldots\wedge\, \delta_n\not{\,|\,}\,p}F(p)=\frac{1}{2}\cdot\sum_{l=0}^{n}(-1)^{l}\,\sum_{\delta^{(l)}\in\,\mathcal{C}_l(\delta)}\sum_{\genfrac{}{}{0pt}{2}{j,p\in\mathds{Z}}{(j,p)\not = (0,0)}}\sum_{P\in\Gamma_0(\delta)/\langle T \rangle}f\left( \begin{pmatrix} 0 & j \\ 0 & \frac{\delta^{(l)}}{\delta}\,p \end{pmatrix}\cdot P \right)
\end{equation}
Thus, we have to show
\begin{equation}
\label{123}
	\sum_{\delta_1\not{\,|\,}p\,\wedge\ldots\wedge\, \delta_n\not{\,|\,}\,p}F(p)=\sum_{\genfrac{}{}{0pt}{2}{A\in\sumset{1}{1}{\delta}{1}{\mathds{M}}\wedge A\not=0\wedge \det A =0}{\delta_1\not{\,|\,}l_2\wedge\ldots\wedge \delta_n\not{\,|\,}\,l_2}}f(A)\,.
\end{equation}
This is equivalent to prove that every matrix $A\in\sumset{1}{1}{\delta}{1}{\mathds{M}}$, with $\delta_1{\not|}\,l_2\wedge\ldots\wedge \delta_n{\not|}\,l_2$, $A\not= 0$ and $\det A= 0$, can be uniquely\footnote{Up to a sign ambiguity, which will turn out to be irrelevant shortly.} decomposed as
\begin{equation}
\label{124}
	A=\begin{pmatrix} n'_1 & l_1 \\ n_2 & \frac{1}{\delta} \, l_2 \end{pmatrix}= \begin{pmatrix} 0 & j \\ 0 & \frac{\delta^{(l)}}{\delta}\,p \end{pmatrix} \cdot \begin{pmatrix} a & b \\ \delta \,c & d \end{pmatrix} =\begin{pmatrix} \delta \, j \, c & j d \\ p\,c & \frac{p\,d}{\delta} \end{pmatrix}\,,
\end{equation}
where $\begin{pmatrix} a & b \\ \delta \,c & d \end{pmatrix}\in\Gamma_0(\delta)/\langle T \rangle$. We will observe that $(j,p)$ and $(-j,-p)$ label the same orbit. Since $\det A = 0$ it follows that
\begin{equation}
\label{125}
	n'_1\,\frac{l_2}{\delta}=n_2\,l_1=:n_0\in\mathds{Z}\,.
\end{equation}
Because of $\delta_i{\not|}\,l_2$ we can infer $\delta|n'_1=\delta\,n_1$. In consideration of \eqref{124}, i.e.\
\begin{equation}
\label{126}
	c= \frac{n_1}{j}\in\mathds{Z} \quad \text{and} \quad d=\frac{l_1}{j}\in\mathds{Z}
\end{equation}
it follows that
\begin{equation}
\label{127}
	j|\gcd(n_1,l_1)\Leftrightarrow \exists\,k\in\mathds{Z}\,:\,j\,k=\gcd(n_1,l_1)\Leftrightarrow j=\frac{\gcd(n_1,l_1)}{k}\,.
\end{equation}
Therefore, using \eqref{124} again,
\begin{equation}
\label{128}
	p=\frac{n_2}{n_1}\,j=\frac{n_2}{n_1}\,\frac{\gcd(n_1,l_1)}{k}=\frac{n_0}{\lcm(n_1,l_1)}\,\frac{1}{k}\,.
\end{equation}
Since the least common multiple of two numbers $n_1$ and $l_1$ is the product of the highest powers of all prime-factors that are present in their prime-factorisations, it follows together with the fact that $n_1$ and $l_1$ are both divisors of $n_0$ that $p\in\mathds{Z}$ for at least $k=\pm 1$.

Since $a\,d-\delta\,b\,c=1$ it is true that $\gcd(c,d)=1$. Thus, $|j|=\gcd(n_1,l_1)\Leftrightarrow k=\pm 1$. Here we can observe the above mentioned sign ambiguity. Both $k=1$ and $k=-1$ yield a consistent solution of \eqref{124} for the same matrix $A$. Hence, we have
\begin{equation}
\label{129}
	\begin{pmatrix} 0 & j \\ 0 & \frac{\delta^{(l)}}{\delta}\,p \end{pmatrix} \cdot \Gamma_0(\delta)= \begin{pmatrix} 0 & -j \\ 0 & -\frac{\delta^{(l)}}{\delta}\,p \end{pmatrix} \cdot \Gamma_0(\delta)\,.
\end{equation}
Let $P_1,P_2\in\Gamma_0(\delta)$. Then it is clear that
\begin{equation}
\label{130}
	\begin{pmatrix} 0 & j \\ 0 & \frac{\delta^{(l)}}{\delta}\,p \end{pmatrix} \cdot P_1=\begin{pmatrix} 0 & j \\ 0 & \frac{\delta^{(l)}}{\delta}\,p \end{pmatrix} \cdot P_2
\end{equation}
if
\begin{equation}
\label{131}
	P_1=\begin{pmatrix} 1 & m \\ 0 & 1 \end{pmatrix}\cdot P_2 \equiv T^m \cdot P_2\,.
\end{equation}
So the lemma is proven.
\end{proof}

Now we can begin the
\begin{proof}[Proof of theorem \ref{main}]
The idea of the proof is as follows. Let $\delta={\delta_1}^{x_1}\cdot\ldots\cdot {\delta_n}^{x_n}$. The problematic part of the sum over all $A\in\sumset{1}{1}{\delta}{1}{\mathds{M}}$ will be the sum over $l_2$, with $A$ parametrised as in \eqref{103}. To perform this summation, we split the sum over all $l_2\in\mathds{Z}$ into two sums $\delta_1{\not|}\,l_2\,\wedge\,\ldots\,\wedge\,\delta_n{\not|}\,l_2 $ and $\delta_1|l_2\,\vee\,\ldots\,\vee\,\delta_n |l_2$. Using lemmas \ref{novanish} and \ref{vanish} the first sum can be decomposed into sums over integrals which can be computed directly via lemma \ref{refint}. Using lemma \ref{choice} the second sum can be decomposed into different sums of the form $d|l_2$, which lead to a reduction via lemma \ref{red} to $\frac{\delta}{d}|\delta$. We will use this procedure successively for all reduced sums until we are left with $\delta'=1$. Firstly, we will apply it to matrices with non-vanishing determinant to get the constants $\sumset{1}{1}{\delta}{1}{C}(d)$. After that we will show that the contribution of the zero matrix takes exactly the value to complete the contribution of the matrices with non-vanishing determinant. At the end we will look at the non-zero matrices with vanishing determinant. The finiteness of this expression fixes the multiplicative constant $\sumset{1}{1}{\delta}{1}{A}$.

One of the results of this section up to now is the fact that every integral which has to be computed arises as the contribution of matrices of the form
\begin{equation}
\label{main4}
	\begin{pmatrix}
		k & j \\
		0 & \frac{\delta'^{(l)}}{\delta'}\, p
	\end{pmatrix}
	\quad \text{and} \quad
	\begin{pmatrix}
		0 & j \\
		0 & \frac{\delta'^{(l)}}{\delta'}\, p
	\end{pmatrix}
\end{equation}
for $\delta'|\delta$ and $\delta'^{(l)}\in\mathcal{C}_l(\delta')$. The general form of these contributions is given by
\begin{equation}
\label{main3}
	\frac{1}{\delta'^{(l)}}\,\sum_{\genfrac{}{}{0pt}{2}{k>j\ge 0}{p\not = 0}} \frac{T_2}{\delta'/\delta'^{(l)}}\,\exp\left(-2\pi i \,T\det \begin{pmatrix} k & j \\ 0 & \frac{\delta'}{\delta'^{(l)}}\,p \end{pmatrix}\right)
\end{equation}
\[
	\quad\times\int_{-\infty}^{\infty}d\tau_1\int_0^{\infty}\frac{d\tau_2}{{\tau_2}^2}\exp\left(-\frac{\pi\, T_2}{\tau_2\,U_2}\left| \begin{pmatrix} 1 & U \end{pmatrix} \cdot \begin{pmatrix} k & j \\ 0 & \frac{\delta'^{(l)}}{\delta'} \,p \end{pmatrix} \cdot \begin{pmatrix} \tau \\1 \end{pmatrix} \right|^2\right)=
\]
\[
	=\frac{1}{\delta'^{(l)}}\,\sum_{\genfrac{}{}{0pt}{2}{0\leq j<k }{p\not = 0}} T_2\,\ex{-2\pi i \,T\, \frac{\delta'^{(l)}}{\delta'}\,k p}\,\int_{-\infty}^{\infty}d\tau_1\int_0^{\infty}\frac{d\tau_2}{{\tau_2}^2}\exp\left(-\frac{\pi\,\frac{\delta'^{(l)}}{\delta'}\, T_2}{\tau_2\,\frac{\delta'^{(l)}}{\delta'}\,U_2}\left|k\tau+j+p\,\frac{\delta'^{(l)}}{\delta'}\,U\right|^2\right)=
\]
\[
	=-\frac{1}{\delta'^{(l)}}\,4\text{Re} \ln\prod_{n=1}^{\infty} \left(1-\ex{2\pi i\,\frac{T}{\delta'/\delta'^{(l)}}\,n}\right)\,,
\]
for matrices with non-vanishing determinant and
\begin{equation}
\label{main5}
		\frac{1}{\delta'^{(l)}}\,\int_{-1/2}^{+1/2}d\tau_1\int_0^\infty\frac{d\tau_2}{{\tau_2}^2}\left[\frac{T_2}{\delta'/\delta'^{(l)}}\sideset{}{'}\sum_{j,p}\exp\left(-\frac{\pi\,T_2}{\tau_2\,U_2}\left|\begin{pmatrix} 1 & U \end{pmatrix}\begin{pmatrix} 0 & j \\ 0 & \frac{\delta'^{(l)}}{\delta}\,p  \end{pmatrix}\begin{pmatrix} \tau \\ 1  \end{pmatrix}\right|^2\right)-\tau_2\theta_{R_\Gamma}(\tau)\right]=
\end{equation}
\[
	=\frac{1}{\delta'^{(l)}}\,\int_{-1/2}^{+1/2}d\tau_1\int_0^\infty\frac{d\tau_2}{{\tau_2}^2}\left[\frac{T_2}{\delta'/\delta'^{(l)}}\sideset{}{'}\sum_{j,p}\exp\left(\frac{\pi\,\frac{\delta'^{(l)}}{\delta'}\,T_2}{\tau_2\,\frac{\delta'^{(l)}}{\delta'}\,U_2}\left|j+\frac{\delta'^{(l)}}{\delta'}U\,p\right|^2\right)-\tau_2\theta_{R_\Gamma}(\tau)\right]=
\]
\[
	=-\frac{1}{\delta'^{(l)}}\left[\,4\text{Re}\ln \eta\left(\frac{U}{\delta'/\delta'^{(l)}}\right)-\ln\left(\frac{T_2}{\delta'/\delta'^{(l)}}\,\frac{U_2}{\delta'/\delta'^{(l)}}\right)-\ln\left(\frac{8\pi\,\ex{1-\gamma_E}}{3\sqrt{3}}\right)\right]
\]
for non-zero matrices with vanishing determinant.

Let us look closer at our proposition, i.e.\ formula \eqref{main1}. The different terms in the sum on the right-hand side equal $\sumset{1}{1}{1}{1}{I}(\frac{T}{d},\frac{U}{d})$ and it is feasible to guess that they stem from a reduction $\delta \to 1$. However, how does the sum over all divisors arise from the reduction?
To see that, observe that at every intermediate step of reduction $d=\delta_1^{x'_1}\cdot\ldots\cdot\delta_n^{x'_n}|\delta$, lemma \ref{novanish} ensures that there is always a contribution of matrices $A$ with $\delta_{i_1}{\not|}\,l_2\,\wedge\,\ldots\,\wedge\,\delta_{i_{n'}}{\not|}\,l_2$ and
\begin{equation}
\label{main6}
	A=\begin{pmatrix} k & j \\ 0 & \frac{1}{d}\,p \end{pmatrix}\,.
\end{equation}
For matrices with vanishing determinant there holds an analogous statement with $k=0$. Since these two assertions are true for every divisor $d$ of $\delta$, this leads to a sum over all these divisors $d$ in \eqref{main}. It remains to compute the prefactors $\sumset{1}{1}{\delta}{1}{C}(d)$ of these summands and the overall constant $\sumset{1}{1}{\delta}{1}{A}$. 

As mentioned, imposing divisibility conditions on $l_2=\delta^{(l)}\, l'_2$, with $\delta^{(l)}\in\mathcal{C}_l(\delta)$, results in a reduction of $\delta$ to $\frac{\delta}{\delta^{(l)}}$. On the resulting sum we can again impose divisibility conditions $l'_2=\delta'^{(l)}\,l''_2$, which gives rise to another reduction, and so on. Now, we will show the following
\begin{claim}
The following two procedures are equivalent
\begin{enumerate}
	\item successively apply lemmas \ref{novanish}, \ref{vanish}, \ref{choice} and \ref{red} to \eqref{102}
	\item sum over all $d_1\,l_2$, $d_2\,l_2$, ..., $d_{\sigma_0(\delta)}\,l_2$, where $\sigma_0(\delta)$ denotes the numbers of divisors of $\delta$, and then apply lemmas \ref{novanish} and \ref{vanish}.
\end{enumerate}
\end{claim}
\begin{proof}
Recall that the reduction of $\delta'$ to $\frac{\delta'}{\delta'^{(l)}}$ gives rise to a multiplicative constant (cf. lemma \ref{red})
\[
	\sumset{1}{1}{\delta'^{(l)}|\delta'}{1}{\mathcal{I}}(T,U)=\frac{\left[\Gamma_0\left(\frac{\delta'}{\delta'^{(l)}}\right):\Gamma_0(\delta')\right]}{\delta'^{(l)}}\,\sumset{1}{1}{\frac{\delta'}{\delta'^{(l)}}}{1}{\mathcal{I}}(T,U)\,.
\]
When reducing by a prime number $p$ which divides $\delta'$ as well as $\frac{\delta'}{\delta'^{(l)}}$ , this constant is $1$. For every prime number $p|\delta'$ which does not divide $\frac{\delta'}{\delta'^{(l)}}$ a factor of $\frac{p+1}{p}$ has to be multiplied, i.e.\
\begin{equation}
\label{main7}
	\frac{\left[\Gamma_0\left(\frac{\delta'}{\delta'^{(l)}}\right):\Gamma_0(\delta')\right]}{\delta'^{(l)}}=\prod_{p|\delta'\,\wedge\,p\cancel{\,|\,}\frac{\delta'}{\delta'^{(l)}}\,\wedge\, p\in\mathds{P}}\frac{p+1}{p}
\end{equation}
It is important that the multiplicative constant which we gain by reducing from $\delta$ to $\delta'$ does not depend on the path on which we did it. It does only depend on $\delta$ and $\delta'$. That means that the constant is the same if we first reduce by $\delta_i$ and then by $\delta_j$, first by $\delta_j$ and then by $\delta_i$ or by $\delta_i\,\delta_j$. By reducing $\delta$ by $\delta_i$ we mean reducing $\delta$ to $\frac{\delta}{\delta_i}$. Thus, to show the claim above, we have to look at the possible ways (with sign) of how we can reduce $\delta$ to $\delta'$.

Let us look at an overall reduction of $\delta$ by $d|\delta$. Furthermore, let us denote the numbers of primes which divide $n\in\mathds{Z}$ by $\#p(n)$. Then the preimage of $\frac{\delta}{d}$ under (direct) reduction is given by the set of all $\frac{\delta}{d}\,\delta^{(l)}$ with $\delta^{(l)}\in\mathcal{C}_l(d)$ and $l\le \# p(d)$.\footnote{Recall that we use lemma \ref{choice} and lemma \ref{red} for reduction.} Reducing $\frac{\delta}{d}\,\delta^{(l)}$ by $\delta^{(l)}$ results in a sign $(-1)^{l+1}$. It is clear that there is only one way of reducing $\delta$ by $\delta_i$, with $\delta_i\in\mathds{P}$. Now let us assume that it is true that we count every divisor $\frac{\delta}{d'}$ exactly once for $d'<d$. If we define $l_0:=\#p(d)$ we can infer how often $\frac{\delta}{d}$ has been counted
\begin{equation}
\label{main8}
	\sum_{l=1}^{l_0}\sum_{\delta^{(l)}\in\mathcal{C}_l(d)}(-1)^{l+1}=\sum_{l=1}^{l_0} {\binom{l_0}{l}} (-1)^{l+1}=1-(1-1)^{l_0}=1.
\end{equation}
Here we used
\begin{equation}
\label{main9}
	|\mathcal{C}_l(d)|={\binom{l_0}{l}}\,.
\end{equation}
Hence, if the assertion is true that we count every divisor $\frac{\delta}{d'}$ exactly once for $d'<d$ it follows that it is also true for $d$. Since it holds for $\mathds{P}\ni\delta_i|\delta$ we can infer that the claim is true.
\end{proof}

Above, we motivated that the result of \eqref{102} is a sum over all divisors $d$ of $\delta$. Afterwards we have shown that it is equivalent to iteratively use lemmas \ref{novanish}, \ref{vanish}, \ref{choice} and \ref{red} or to look at the sum over all restricted sums $l_2=d_1\,l'_2$, $l_2=d_2\,l'_2$, ..., $l_2=d_{\sigma_0(\delta)}\,l'_2$ and apply lemmas \ref{novanish} and \ref{vanish} to these sums. Now, let us look at such a restricted sum $d|l_2$. It is given by
\begin{equation}
\label{main11}
\begin{split}
	\sumset{1}{1}{\delta'|\delta}{1}{\mathcal{I}}(T,U)&=\frac{\left[\Gamma_0\left(\frac{\delta}{\delta'}\right):\Gamma_0(\delta)\right]}{\delta'}\,\int_{R_{\Gamma_0(\delta/\delta')}}\hypmes{\tau}\,\sum_{A\in\sumset{1}{1}{\delta'|\delta}{1}{\mathds{M}}}\ex{-2\pi i \,T\det{A}}\,\frac{T_2}{\delta/\delta'}\,\\&\times\exp\left[-\frac{\pi\,T_2}{\tau_2\,U_2}\left| \begin{pmatrix} 1 & U \end{pmatrix} A \begin{pmatrix} \tau\\1 \end{pmatrix}\right|^2\right]\;.
\end{split}
\end{equation}
Note that this is strictly spoken a formal expression, since it contains infinite contributions. These will eventually cancelled by the regulator from \eqref{102}. We will show explicitly that this is always possible by determining $\sumset{1}{1}{\delta}{1}{A}$ later.

Since $\frac{\delta}{\delta'}\in\mathds{Z}$ we can write $\frac{\delta}{\delta'}={\delta_{i_1}}^{y_1}\cdot\ldots\cdot{\delta_{i_m}}^{y_m}$ for $\delta_{i_j}\in\mathds{P}$. Imposing $\delta_{i_1}{\not|}\, l'_2\wedge\ldots\wedge\delta_{i_m}{\not|}\, l'_2$ results, using lemma \ref{novanish} in integrals of the form
\begin{equation}
\label{main12}
\begin{split}
\mathcal{I}'_{\delta^{(l)}\delta'|\delta}(T,U)&=\frac{\left[\Gamma_0\left(\frac{\delta}{\delta'}\right):\Gamma_0(\delta)\right]}{\delta^{(l)}\,\delta'}\,\int_{R_{\Gamma_0\left(\delta/\delta'\right)}}\hypmes{\tau}\,\sum_{\genfrac{}{}{0pt}{2}{k>j\ge 0}{p\not=0}}\sum_{P\in\Gamma_0\left(\frac{\delta}{\delta'}\right)}\ex{-2\pi i \,T\frac{\delta'^{(l)}\,\delta'}{\delta}}\,\frac{T_2}{\delta/\delta^{(l)}\,\delta'}\,\\&\times\exp\left[-\frac{\pi\,T_2}{\tau_2\,U_2}\left| \begin{pmatrix} 1 & U \end{pmatrix} \begin{pmatrix} k & j \\ 0 & \frac{\delta^{(j)}\,\delta'}{\delta}\,p\end{pmatrix} \cdot P\ \begin{pmatrix} \tau\\1 \end{pmatrix}\right|^2\right]
\end{split}
\end{equation}
for matrices with non-vanishing determinant. The matrix multiplication with $P$ can be interpreted as a modular transformation on $\tau$, cf. eq. \eqref{35}. Since we sum over all matrices $P\in\Gamma_0\left(\frac{\delta}{\delta'}\right)$, we are left with an integral over $\mathds{H}^+$ and are allowed to use lemma \ref{refint}. If we apply the same reasoning to matrices with vanishing determinant (using lemma \ref{vanish}), set $k=0$ and adjust the sum, we have to take into account that not all $P_1,P_2 \in \Gamma_0\left(\frac{\delta}{\delta'}\right)$ yield different matrices, cf. the proof of lemma \ref{vanish}. There we found out that two matrices $P_1,P_2$ with $P_1=T^m\,P_2$ for some $m\in\mathds{Z}$ label the same orbit. Hence, we have to integrate the contributions of matrices with vanishing determinant over $\mathcal{H}^+/\langle T \rangle$, which is given by a stripe $\left\{\tau\in\mathds{H}^+:|\tau_1|<\frac{1}{2}\right\}$. If we denote the contribution of all these matrices by $\mathcal{I}''_{\delta^{(l)}\delta'|\delta}$, we get using \ref{refint}
\begin{equation}
\label{main13}
	\mathcal{I}'_{\delta^{(l)}\delta'|\delta}(T,U)+\mathcal{I}''_{\delta^{(l)}\delta'|\delta}(T,U)+\mathcal{I}_{\text{zero-matrix}}+\text{reg.}=
\end{equation}
\[
	=-\,\frac{\left[\Gamma_0\left(\frac{\delta}{\delta'}\right):\Gamma_0(\delta)\right]}{\delta^{(l)}\,\delta'}\f{}{\delta/\delta'\delta^{(l)}}\;,
\]
where reg.\ denotes the fraction of the regulator from \eqref{102} which cancels the divergent contributions.

Thus, we understood the principal form of the summands in \eqref{main1}. It is interesting that there always appear blocks containing the logarithm of Dedekind $\eta$-functions, very similar to the result for the case $\alpha=\beta=\gamma=\delta=1$.

Our next aim is to show
\[
	\sumset{1}{1}{\delta}{1}{C}(d)=\prod_{p|d\,\wedge\,p\left|\frac{\delta}{d}\right.\,\wedge\,p\in\mathds{P}}\left(1-\frac{1}{p}\right)\,.
\]

Using the above construction, we can show an intermediate result towards the multiplicative constants $\sumset{1}{1}{\delta}{1}{C}(d)$, namely
\begin{equation}
\label{main10}
	\sumset{1}{1}{\delta}{1}{C}(d)=\sum_{l=0}^{\# p\left(\frac{\delta}{d}\right)}\sum_{\delta^{(l)}\in\mathcal{C}_l\left(\frac{\delta}{d}\right)}(-1)^l\,\frac{1}{\delta^{(l)}}\,\prod_{p|\delta\,\wedge\,p\cancel{\,|\,}\delta^{(l)}\cdot d\,\wedge\, p\in\mathds{P}}\frac{p+1}{p}
	\;.
\end{equation}
To do this, let's compute the prefactor of the block which is associated to $d$ being a divisor of $\delta$. Equation \eqref{main10} sums over all preimages of $d$ under reduction of restricted sums via lemmas \ref{novanish} and \ref{vanish}. These preimages are given by the set of all $d\cdot \delta^{(l)}$, with $\delta^{(l)}\in\mathcal{C}_l\left(\frac{\delta}{d}\right)$ and $l\le\#p\left(\frac{\delta}{d}\right)$. Reducing $d\cdot \delta^{(l)}$ by $\delta^{(l)}$ to $d$ via lemmas \ref{novanish} and \ref{vanish} results, firstly, in a sign $(-1)^{l}$. Secondly, if we look at \eqref{main11} and \eqref{main12}, we observe that this contribution causes a prefactor of $\frac{1}{\delta^{(l)}}$. Thirdly, we have to account for the multiplicative constants \eqref{main7} arising from reducing $\delta$ to $d\cdot\delta^{(l)}$ via the claim. This factor is the last product in \eqref{main10}.

By definition it is clear that\footnote{We use $\psi(n)=[\Gamma:\Gamma_0(n)]=n\prod_{p|n\,\wedge\,p\in\mathds{P}}\left(1+\frac{1}{p}\right)$.}
\begin{equation}
\label{main16new}
	\frac{\psi(p\cdot d)}{\psi(d)}\frac{1}{p}=\left\{
		\begin{array}{cl} 1	\,,&  \mbox{if }p|d\\ 
						  \frac{p+1}{p}				\,,&  \mbox{if }p{\not|}\,d
		\end{array}\right.\,.
\end{equation}
This implies together with $p\cancel{\,|\,}\delta^{(l)}\cdot d\Leftrightarrow p\cancel{\,|\,}\delta^{(l)}\wedge\,p\cancel{\,|\,}d$ for all prime numbers $p\in\mathds{P}$ that
\begin{equation}
	\prod_{p|\delta\,\wedge\,p\cancel{\,|\,}\delta^{(l)}\cdot d\,\wedge\, p\in\mathds{P}}\frac{p+1}{p}=
	\prod_{p|\frac{\delta}{d}\,\wedge\,p\cancel{\,|\,}\delta^{(l)}\,\wedge\,p\in\mathds{P}}\frac{\psi(p\cdot d)}{\psi(d)}\frac{1}{p}\,.
\end{equation}
Thus, equation \eqref{main10} is equivalent to
\begin{equation}
\label{mainnew1}
	\sumset{1}{1}{\delta}{1}{C}(d)=\sum_{l=0}^{\# p\left(\frac{\delta}{d}\right)}\sum_{\delta^{(l)}\in\mathcal{C}_l\left(\frac{\delta}{d}\right)}(-1)^l\,\frac{1}{\delta^{(l)}}\,\prod_{p|\frac{\delta}{d}\,\wedge\,p\cancel{\,|\,}\delta^{(l)}\,\wedge\,p\in\mathds{P}}\frac{\psi(p\cdot d)}{\psi(d)}\frac{1}{p}
	\,.
\end{equation}
Multiplying out
\begin{equation}
\label{main14}
	\sumset{1}{1}{\delta}{1}{C}(d)=\prod_{p\left|\frac{\delta}{d}\right.}\limits\left(\frac{\psi(p\cdot d)}{\psi(d)}\frac{1}{p}-\frac{1}{p}\right)\,.
\end{equation}
yields exactly equation \eqref{mainnew1}. If we use
\begin{equation}
\label{main16}
	\frac{\psi(p\cdot d)}{\psi(d)}\frac{1}{p}-\frac{1}{p}=\left\{
		\begin{array}{cl} 1-\frac{1}{p}	\,,&  \mbox{if }p|d\\ 
						  1				\,,&  \mbox{if }p{\not|}\,d
		\end{array}\right.\,,
\end{equation}
we get \eqref{main2}. It is easy to see that $\sumset{1}{1}{\delta}{1}{A}$ is the inverse of the sum over all $\sumset{1}{1}{\delta}{1}{C}(d)$ as follows. The regulator in \eqref{master} has to cancel the divergent contributions from the orbit of non-zero matrices with vanishing determinant. Since it has to regulate all of them (belonging to various divisors of $d$), it must be the inverse of the sum of their multiplicative coefficients.

Hence, we took care of the orbits for non-vanishing matrices. Now let us look at the remaining contribution of the zero matrix. It is given by
\begin{equation}
\label{main17}
	\mathcal{I}_{\text{zero-matrix}}=\int_{R_{\Gamma_0(\delta)}} \hypmes{\tau}\, \frac{T_2}{\delta}=\frac{\pi}{3}\,T_2\, \frac{[\Gamma:\Gamma_0(\delta)]}{\delta}=\frac{\pi}{3}\,T_2\,\prod_{p|\delta\,\wedge\,p\in\mathds{P}}\left(1+\frac{1}{p}\right)
\end{equation}
\[
	\quad= T_2\,\frac{\pi}{3}\sum_{l=0}^{\#p(\delta)}\sum_{\delta^{(l)}\in\mathcal{C}_l(\delta)}\frac{1}{\delta^{(l)}}\,.
\]
To prove theorem \ref{main} it remains to show
\begin{equation}
\label{main18}
	T_2\,\frac{\pi}{3}\,\sum_{d|\delta}\,\frac{1}{d}\prod_{p|d\,\wedge\, p\left|\frac{\delta}{d}\right.}\left(1-\frac{1}{p}\right)=T_2\,\frac{\pi}{3}\sum_{l=0}^{\#p(\delta)}\sum_{\delta^{(l)}\in\mathcal{C}_l(\delta)}\frac{1}{\delta^{(l)}}\,.
\end{equation}
For $\delta=\delta_1\cdot\ldots\cdot\delta_n$, i.e.\ no prime factor $\delta_i$ occurs more than once, this equation is trivial. We will now show that if this equation holds for $\delta$, then it holds for $\delta_i\cdot\delta$, with $\delta_i|\delta$ and $\delta_i\in\mathds{P}$. This will give us the complete proof of theorem \ref{main}.

Let $\delta={\delta_1}^{x_1}\cdot\ldots\cdot{\delta_n}^{x_n}\in\mathds{Z}$, $\delta_i\in\mathds{P}$ with $\delta_i|\delta$ and let \eqref{main18} be true for $\delta$. We will show
\begin{equation}
\label{main19}
	\sum_{d|\delta\cdot\delta_i}\sumset{1}{1}{\delta_i\cdot\delta}{1}{C}(d)\,\frac{1}{d}=\sum_{d|\delta}\sumset{1}{1}{\delta}{1}{C}(d)\,\frac{1}{d}
\end{equation}
by comparing
\begin{equation}
\label{main20}
	C_\delta(d):=\prod_{p|d\,\wedge\,p|\frac{\delta}{d}\,\wedge\,p\in\mathds{P}}\left(1-\frac{1}{p}\right)\quad\text{and}\quad
	C_{\delta_i\delta}(d):=\prod_{p|d\,\wedge\,p|\frac{\delta_i\delta}{d}\,\wedge\,p\in\mathds{P}}\left(1-\frac{1}{p}\right)\,.
\end{equation}
It is clear that if there is some difference between the left-hand and the right-hand side in equation \eqref{main19}, it must be caused by $\delta_i$. Let us consider different cases.

The case $d{\not|}\,\delta$ and $d{\not|}\,\delta_i\cdot\delta$ is not interesting, while the case $d|\delta$ and $d{\not|}\,\delta_i\cdot\delta$ is a contradiction. 

Let now $d{\not|}\,\delta$ and $d|\delta_i\cdot\delta$. We denote the power of $\delta_i$ in the prime factorisation of $\delta$ by $p_{\delta_i}(\delta)=x_i$. Then $p_{\delta_i}(d)=x_i+1$ and $d=d'\cdot{\delta_i}^{x_i+1}$ with $d'\left|\frac{\delta}{{\delta_i}^{x_i+1}}\right.$. 

Since $d{\not|}\,\delta$, the term $C_\delta(d)$ is not present in the left-hand side of \eqref{main19}.

Because of $\delta_i{\not|}\,d'$ and $\delta_i{\not|}\,\frac{\delta\delta_i}{d'{\delta_i}^{x_i+1}}$ it is true that
\begin{equation}
\label{main21}
	C_{\delta\delta_i}(d)=C_{\delta\delta_i}(d'{\delta_i}^{x_i+1})=\prod_{p|d'{\delta_i}^{x_i+1}\,\wedge\,p\left|\frac{\delta\delta_i}{d'{\delta_i}^{x_i+1}}\right.\,\wedge\,p\in\mathds{P}}\left(1-\frac{1}{p}\right)=
\end{equation}
\[
	\quad\,\qquad =\prod_{p|d'\,\wedge\,p\left|\frac{\delta}{d'}\right.\,\wedge\,p\in\mathds{P}}\left(1-\frac{1}{p}\right)=
	C_\delta(d')
\]
Next, we consider the case $d|\delta$ and $d {\not|} \,\delta_i\delta$. From $\delta_i{\not|}\,d$ we deduce $C_\delta(d)=C_{\delta_i\delta}(d)$. If $\delta_i|d$ there are four cases:
\begin{enumerate}
\item  If $\delta_i\left|\frac{\delta}{d}\right.\,\wedge\,\delta_i\left|\frac{\delta_i\delta}{d}\right.$ it is evident that $C_\delta(d)=C_{\delta_i\cdot\delta}$.
\item  Considering $\delta_i\left|\frac{\delta}{d}\right.\,\wedge\,\delta_i\not\left|\frac{\delta_i\delta}{d}\right.$ yields a contradiction.
\item  Setting $\delta_i\not\left|\,\frac{\delta}{d}\right.\,\wedge\,\delta_i\not\left|\,\frac{\delta_i\delta}{d}\right.$ results in $p_{\delta_i}(d)=x_i$ and, thus,
\begin{equation}
\label{main22}
	C_\delta(d)\left(1-\frac{1}{\delta_i}\right)=C_{\delta_i\delta}(d)\,.
\end{equation}
\item  The last case $\delta_i\not\left|\,\frac{\delta}{d}\right.\,\wedge\,\delta_i\not\left|\,\frac{\delta_i\delta}{d}\right.$ is a contradiction since $\delta_i|d$ and $d|\delta$.
\end{enumerate}
Altogether, the difference between the left-hand and the right-hand side of equation \eqref{main19} reads
\begin{equation}
\label{main23}
	\sum_{\genfrac{}{}{0pt}{2}{d|\delta}{p_{\delta_i}(d)=x_i}}\frac{1}{d}\,C_\delta(d)-\sum_{\genfrac{}{}{0pt}{2}{d|\delta_i\delta}{p_{\delta_i}(d)=x_i}}\frac{1}{d}\,C_{\delta_i\delta}(d)-\sum_{\genfrac{}{}{0pt}{2}{d|\delta_i\delta}{p_{\delta_i}(d)=x_i+1}}\frac{1}{d}\,C_{\delta_i\delta}(d)=
\end{equation}
\[
	\quad=\sum_{\genfrac{}{}{0pt}{2}{d|\delta}{p_{\delta_i}(d)=x_i}}\frac{1}{d}\,C_\delta(d)-\left(1-\frac{1}{\delta_i}\right)\sum_{\genfrac{}{}{0pt}{2}{d|\delta}{p_{\delta_i}(d)=x_i}}\frac{1}{d}\,C_\delta(d)-\sum_{\genfrac{}{}{0pt}{2}{d'|\delta}{p_{\delta_i}(d')=0}}\frac{1}{d'{\delta_i}^{x_i+1}}\,C_\delta(d')
\]
\[
	\quad=\frac{1}{\delta_i}\,\sum_{\genfrac{}{}{0pt}{2}{d|\delta}{p_{\delta_i}(d)=x_i}}\frac{1}{d}\,C_\delta(d)-\frac{1}{\delta_i}\sum_{\genfrac{}{}{0pt}{2}{d|\delta}{p_{\delta_i}(d)=x_i}}\frac{1}{d}\,C_\delta(d)=0
\]
This proves \eqref{main19}. If \eqref{main18} is true for $\delta$ it follows that it is true for $\delta_i\delta$, with $\delta_i|\delta$ and $\delta_i\in\mathds{P}$. Since it is true for $\delta=\delta_1\cdot\ldots\cdot\delta_n$ we can infer that it is also true for all $\delta\in\mathds{Z}$ and we have proven theorem \ref{main}. 

\end{proof}
Let us analyse the symmetries of equation \eqref{main1}. The building block
\[
	\f{}{d}
\]
has the symmetry $\Gamma(1/d,d)=\Gamma_0(1/d)\cap \Gamma^0(d)$ acting on $T$ or $U$. Hence, the sum over all these building blocks, i.e. equation \eqref{main1}, is symmetric under
\begin{equation}
\label{mainsym1}
	\bigcap_{d|\delta}\Gamma\left(\frac{1}{d},d\right)=\Gamma^0(\delta)
\end{equation}
acting on $T$ and $U$ independently. But there are more symmetries of eq. \eqref{main1}. If we examine eq. \eqref{main2}, we observe that 
\begin{equation}
\label{mainsym2}
	\sumset{1}{1}{\delta}{1}{C}(d)=\sumset{1}{1}{\delta}{1}{C}\left(\frac{\delta}{d}\right)\,.
\end{equation}
Thus, \eqref{main1} admits the additional involutive symmetries
\begin{align}
	&\sumset{1}{1}{\delta}{1}{\mathfrak{T}}: T\mapsto T'=-\frac{\delta}{T} \quad \text{and}\\
	&\sumset{1}{1}{\delta}{1}{\mathfrak{U}}: U\mapsto U'=-\frac{\delta}{U}
	\;.
\end{align}
This is a generalisation of what is commonly denoted as T-duality. Note in particular, that, although this looks very similar to the usual T-duality, this symmetry is \emph{not} a modular transformation (except for the case $\delta=1$). It also possesses different self-dual points than the common T-duality. We will investigate physical consequences of this fact in another work and for specific models \cite{PalKlapUnpub}.

Last but not least, there is yet another symmetry which interchanges the role of $T$ and $U$. It is given by 
\begin{equation}
\sumset{1}{1}{\delta}{1}{\mathfrak{M}}:(T,U)\mapsto (T',U')=(U,T)
\;.
\end{equation}
It corresponds to the mirror map acting on the fixed plane.

Hence, the complete symmetry group $\sumset{1}{1}{\delta}{1}{\mathfrak{S}}$ of equation \eqref{main1} is given by
\begin{equation}
	\label{mainsym3}
	\sumset{1}{1}{\delta}{1}{\mathfrak{S}}=
\left[\left(\Gamma^0(\delta)\ast \sumset{1}{1}{\delta}{1}{\mathfrak{T}}\right)_T\times\left(\Gamma^0(\delta)\ast \sumset{1}{1}{\delta}{1}{\mathfrak{U}}\right)_U\right] \ast \sumset{1}{1}{\delta}{1}{\mathfrak{M}}\,,
\end{equation}
where $\ast$ denotes the free product of groups.
\subsection{Reduction of $\alpha,\,\beta,\,\gamma,\,\delta\in\mathds{Q}$ to $\alpha=\beta=\gamma=1$ and $\delta\in\mathds{Z}$}

In the last section we computed \eqref{master} for $\alpha=\beta=\gamma=1$ and $\delta\in\mathds{Z}$. Considering the more general case $\alpha=\beta=\gamma=1$, $\delta\in\mathds{Q}$ one faces difficulties which turn out to be so severe that they seem to leave no hope for a direct solution. Matters even get worse considering the most general case $\alpha,\beta ,\gamma ,\delta\in\mathds{Q}$. Summing over all matrices with fractional entries appears to be even more complicated than summing over integers has been. Fortunately, in the special case at hand it is not. As we will show in this section, it is possible to reduce all cases to the case $\alpha=\beta=\gamma=1$, $\delta\in\mathds{Z}$ by transforming the moduli appropriately.
 
 The starting point of our consideration is given by the expression of the one-loop partition function on the world-sheet which is associated to the boundary condition $\left(1,\theta^{l_k}\right)$. The most general form it can take is (cf. \eqref{master} before Poisson resummation)
\begin{equation}
\label{55}
\begin{split}
	Z^{\text{one-loop}}_{\left(1,\theta^{l_k}\right)}\left(\tau\right)=&\sum_{\genfrac{}{}{0pt}{2}{n_1,n_2\in\mathds{Z} }{m_1,m_2\in\mathds{Z}}}\ex{2\pi i \,\tau\left(\gamma m_1 \,\alpha n_1+\delta m_2\,\beta n_2 \right)}\\&\times\exp\left[-\frac{\pi\,\tau_2}{T_2\,U_2}\left| T U \,\beta n_2+T \,\alpha n_1 -\gamma m_1\,U+\delta m_2\right|^2\right]\,.
\end{split}
\end{equation}
If we define
\begin{equation}
\label{56}
	A:=\begin{pmatrix} \alpha n_1 & \,\,\delta m_2 \\\beta n_2 & -\gamma m_1  \end{pmatrix}\,,
\end{equation}
\eqref{55} can be written as
\begin{equation}
\label{57}
	Z^{\text{one-loop}}_{\left(1,\theta^{l_k}\right)}\left(\tau\right)=\sum_{A\in\sumset{\alpha}{\beta}{1/\gamma}{1/\delta}{\mathds{M}}}\ex{-2\pi i \,\tau\, \det{A}}\,\exp\left[-\frac{\pi\,\tau_2}{T_2\,U_2}\left| \begin{pmatrix} 1 & U \end{pmatrix} A \begin{pmatrix} T\\1 \end{pmatrix}\right|^2\right]\,.
\end{equation}
Furthermore, let us define
\begin{equation}
\label{58}
\begin{split}
	\mathcal{H}:&=-\frac{\pi\,\tau_2}{T_2\,U_2} \left|\begin{pmatrix} 1 & U \end{pmatrix} A \begin{pmatrix} T\\1 \end{pmatrix}\right|^2\\&=-\frac{\pi\,\tau_2}{T_2\,U_2}\left| T U \,\beta n_2+T \,\alpha n_1 -\gamma m_1\,U+\delta m_2\right|^2\quad \text{and}
\end{split}
\end{equation}
\begin{equation}
\label{59}
	\mathcal{S}:=-2\pi i \,\tau\, det{A}=2\pi i \,\tau\left(\gamma m_1 \,\alpha n_1+\delta m_2\,\beta n_2\right)\,.
\end{equation}
The first reduction we are going to perform becomes visible if we look at
\begin{equation}
\label{60}
	\mathcal{S}=2\pi i \,\tau\left(\alpha\gamma m_1 \, n_1+\beta\delta m_2\, n_2\right)\,.
\end{equation}
It suggests that one can reduce $\left(\alpha,\beta,\gamma,\delta\right)$ to $\left(1,1,\alpha\gamma,\beta\delta\right)$. To establish this, we have to show that $\mathcal{H}$ can be reduced consistently.

It is
\begin{equation}
\label{61}
	\mathcal{H}=-\frac{\pi\,\tau_2}{T_2\,U_2}\left| T U \,\beta n_2+T \,\alpha n_1 -\gamma m_1\,U+\delta m_2\right|^2=
\end{equation}
\[
	\quad\,=-\frac{\pi\,\tau_2}{\alpha \beta T_2\,\frac{\beta}{\alpha}U_2}\left| \alpha\beta\,T\,\frac{\beta}{\alpha} U \, n_2+\alpha\beta\,T \,n_1 -\alpha\gamma\, m_1\,\frac{\beta}{\alpha}U+\beta\delta\, m_2\right|^2\,.
\]
If we rescale the moduli as
\begin{align}
\label{62}
	&T\longmapsto T'=\alpha \beta\,T\\
\label{63}
	&U\longmapsto U'=\frac{\beta}{\alpha}\, U
\end{align}
we see that it is possible to consistently reduce $\left(\alpha,\beta,\gamma,\delta\right)$ to $\left(1,1,\tilde\gamma,\tilde\delta\right)$, where $\tilde\gamma=\alpha\gamma$ and $\tilde\delta=\beta\delta$.

On the other hand, if we perform the rescaling
\begin{align}
\label{65}
	&T\longmapsto T'=\frac{\alpha' \beta'}{\gamma'\delta'}\,T\quad\text{and}\\
\label{66}
	&U\longmapsto U'=\frac{\gamma'\beta'}{\alpha'\delta'}\, U\,,
\end{align}
we see that the more general reduction
\begin{equation}
\label{64}
	\left(\alpha'\alpha,\beta'\beta,\gamma'\gamma,\delta'\delta\right)\longmapsto \left(\gamma'\alpha,\delta'\beta,\alpha'\gamma,\beta'\delta\right)
\end{equation}
takes place.

The first reduction can be recovered from the second by setting $\gamma'=\delta'=\alpha=\beta=1$ and $\alpha'\,\mapsto\,\alpha$ as well as $\beta'\,\mapsto\,\beta$. 

Therefore, as we showed in section \ref{genset}, the Poisson resummed version of \eqref{55} has to posses the modular symmetry group \eqref{40} after appropriate rescaling of the moduli.

This is important since we integrate $\tau_2 Z_{\left(1,\theta^{l_k}\right)}^{\text{one-loop}}$ over a fundamental domain of this symmetry group.

However, we are still left with the case $(1,1,\tilde\gamma ,\tilde\delta )$ with $\tilde\gamma ,\tilde\delta\in\mathds{Q}$ - posing the mentioned difficulties when trying to sum over matrices with rational entries. To further reduce this case, let us examine how we can deal with common factors in $\tilde\gamma ,\tilde\delta\in\mathds{Q}$. Since $\tilde\gamma$ and $\tilde\delta$ are rational numbers, they can be written
\begin{equation}
\label{67}
	\tilde\gamma=\sumfrac{\tilde\gamma} \quad \text{and} \quad \tilde\delta=\sumfrac{\tilde\delta}\,\qquad\qquad\text{with}\;\gcd{(u_{\tilde\gamma},v_{\tilde\gamma})}=\gcd{(u_{\tilde\delta},v_{\tilde\delta})}=1.
\end{equation}
To determine the sought common factor, we modify \eqref{67} to
\begin{equation}
\label{68}
\begin{split}
	&\tilde\gamma=\frac{\gcd\left(v_{\tilde\delta}\,u_{\tilde\gamma},v_{\tilde\gamma}\,u_{\tilde\delta}\right)}{v_{\tilde\gamma}v_{\tilde\delta}}\,\frac{v_{\tilde\delta}\,u_{\tilde\gamma}}{\gcd\left(v_{\tilde\delta}\,u_{\tilde\gamma},v_{\tilde\gamma}\,u_{\tilde\delta}\right)}\quad \text{and}\\
	&\tilde\delta=\frac{\gcd\left(v_{\tilde\delta}\,u_{\tilde\gamma},v_{\tilde\gamma}\,u_{\tilde\delta}\right)}{v_{\tilde\gamma}v_{\tilde\delta}}\,\frac{v_{\tilde\gamma}\,u_{\tilde\delta}}{\gcd\left(v_{\tilde\delta}\,u_{\tilde\gamma},v_{\tilde\gamma}\,u_{\tilde\delta}\right)}\,.
\end{split}
\end{equation}
If we now define
\begin{equation}
\label{71}
	\lambda:=\frac{\gcd\left(v_{\tilde\delta}\,u_{\tilde\gamma},v_{\tilde\gamma}\,u_{\tilde\delta}\right)}{v_{\tilde\gamma}v_{\tilde\delta}}\,,\quad \overline\gamma:=\frac{v_{\tilde\delta}\,u_{\tilde\gamma}}{\gcd\left(v_{\tilde\delta}\,u_{\tilde\gamma},v_{\tilde\gamma}\,u_{\tilde\delta}\right)}\quad \text{and}\quad\overline\delta:=\frac{v_{\tilde\gamma}\,u_{\tilde\delta}}{\gcd\left(v_{\tilde\delta}\,u_{\tilde\gamma},v_{\tilde\gamma}\,u_{\tilde\delta}\right)}
\end{equation}
equation \eqref{68} yields
\begin{equation}
\label{72}
	\alpha\,\gamma=\tilde\gamma=\lambda\,\overline\gamma \quad \text{and} \quad \beta\,\delta=\tilde\delta=\lambda\,\overline\delta\,.
\end{equation}
By definition \eqref{67} it holds that $\gcd\left(\overline\gamma,\overline\delta\right)=1$. Furthermore, inspection of \eqref{71} yields $\overline\gamma,\overline\delta\in\mathds{Z}$.

Because of \eqref{72} it is clear that the symmetry of $\tau_2 Z_{\left(1,\theta^{l_k}\right)}^{\text{one-loop}}$ can also be expressed in terms of $\lambda$, $\overline\gamma$ and $\overline\delta$. If we write $\lambda$ as 
\begin{equation}
	\lambda=\sumfraca{\lambda}
\end{equation}
and look at a matrix multiplication similar to \eqref{38} we find that\footnote{Recall that $x\,y=\lcm\left(x,y\right)\,\gcd\left(x,y\right)$.} 
\begin{equation}
\label{88}
\begin{split}
	&\lcm\left(\frac{v_\lambda}{\gcd\left(v_\lambda,\overline\gamma\right)},\frac{v_\lambda}{\gcd\left(v_\lambda,\overline\delta\right)}\right)= \frac{v_\lambda}{\gcd\left(\gcd\left(v_\lambda,\overline\gamma\right),\gcd\left(v_\lambda,\overline\delta\right)\right)}=\\&\quad=v_\lambda=:\overline\mu\,|\,b \quad \text{and}
\end{split}
\end{equation}
\begin{equation}
\label{89}
\begin{split}
	&\lcm \left( \frac{u_\lambda \, \overline\gamma}{\gcd\left(u_\lambda\,\overline\gamma,v_\lambda\right)},\frac{u_\lambda\,\overline\delta}{\gcd\left(u_\lambda\,\overline\delta,v_\lambda\right)} \right)=u_\lambda\,\lcm\left( \frac{\overline\gamma}{\gcd\left(\overline\gamma,v_\lambda\right)},\frac{\overline\delta}{\gcd\left(\overline\delta,v_\lambda\right)} \right)=\\
	&\quad=u_\lambda\frac{\overline\gamma}{\gcd\left(\overline\gamma,v_\lambda\right)}\,\frac{\overline\delta}{\gcd\left(\overline\delta,v_\lambda\right)}=\frac{v_\lambda}{\gcd\left(v_\lambda,\overline\gamma\,\overline\delta\right)}\,\lambda\,\overline\gamma\,\overline\delta=:\epsilon\,\lambda\,\overline\gamma\,\overline\delta=:\overline\nu\,|\,c\,,
\end{split}
\end{equation}
where use has been made of $\gcd\left(\overline\gamma,\overline\delta\right)=\gcd\left(u_\lambda,v_\lambda\right)=1$, and by definition $\epsilon\in\mathds{Z}$. This means 
\begin{equation}
\label{90}
	\sumset{1}{1}{\lambda\overline\delta}{\lambda\overline\gamma}{\mathds{M}}\cdot P=\sumset{1}{1}{\lambda\overline\delta}{\lambda\overline\gamma}{\mathds{M}} \quad \text{for} \quad P\in\Gamma\left(\overline\mu,\overline\nu\right)\,.
\end{equation}
In order to get rid of the common factor $\lambda$, we will have to reduce
\begin{equation}
	\left(1,1,\lambda\,\overline\gamma,\lambda\,\overline\delta\right)\longmapsto \left(1,1,\lambda,\lambda\, \overline\gamma\, \overline\delta\right)\quad \text{if} \quad \gamma<\delta\quad \text{and}
\end{equation}
\begin{equation}
	\left(1,1,\lambda\,\overline\gamma,\lambda\,\overline\delta\right)\longmapsto \left(1,1,\lambda \,\overline\gamma\, \overline\delta,\lambda\right)\quad \text{if} \quad \gamma>\delta\,.
\end{equation}
This can be achieved by making use of the Smith normal form \cite{smithnf}. The theorem we will use is given by

\begin{thm}[Smith Normal Form]\nopagebreak
\label{SNF}
	Let $M\in\mathds{Z}^{2\times2}$. Then there exist invertible matrices $P\in\GL{2}{Z}$ and $Q\in\GL{2}{Z}$ such that 
	\begin{equation}
	\label{73}
		P\cdot M\cdot Q = \begin{pmatrix} \gamma_1 & 0 \\ 0 & \gamma_2 \end{pmatrix}\,,
	\end{equation}
	with $\gamma_1|\gamma_2$. The numbers $\gamma_1$ and $\gamma_2$ are the elementary divisors of $M$. 
\end{thm}

The fact that the matrices $P$ and $Q$ are invertible and integral valued is equivalent to their determinant being $1$ or $-1$. Theorem \ref{SNF} tells us that there exist matrices $P$ and $Q$ such that
\begin{equation}
\label{74}
	\begin{pmatrix} \overline\gamma & 0 \\ 0 & \overline\delta  \end{pmatrix}=P^{-1} \cdot \begin{pmatrix} 1 & 0 \\ 0 & \overline\gamma\overline\delta \end{pmatrix} \cdot Q^{-1}\,, 
\end{equation}
where we used $\gcd\left(\overline\gamma,\overline\delta\right)=1$. If $\overline\gamma<\overline\delta$ it follows that $\det P = \det Q = 1$ and $\overline\gamma>\overline\delta$ implies $\det P = \det Q = -1$. In the latter case we use
\begin{equation}
\label{75}
	\begin{pmatrix} \overline\gamma & 0 \\ 0 & \overline\delta  \end{pmatrix} =\left(RP\right)^{-1} \cdot \begin{pmatrix} \overline\gamma\overline\delta & 0 \\ 0 & 1 \end{pmatrix} \cdot \left(QR\right)^{-1}
\end{equation}
to replace \eqref{74}, where we used a matrix
\begin{equation}
\label{76}
	R:=\begin{pmatrix} 0 & 1 \\ 1 & 0  \end{pmatrix}\,.
\end{equation}
Since $\det R = -1$ it is evident that $\det RP=\det QR=1 $. Having reduced to the case $(1,1,\lambda\overline\gamma,\lambda\overline\delta)$, $\mathcal{H}$ reads
\begin{equation}
\label{77}
\begin{split}
	\mathcal{H}&=-\frac{\pi\,\tau_2}{T'_2\,U'_2} \left|\begin{pmatrix} 1 & U' \end{pmatrix}\cdot \begin{pmatrix} n_1 & \,\,\,\lambda\,\overline\delta\, m_2 \\ n_2 & -\lambda\,\overline\gamma\, m_1  \end{pmatrix}\cdot \begin{pmatrix} T'\\1 \end{pmatrix}\right|^2=\\
	 &=-\frac{\pi\,\tau_2}{T'_2\,U'_2} \left|\begin{pmatrix} 1 & U' \end{pmatrix}\cdot \begin{pmatrix} n_1 & \,\,\,\overline\delta\, m_2 \\ n_2 & -\overline\gamma\, m_1  \end{pmatrix}\cdot \begin{pmatrix} 1 & 0 \\ 0 & \lambda \end{pmatrix}\cdot \begin{pmatrix} T'\\1 \end{pmatrix}\right|^2=\\
	 &=-\frac{\pi\,\lambda\tau_2}{\left(T_2/\lambda\right)\,U'_2} \left|\begin{pmatrix} 1 & U' \end{pmatrix}\cdot \begin{pmatrix} n_1 & \,\,\,\overline\delta\, m_2 \\ n_2 & -\overline\gamma\, m_1  \end{pmatrix}\cdot \begin{pmatrix} T'/\lambda\\1 \end{pmatrix}\right|^2\,.
\end{split}
\end{equation}
Let w.l.o.g.\ $\overline\gamma<\overline\delta$. Then it holds that\footnote{Recall that $S:=\begin{pmatrix}0 & -1\\1 & 0\end{pmatrix}$.}
\begin{equation}
\label{78}
\begin{split}
	&\begin{pmatrix} n_1 & \,\,\,\overline\delta\, m_2 \\ n_2 & -\overline\gamma\, m_1  \end{pmatrix}=\begin{pmatrix} n_1 & 0 \\ n_2 & 0 \end{pmatrix} + \begin{pmatrix} 0 & 1 \\ -1 & 0 \end{pmatrix}\cdot \begin{pmatrix} \overline\gamma & 0 \\ 0 & \overline\delta \end{pmatrix}\cdot \begin{pmatrix} 0 & m_1 \\ 0 & m_2 \end{pmatrix}=\\
	&=\left(S^{-1} P S\right)^{-1}\left(\left(S^{-1}PS\right) \cdot \begin{pmatrix} n_1 & 0 \\ n_2 & 0 \end{pmatrix}+ S^{-1} \cdot \begin{pmatrix} 1 & 0 \\ 0 & \overline\gamma\overline\delta \end{pmatrix} \cdot Q^{-1}\cdot \begin{pmatrix} 0 & m_1 \\ 0 & m_2 \end{pmatrix} \right)\;.
\end{split}
\end{equation}
If we define
\begin{align}
\label{79}
 &\begin{pmatrix} n'_1 \\ n'_2 \end{pmatrix}:=\left(S^{-1}PS\right) \cdot \begin{pmatrix} n_1 \\ n_2 \end{pmatrix}\,,\\
 &\begin{pmatrix} m'_1 \\ m'_2 \end{pmatrix}:=Q^{-1}\cdot \begin{pmatrix} m_1 \\ m_2 \end{pmatrix}
\end{align}
and\footnote{Here $\left(\left(S^{-1} P S\right)^{-1}\right)^{\sharp}$ acts as a modular transformation on $U$. For the definition of $\sharp$ cf. \eqref{wilsonpart6}.}
\begin{equation}
 U'':=\left(\left(S^{-1} P S\right)^{-1}\right)^{\sharp} U'\,,
\end{equation}
then the combination of \eqref{59}, \eqref{77} and \eqref{78} yields
\begin{equation}
\label{80}
 \mathcal{S}=-2\pi i \tau \det A = 2\pi i \tau\left( \lambda m'_1 n'_1+ \lambda\,\overline\gamma\,\overline\delta\, m'_2 n'_2\right)
\end{equation}
and
\begin{equation}
\label{81}
 \mathcal{H}=-\frac{\pi\,\tau_2}{T'_2\,U''_2} \left|\begin{pmatrix} 1 & U'' \end{pmatrix}\cdot \begin{pmatrix} n'_1 & \,\,\lambda\overline\gamma\,\overline\delta\, m'_2 \\ n'_2 & - \lambda\, m'_1  \end{pmatrix}\cdot \begin{pmatrix} T''\\1 \end{pmatrix}\right|^2\,.
\end{equation}
If $\overline\gamma>\overline\delta$, we replace $P$ by $RP$ and $Q$ by $QR$
and follow the same line of reasoning.

But now another question arises. What is the summation domain of $n'_1$, $n'_2$, $m'_1$ and $m'_2$? To answer this question let us parameterise $S^{-1}PS$ and $Q^{-1}$ as 
\begin{equation}
\label{82}
 S^{-1}PS=\begin{pmatrix} a_{11} & a_{12} \\ a_{21} & a_{22} \end{pmatrix}
\end{equation}
and
\begin{equation}
\label{83}
 Q^{-1}= \begin{pmatrix} b_{11} & b_{12} \\ b_{21} & b_{22} \end{pmatrix}\,.
\end{equation}
Then the primed variables can be written as 
\begin{equation}
\label{84}
 n'_i=a_{i1}\,n_1+a_{i2}\,n_2\quad \text{and}
\end{equation}
\begin{equation}
\label{85}
 m'_i=b_{i1}\,n_1+b_{i2}\,n_2\,.
\end{equation}
Since $\left(S^{-1}PS\right)\in\SL{2}{Z}$ and $Q^{-1}\in\SL{2}{Z}$
(see comment below theorem \ref{SNF}), 
it follows that $\gcd\left(a_{i1},a_{i2}\right)=\gcd\left(b_{i1},b_{i2}\right)=1$ for $i=1,2$\footnote{This holds true, because the Diophantine equations $\det{\left(S^{-1}PS\right)}=1$ and $\det{Q^{-1}}=1$ posses solutions.}. Since we had to sum the unprimed variables over $\mathds{Z}$, it follows that we have to sum the primed variables over $\mathds{Z}$, too. Again we have to ensure that we did not change the symmetry of $\tau_2 Z_{\left(1,\theta^{l_k}\right)}^{\text{one-loop}}$. Therefore, we have to examine whether $\sumset{1}{1}{\lambda\overline\gamma\overline\delta}{\lambda}{\mathds{M}}\cdot P=\sumset{1}{1}{\lambda\overline\gamma\overline\delta}{\lambda}{\mathds{M}}$ for $P\in\Gamma\left(\overline\mu,\overline\nu\right)$ (cf. \eqref{88}, \eqref{89} and \eqref{90}). This means for all $M_1\in\sumset{1}{1}{\lambda\overline\gamma\overline\delta}{\lambda}{\mathds{M}}$ and $P\in\Gamma\left(\overline\mu,\overline\nu\right)$ there has to exist a matrix $M_2\in\sumset{1}{1}{\lambda\overline\gamma\overline\delta}{\lambda}{\mathds{M}}$ such that $M_1\cdot P = M_2$. Let us look at such a matrix multiplication:
\begin{equation}
\label{91}
 \begin{pmatrix}
  n_1 & \frac{1}{\lambda}\,l_1\\
  n_2 & \frac{1}{\lambda\,\overline{\gamma}\,\overline\delta}\,l_2
 \end{pmatrix} \cdot
 \begin{pmatrix}
  a & b \\
  c & d
 \end{pmatrix}=
 \begin{pmatrix}
  n_1 \, a + l_1 \, \frac{1}{\lambda} \, c & \frac{1}{\lambda} \left(n_1 \, \lambda \, b + l_1 \, d\right)\\
  n_2 \, a + l_2 \, \frac{1}{\lambda \, \overline\gamma \, \overline{\delta}} \, c & \frac{1}{\lambda \, \overline\gamma \, \overline\delta} \left(n_1 \, \lambda \, \overline\gamma \, \overline\delta \, b + l_1 \, d\right)
 \end{pmatrix}
\end{equation}
Here we can read off
\begin{equation}
\label{92}
 v_\lambda\,| \,b \quad \text{and}
\end{equation}
\begin{equation}
\label{93}
 \lcm\left(u_\lambda,\frac{u_\lambda\,\overline\gamma\,\overline\delta}{\gcd\left(\overline\gamma\,\overline\delta\right)}\right)=u_\lambda\frac{\overline\gamma\,\overline\delta}{\gcd\left(\overline\gamma\,\overline\delta\right)}=u_\lambda\,\frac{\overline\gamma}{\gcd\left(\overline\gamma,v_\lambda\right)}\,\frac{\overline\delta}{\gcd\left(\overline\delta,v_\lambda\right)}\,|\,c\,,
\end{equation}
where we used $\lcm\left(z\,x,z\,y\right)=z\lcm\left(x,y\right)$ and $\gcd\left(x\,y,z\right)=\gcd\left(x,z\right)\gcd\left(y,z\right)$ for

$\gcd\left(x,y\right)=1$ and $x,y,z>0$. Since \eqref{92} and \eqref{93} coincides with \eqref{88} and \eqref{89}, the symmetry does not change.

Now we are ready for the last reduction. Let again w.l.o.g $\overline\gamma<\overline\delta$. We will show that
\begin{equation}
\label{69}
 \left(1,1,\lambda,\lambda\,\overline\gamma\,\overline\delta\right)\longmapsto \left(1,1,1,\overline\gamma\,\overline\delta\right)
\end{equation}
if we rescale $T'$ as
\begin{equation}
\label{70}
 T'\longmapsto T''=\frac{T'}{\lambda}\,.
\end{equation}
For that purpose let us look at the partition function which is associated to $\left(1,\theta^{l_k}\right)$. Up to now we have shown that it can be written as
\begin{equation}
\label{86}
 \tau_2 Z_{\left(1,\theta^{l_k}\right)}^{\text{one-loop}}\left(\tau\right)=\sum_{\sumset{1}{1}{\lambda\overline\gamma\overline\delta}{\lambda}{\mathds{M}}}\ex{-2\pi i \,T'\det{A}}\,\frac{T'_2}{\lambda\overline\gamma \lambda\overline\delta}\,\exp\left[-\frac{\pi\,T'_2}{\tau_2\,U''_2}\left| \begin{pmatrix} 1 & U'' \end{pmatrix} A \begin{pmatrix} \tau\\1 \end{pmatrix}\right|^2\right]\,.
\end{equation}
This is equivalent to
\begin{equation}
\label{87}
 \tau_2 Z_{\left(1,\theta^{l_k}\right)}^{\text{one-loop}}\left(\tau\right)=\frac{1}{\lambda}\sum_{A\in\sumset{1}{1}{\overline\gamma\overline\delta}{1}{\mathds{M}}}\ex{-2\pi i \,\left(T'/\lambda\right)\det{A}}\,\frac{T'_2/\lambda}{\overline\gamma\overline\delta}\,\exp\left[-\frac{\pi\,\left(T'_2/\lambda\right)}{\lambda\tau_2\,U''_2}\left| \begin{pmatrix} 1 & U'' \end{pmatrix} A \begin{pmatrix} \lambda\tau\\1 \end{pmatrix}\right|^2\right]
\end{equation}
and nearly what we wanted to achieve. But how can we deal with $\lambda\tau$? First we have to observe that $A$ is of the form
\begin{equation}
\label{94}
 A=\begin{pmatrix} n_1 & l_1 \\ n_2 & \frac{1}{\overline\gamma\,\overline\delta} \end{pmatrix}\,.
\end{equation}
In the last section we gave a procedure to compute a sum over these matrices assuming that the integration domain is $R_{\Gamma_0\left(\overline\gamma\overline\delta\right)}$. One crucial point was the possibility to reinterpret a matrix multiplication of representative matrices with matrices in $\Gamma^0\left(\overline\gamma\overline\delta\right)$ as a modular transformation of $\tau$ by an element in $\Gamma_0\left(\overline\gamma\,\overline\delta\right)$. We only used two universal properties. Firstly, a fundamental domain of $\Gamma_0\left(\overline\gamma\,\overline\delta\right)$ is defined as a maximal inequivalent set of complex numbers $\tau\in\mathds{H}^+$ and by acting with $\Gamma_0\left(\overline\gamma\,\overline\delta\right)$ on $R_{\Gamma_0\left(\overline\gamma\,\overline\delta\right)}$ we get the whole complex plane. Secondly, the contributions of representative matrices with non-vanishing determinant have to be integrated over $\mathds{H}^+$, while those with vanishing determinant over $\mathds{H}^+/\langle T \rangle$. The latter resulted from the fact that two matrices $P_1,P_2\in\Gamma_0\left(\overline\gamma\,\overline\delta\right)$ lead to the same matrix with vanishing determinant, $A_0\cdot P_1=A_0\cdot P_2$, if these two matrices are connected by an element of $\langle T \rangle$. Now, we will show that the same holds true here (up to a multiplicative constant). The crucial point is that we act with a matrix $P\in\Gamma_0\left(\overline\gamma\,\overline\delta\right)$ not on $\tau$, but on $\lambda\tau$ as a modular transformation. This means
\begin{equation}
\label{95}
 P_{\lambda\tau} \lambda\tau=\frac{a\,\lambda\tau+b}{\overline\gamma\,\overline\delta\,c\,\lambda\tau+d}=\lambda\,\frac{a\tau+\frac{b}{\lambda}}{\lambda\,\overline\gamma\,\overline\delta\,\tau+d}=\lambda\, P'_\tau \, \tau\,,
\end{equation}
which shows that we can alternatively act with a transformation $P'\in\Gamma\left(\frac{1}{\lambda},\lambda\,\overline\gamma\,\overline\delta\right)$ on $\tau$ and rescale afterwards by $\lambda$. The next problem which arises is the fact that we integrate the partition function over a fundamental domain of $\Gamma\left(v_\lambda,\epsilon\,\lambda\,\overline\gamma\,\overline\delta\right)$ and not over a fundamental domain of $\Gamma\left(\frac{1}{\lambda},\lambda\,\overline\gamma\,\overline\delta\right)$. This can be resolved by observing
\begin{equation}
\label{96}
 \Gamma\left(v_\lambda,\epsilon\,\lambda\,\overline\gamma\,\overline\delta\right)\subset\Gamma\left(\frac{1}{\lambda},\lambda\,\overline\gamma\,\overline\delta\right)
\end{equation}
and, therefore,
\begin{equation}
\label{97}
 R_{\Gamma\left(v_\lambda,\epsilon\,\lambda\,\overline\gamma\,\overline\delta\right)}=\bigcup_{k=1}^{\left[\Gamma\left(\frac{1}{\lambda},\lambda\,\overline\gamma\,\overline\delta\right):\Gamma\left(v_\lambda,\epsilon\,\lambda\,\overline\gamma\,\overline\delta\right)\right]}\,M_k\,R_{\Gamma\left(\frac{1}{\lambda},\lambda\,\overline\gamma\,\overline\delta\right)}\,,
\end{equation}
with $M_k\in\Gamma\left(\frac{1}{\lambda},\lambda\,\overline\gamma\,\overline\delta\right)$. Thus, the relevant integral to compute one-loop gauge threshold corrections reads
\begin{equation}
\label{98}
 \sumset{\alpha}{\beta}{\delta}{\gamma}{I}(T,U)=\sumset{\alpha}{\beta}{\delta}{\gamma}{A}\,\frac{1}{\lambda}\,\int_{R_{\Gamma\left(v_\lambda,\epsilon\,\lambda\,\overline\gamma\,\overline\delta\right)}}\hypmes{\tau}\,\sum_{A\in\sumset{1}{1}{\overline\gamma\,\overline\delta}{1}{\mathds{M}}}\ex{-2\pi i \,T'/\lambda\det{A}}\,\frac{T'_2/\lambda}{\overline\gamma\, \overline\delta}
\end{equation}
\[
 \qquad\qquad\times\,\exp\left[-\frac{\pi\,T'_2/\lambda}{\lambda\tau_2\,U''_2}\left| \begin{pmatrix} 1 & U'' \end{pmatrix} A \begin{pmatrix} \lambda\tau\\1 \end{pmatrix}\right|^2\right]-\int_{R_\Gamma}\hypmes{\tau} \tau_2=
\]
\[
 \qquad =\sumset{\alpha}{\beta}{\delta}{\gamma}{A}\,\frac{1}{\lambda}\,
 \sum_{k=1}^{\left[\Gamma\left(\frac{1}{\lambda},\lambda \, \overline\gamma \, \overline\delta\right) : \Gamma\left( v_\lambda,\epsilon\,\lambda\,\overline\gamma\,\overline\delta\right)\right]}
 \int_{M_k\,R_{\Gamma\left(\frac{1}{\lambda},\lambda\,\overline\gamma\,\overline\delta\right)}}\hypmes{\tau} \, 
 \sum_{A\in\sumset{1}{1}{\overline\gamma\,\overline\delta}{1}{\mathds{M}}}\ex{-2\pi i \,T'/\lambda\det{A}}\,\frac{T'_2/\lambda}{\overline\gamma\, \overline\delta}
\]
\[
 \quad\qquad\times\,\exp\left[-\frac{\pi\,T'_2/\lambda}{\lambda\tau_2\,U''_2}\left| \begin{pmatrix} 1 & U'' \end{pmatrix} A \begin{pmatrix} \lambda\tau\\1 \end{pmatrix}\right|^2\right]-\int_{R_\Gamma}\hypmes{\tau} \tau_2\,.
\]
Applying the procedure of the last section on this integral results in considering the action of subgroups of $\Gamma\left(\frac{1}{\lambda},\lambda\,\overline\gamma\,\overline\delta\right)$ on $M_k R_{\Gamma\left(\frac{1}{\lambda},\lambda\,\overline\gamma\,\overline\delta\right)}$. By observing that the latter is a fundamental domain of $\Gamma\left(\frac{1}{\lambda},\lambda\,\overline\gamma\,\overline\delta\right)$ for all $k$ it follows that the contributions of the representative matrices with non-vanishing determinant have to be integrated over $\mathds{H}^+$. Moreover, it is evident that two matrices of $\Gamma_0\left(\overline\gamma\,\overline\delta\right)$ lead to the same matrix with vanishing determinant if these two matrices are connected by an element of $\langle T^{1/\lambda} \rangle$.

Above we argued that we have to rescale by $\lambda$ after a modular transformation $P'$ on $\tau$. Since the integration measure $\hypmes{\tau}$ is invariant under scaling and modular transformations of $\tau$, we have to integrate those contributions of matrices with non-vanishing determinant over $\lambda\,\mathds{H}^+=\mathds{H}^+$ and those with vanishing determinant over $\lambda\,(\mathds{H}^+/\langle T^{1/\lambda} \rangle)=\mathds{H}^+/\langle T \rangle$. Here we mean rescaled domains via a transformation $\tau\longmapsto\tau'=\lambda\tau$. As an example consider a rescaled open interval: $\lambda=2$ and $2\,]-1,2[=]-2,4[$.

Therefore, we have shown that the universal properties, mentioned below equation \eqref{94}, are fulfilled. Hence, we gain
\begin{equation}
\label{99}
\begin{split}
\sumset{\alpha}{\beta}{\delta}{\gamma}{I}=&\sumset{\alpha}{\beta}{\delta}{\gamma}{A}\,\frac{\left[\Gamma\left(\frac{1}{\lambda},\lambda \, \overline\gamma \, \overline\delta\right) : \Gamma\left( v_\lambda,\epsilon\,\lambda\,\overline\gamma\,\overline\delta\right)\right]}{\lambda}\,\int_{R_{\Gamma_0\left(\overline\gamma\overline\delta\right)}}\hypmes{\tau}\,\sum_{A\in\sumset{1}{1}{\overline\gamma\,\overline\delta}{1}{\mathds{M}}}\ex{-2\pi i \,T'/\lambda\det{A}}\\
  &\times\,\frac{T'_2/\lambda}{\overline\gamma\, \overline\delta}\,\exp\left[-\frac{\pi\,T'_2/\lambda}{\tau_2\,U''_2}\left| \begin{pmatrix} 1 & U'' \end{pmatrix} A \begin{pmatrix} \tau\\1 \end{pmatrix}\right|^2\right]-\int_{R_\Gamma}\hypmes{\tau} \tau_2
\end{split}
\end{equation}
Moreover, from the construction given in the last section it follows that an overall factor in front of the first integral gets absorbed in $\sumset{\alpha}{\beta}{\gamma}{\delta}{A}$. This is true because of the finiteness of the result. Thus, we have shown that it is possible to reduce \eqref{69} via \eqref{70}. 

For the case $\overline\gamma>\overline\delta$ it follows in complete analogy
\begin{equation}
\label{100}
 \left(1,1,\lambda\,\overline\gamma\,\overline\delta,\lambda\right)\longmapsto \left(1,1,\overline\gamma\,\overline\delta,1\right)
\end{equation}
if we rescale $T'$ as
\begin{equation}
\label{101}
 T'\longmapsto T''=\frac{T'}{\lambda}\,.
\end{equation}

This completes our treatment of the reduction of all cases $(\alpha,\beta,\gamma,\delta)$ to those of the form $(1,1,1,\delta)$.

We have shown that this reduction is always possible and we gave a procedure to achieve this. We started with $(\alpha,\beta,\gamma,\delta)$ and showed that this is equivalent to $(1,1,\tilde\gamma,\tilde\delta)=(1,1,\alpha\gamma,\beta\delta)$ via \eqref{62}, \eqref{63}. By definition it is $(1,1,\tilde\gamma,\tilde\delta)=(1,1,\lambda\bar\gamma,\lambda\bar\delta)$ with $\lambda\in\mathds{Q}$ and $\bar\gamma,\bar\delta\in\mathds{Z}$. Making use of the SNF, we showed that we can transform $U$ by a modular transformation $P$ to $U'=(P^{T})^\sharp\, U$ so that $(1,1,\lambda,\lambda\bar\gamma\bar\delta)$ (w.l.o.g.\ ). We got rid of the factor $\lambda$ by the rescaling $T\mapsto \frac{T}{\lambda}$.

Using
\begin{equation}
\sumset{1}{1}{\bar\gamma}{\bar\delta}{C}
=
\sumset{1}{1}{1}{\bar\delta}{C}\,\sumset{1}{1}{\bar\gamma}{1}{C}
=
\sumset{1}{1}{1}{\bar\delta}{C}\,\sumset{1}{1}{1}{\bar\gamma}{C}
\;,
\end{equation}

we obtain,
\begin{align}
	\nonumber\sumset{\alpha}{\beta}{\delta}{\gamma}{I}=&-\sumset{1}{1}{\bar \delta \bar \gamma}{1}{A}\,\sum_{d|\bar\delta\wedge g | \bar\gamma}\sumset{1}{1}{\bar \delta}{1}{C}(d) \sumset{1}{1}{\bar \gamma}{1}{C}(g) \\ \label{result}&\times
\left[
\ln\left(\frac{T''_2}{g d}\left|\eta\left(\frac{T''}{g d}\right)\right|^4\,\frac{ U''_2}{g d}\left|\eta\left(\frac{U''}{g d}\right)\right|^4\right)
+
\ln\left(\frac{8\,\pi\,\ex{1-\gamma_E}}{3\sqrt{3}}\right)\right]\,,
\end{align}
where
\begin{align}
&T''=\frac{\alpha\beta}{\lambda} T\,,\\
&U''=\frac{\beta}{\alpha} U'\,.
\end{align}

This leads to a symmetry group
\begin{equation}\label{resultsymmetry}
\sumset{\alpha}{\beta}{\delta}{\gamma}{\mathfrak{S}}
=
\left[
\;
\left(
\Gamma\left(\frac{\alpha\beta}{\lambda},\bar\gamma\bar\delta\frac{\lambda}{\alpha\beta}\right)
\ast
\sumset{\alpha}{\beta}{\delta}{\gamma}{\mathfrak{T}}
\right)_T
\times
\left(
\Gamma\left(\frac{\beta}{\alpha},\bar\gamma\bar\delta\frac{\alpha}{\beta}\right)
\ast
\sumset{\alpha}{\beta}{\delta}{\gamma}{\mathfrak{U}}
\right)_{U'}
\;
\right]
\ast
\sumset{\alpha}{\beta}{\delta}{\gamma}{\mathfrak{M}}\,,
\end{equation}
with
\begin{align}
	&\sumset{\alpha}{\beta}{\delta}{\gamma}{\mathfrak{T}}:T\mapsto T'=-\frac{\gamma\delta\lambda^2}{\alpha^2 \beta^2 \,T}\,,\\
	&\sumset{\alpha}{\beta}{\delta}{\gamma}{\mathfrak{U}}:U'\mapsto -\frac{\gamma\delta\alpha^2}{\beta^2\,U'}\,,\\
	&\sumset{\alpha}{\beta}{\delta}{\gamma}{\mathfrak{M}}:(T,U)\mapsto(T',U')=\left(\frac{\lambda}{\alpha^2}U',\frac{\alpha^2}{\lambda} T\right)\,,
\end{align}
$\ast$ denotes the free product of groups.

As mentioned before, the three involutive symmetries $\sumset{\alpha}{\beta}{\delta}{\gamma}{\mathfrak{T}}$, $\sumset{\alpha}{\beta}{\delta}{\gamma}{\mathfrak{U}}$ and $\sumset{\alpha}{\beta}{\delta}{\gamma}{\mathfrak{M}}$ are not modular transformations in general. $\sumset{\alpha}{\beta}{\delta}{\gamma}{\mathfrak{M}}$ corresponds to the mirror map acting on the fixed plane, while the other two correspond to a generalisation of what is usually called T-duality, in the sense that they exchange large with small radii. However, notice that---unlike the usual interpretation of T-duality---these symmetries are not contained in $\PSL{2}{Z}$ and in none of its subgroups. They pose an additional structure to the modular transformations (which form a subgroup of the modular group in general). 

Physical consequences of this observations (e.g.\ self-dual points different from $1$) in concrete models are being investigated \cite{PalKlapUnpub}.

As a remark, this can be regarded as the proof of a conjecture made in \cite{ErlerTdual}: that there always exists an involutive symmetry interchanging large and small radii, in any orbifold model (even with non-vanishing Wilson lines). Until now, its existence could only be shown for simple toy-models.

\section{Conclusions and Outlook}

Our goal in this work has been the calculation of threshold corrections in general abelian toroidal orbifold models, allowing for arbitrary discrete Wilson lines. So far, only threshold corrections in the absence of discrete Wilson lines were known. However, the phenomenologically most promising models possess non-vanishing discrete Wilson lines.

The path followed in our work can be divided in two parts. The first part consists of chapters 2 and 3, in which we aim to reformulate the task at hand in terms of a well-defined technical problem. We were able to show that every orbifold model can be assigned four characteristic numbers\footnote{Note, that we compute the integrals in $\eqref{master}$ for all \emph{rational} numbers $\alpha,\beta,\gamma,\delta$. However, in physical models, the most general case is $(1,1,\gamma',\delta')$ with $\gamma',\delta'\in \mathds{Q}$, which is equivalent to $(\alpha,\beta,\gamma,\delta)\in\mathds{Z}^4$.} $(\alpha,\beta,\gamma,\delta)\in\mathds{Z}^4$ (one set for every fixed plane) which determine a special integral 

\begin{equation}
\begin{split}
	\sumset{\alpha}{\beta}{\delta}{\gamma}{I}(T,U)=\sumset{\alpha}{\beta}{\delta}{\gamma}{A}\,&\int_{R_{\Gamma'}}\hypmes{\tau}\,\sum_{A\in\sumset{\alpha}{\beta}{\delta}{\gamma}{\mathds{M}}}\ex{-2\pi \I \,T\det{A}}\,\frac{T_2}{\gamma \delta}\,\\&\times\exp\left[-\frac{\pi\,T_2}{\tau_2\,U_2}\left| \begin{pmatrix} 1 & U \end{pmatrix} A \begin{pmatrix} \tau\\1 \end{pmatrix}\right|^2\right]-\int_{R_\Gamma}\hypmes{\tau} \tau_2\,,
\end{split}
\tag{\ref{master}}
\end{equation}

(depending on the moduli of the fixed plane). Knowledge of this integral (for every fixed plane) together with the beta function coefficients is enough to calculate the threshold corrections $\Delta_a$.

The second part of our work is devoted to solving these integrals $\sumset{\alpha}{\beta}{\delta}{\gamma}{I}(T,U)$. This problem turns out to be quite difficult, mainly for two reasons: the domain of integration is the fundamental domain of some sub-group of $\PSL{2}{Z}$ and the integrand contains an infinite sum over all matrices which fulfil certain divisibility conditions. The former is difficult to construct in general and it is hard to find a parameterisation for the latter.

Fortunately, it suffices to solve the case $(1,1,1,\delta)$ with $\delta\in\mathds{Z}$, since all other cases can be mapped onto this one by fractional linear transformations of the fixed plane moduli. Still, both problems survive in a less complicated form.

The first one can be circumvented by only using defining properties of a fundamental domain. To tackle the second problem, we had to develop techniques to cope with the divisibility condition in the infinite sum of matrices. The tricks which emanate from this (cf.\ lemmas \ref{red}, \ref{choice}, \ref{novanish} and \ref{vanish}) were unfamiliar to us before and we were not able to find any similar techniques in the literature. Though, we think that it is most improbable that we were the first to ever use such techniques and it would be interesting to find works using them (or similar versions thereof).

Eventually, we obtained the result
\begin{align}
	\nonumber\sumset{\alpha}{\beta}{\delta}{\gamma}{I}=&-\sumset{1}{1}{\bar \delta \bar \gamma}{1}{A}\,\sum_{d|\bar\delta\wedge g | \bar\gamma}\sumset{1}{1}{\bar \delta}{1}{C}(d) \sumset{1}{1}{\bar \gamma}{1}{C}(g) \\ &\times
\left[
\ln\left(\frac{T''_2}{g d}\left|\eta\left(\frac{T''}{g d}\right)\right|^4\,\frac{ U''_2}{g d}\left|\eta\left(\frac{U''}{g d}\right)\right|^4\right)
+
\ln\left(\frac{8\,\pi\,\ex{1-\gamma_E}}{3\sqrt{3}}\right)\right]\;,
\tag{\ref{result}}
\end{align}
with $T''=\frac{\alpha\beta}{\lambda} T$, $U''=\frac{\beta}{\alpha} U'$.

It is particularly interesting how naturally number theoretic notions like prime numbers, greatest common divisor, lowest common multiple, etc.\ appear in our results and proofs. We found this quite surprising and suspect that there are good reasons for this beyond our present understanding of the problem. Therefore, we would be especially interested in understanding what actually was computed in our work from a mathematical point of view. It is known that one-loop string thresholds have a close relationship to the Ray-Singer/analytic torsion \cite{RaySinger}. The result for the case $(1,1,1,1)$ agrees with (the logarithm of) the result of Ray-Singer for a (complex) line-bundle with flat connection and, hence, could be viewed as the analytic torsion of the fixed plane. However, the interpretation of the result for general $(\alpha,\beta,\gamma,\delta)$ remains unclear. Especially the coefficients $\sumset{1}{1}{\delta}{1}{C}(d)$, $\sumset{\alpha}{\beta}{\delta}{\gamma}{A}$ appearing in our results should have some close relationship to bundle cohomology.

Another interesting aspect of our results are the modular symmetries. We obtain the symmetry group
\begin{equation}
\sumset{\alpha}{\beta}{\delta}{\gamma}{\mathfrak{S}}
=
\left[
\;
\left(
\Gamma\left(\frac{\alpha\beta}{\lambda},\bar\gamma\bar\delta\frac{\lambda}{\alpha\beta}\right)
\ast
\sumset{\alpha}{\beta}{\delta}{\gamma}{\mathfrak{T}}
\right)_T
\times
\left(
\Gamma\left(\frac{\beta}{\alpha},\bar\gamma\bar\delta\frac{\alpha}{\beta}\right)
\ast
\sumset{\alpha}{\beta}{\delta}{\gamma}{\mathfrak{U}}
\right)_{U'}
\;
\right]
\ast
\sumset{\alpha}{\beta}{\delta}{\gamma}{\mathfrak{M}}\,,
\tag{\ref{resultsymmetry}}
\end{equation}
where $\ast$ denotes the free product of groups.

Besides the expected occurrence of modular symmetries and the mirror map $\sumset{\alpha}{\beta}{\delta}{\gamma}{\mathfrak{M}}$, there appear two involutive symmetries $\sumset{\alpha}{\beta}{\delta}{\gamma}{\mathfrak{T}}$ and $\sumset{\alpha}{\beta}{\delta}{\gamma}{\mathfrak{U}}$. These exchange small with large radii and, in that sense, are a generalisation of T-duality. It should be stressed that, in contrast to the usual version of T-duality (which is a transformation in $\PSL{2}{Z}$), these are \emph{not} modular symmetries. Work on physical implications of these observations in various models is in progress \cite{PalKlapUnpub}.

In that context, we would also like to point out some side results of our work. In chapter 3 we analysed the momentum and winding lattices of the $\mathcal{N}=2$ sector. We were able to show how discrete Wilson lines effect these lattices and how to parametrise them. This was needed in the context of our work, to be able to characterise every obrifold model (in our sense) by four numbers $\alpha,\beta,\gamma,\delta$. However, we think that these results are of more general use in the context of orbifold model building. They should enable one to determine the spectrum of the $\mathcal{N}=2$ sector of a general orbifold model, i.e.\ with non-factorizable lattice and discrete Wilson lines and might also give hints on how to repeat this analysis for the other sectors of boundary conditions. To the extent of our knowledge, this is still an open problem and we think that our results might be applicable to solve this problem.

Finally, we would like to state that we think that it should be possible to generalise the given procedure in order to be able to calculate even more general integrals than $\sumset{\alpha}{\beta}{\delta}{\gamma}{I}(T,U)$. Necessary for our method to work in principle are the following properties: the domain of integration should be the fundamental domain of some group, this group should be a subgroup of the symmetry group of the integrand, the infinite sum over matrices should be now an (infinite) sum over group elements and one has to know a basic building block to trace everything back to. This might make the method applicable, in principle, to other technical problems.
\section*{Acknowledgement}
We thank Andre Lukas for helpful comments.

C.P.\ thanks Martin Schottenloher and the mathematical institute of the University of Munich (LMU) for hospitality and support during this work.

M.K.\ thanks Michael Ratz and the physics department of the Technical University Munich for hospitality and support during parts of this work. M.K.\ was supported by the Graduiertenkolleg GRK 1054 of the German Research Foundation (DFG) and the Lamb \& Flag scholarship of St John's College Oxford.

%------------------------------------------------------------------------------------------------------
%                                               Bibliography
%------------------------------------------------------------------------------------------------------

\bibliographystyle{unsrt}
\bibliography{mathpaper}

\end{document}